\documentclass[11pt]{article}%
\usepackage[letterpaper, top=1in, bottom=1in, left=1in, right=1in]{geometry}
\usepackage{amsfonts, amsmath, amsthm, amssymb}
\usepackage{comment}
\usepackage{enumitem, graphicx, stmaryrd, color, bbm}
\usepackage{url, algorithm}
\usepackage{float}
\usepackage{subcaption}
\floatstyle{ruled}

\usepackage[driverfallback=hypertex,pagebackref=true,colorlinks]{hyperref}
	\hypersetup{linkcolor=[rgb]{.7,0,.7}}
	\hypersetup{citecolor=[rgb]{.5,0,.5}}
	\hypersetup{urlcolor=[rgb]{.7,0,.7}}

\usepackage[capitalize]{cleveref}
\newtheorem{theorem}{Theorem}[section]
\newtheorem{lemma}[theorem]{Lemma}
\newtheorem{claim}[theorem]{Claim}
\newtheorem{definition}[theorem]{Definition}
\newtheorem{observation}[theorem]{Observation}

\newtheorem*{remark}{Remark}					

\newtheorem{proposition}[theorem]{Proposition}
\newtheorem{example}[theorem]{Example}

\newtheorem{fact}[theorem]{Fact}
\Crefname{theorem}{Theorem}{Theorems}
\Crefname{proposition}{Proposition}{Propositions}
\Crefname{claim}{Claim}{Claims}
\Crefname{lemma}{Lemma}{Lemmas}

\renewcommand{\implies}{\Longrightarrow}		
\newcommand{\abs}[1]{\left|#1\right|}		
\newcommand{\st}{\,|\,} 				 		
\newcommand{\E}{\mathop{\mathbb{E}}}  		
\newcommand{\V}{\mathop{\mathbb{V}}}  		
\newcommand{\N}{\mathbb{N}} 			  		
\newcommand{\ind}{\text{I}}					

\newcommand{\mc}{\mathcal}

\newcommand{\eps}{\mathop{\epsilon}}
\newcommand{\set}[1]{\left\{ #1 \right\}}   
\newcommand{\brac}[1]{\left( #1 \right)}    
\newcommand{\sqbrac}[1]{\left[ #1 \right]}  

\newcommand{\Prob}[1]{\Pr\sqbrac{#1}}		

\newcommand{\norm}[1]{\left\vert #1 \right\vert}			

\newcommand{\relent}[2]{D(#1\ \Vert\ #2)}					
\newcommand{\lonorm}[2]{\left\Vert #1 - #2 \right\Vert_1}

\newcommand{\ts}{\textsuperscript}		

\newcommand{\val}{\textup{val}} 	 			

\newcommand{\cC}{\mathcal{C}}
\newcommand{\cG}{\mathcal{G}}
\newcommand{\cP}{\mathcal{P}}

\newcommand{\Bernoulli}{\mathrm{Bernoulli}}
\newcommand{\textdef}[1]{\textnormal{\textsf{#1}}}
\newcommand{\eqdef}{\stackrel{\mathsf{def}}{=}}
\newcommand{\poly}{\mathsf{poly}}

\newif\ifnotes
\notestrue

\begin{document}
\title{Parallel Repetition For All 3-Player Games Over Binary Alphabet}

\author{Uma Girish\thanks{\scriptsize Princeton University. E-mail: \href{ugirish@cs.princeton.edu}{\texttt{ugirish@cs.princeton.edu}.} Research supported by the Simons Collaboration on Algorithms and Geometry, by a Simons Investigator Award, by the National Science Foundation grants No. CCF-1714779, CCF-2007462 and by the IBM Phd Fellowship.} \and Justin Holmgren\thanks{\scriptsize NTT Research.  E-mail: \href{justin.holmgren@ntt-research.com}{\texttt{justin.holmgren@ntt-research.com}.}} \and Kunal Mittal\thanks{\scriptsize Princeton University. E-mail: \href{kmittal@cs.princeton.edu}{\texttt{kmittal@cs.princeton.edu}.} Research supported by the Simons Collaboration on Algorithms and Geometry, by a Simons Investigator Award and by the National Science Foundation grants No. CCF-1714779, CCF-2007462.} \and Ran Raz\thanks{\scriptsize Princeton University.  E-mail: \href{ranr@cs.princeton.edu}{\texttt{ranr@cs.princeton.edu}.}  Research supported by the Simons Collaboration on Algorithms and Geometry, by a Simons Investigator Award and by the National Science Foundation grants No. CCF-1714779, CCF-2007462. } \and Wei Zhan\thanks{\scriptsize Princeton University.  E-mail: \href{weizhan@cs.princeton.edu}{\texttt{weizhan@cs.princeton.edu}.} Research supported by the Simons Collaboration on Algorithms and Geometry, by a Simons Investigator Award and by the National Science Foundation grants No. CCF-1714779, CCF-2007462.}}
\date{}
\maketitle

\begin{abstract}
We prove that for every 3-player (3-prover) game, with binary questions and answers and value $<1$, the value of the $n$-fold parallel repetition of the game decays  polynomially fast to 0. That is, for every such game, there exists a constant $c>0$, such that the value of the $n$-fold parallel repetition of the game is at most $n^{-c}$.

Along the way to proving this theorem, we prove two additional parallel repetition theorems for multiplayer (multiprover) games, that may be of independent interest:

\paragraph{Playerwise Connected Games (with any number of players and any Alphabet size):}
We identify a large class of multiplayer games and prove that for every game with value $<1$ in that class, the value of the $n$-fold parallel repetition of the game decays polynomially fast to 0.

More precisely, our result applies for \emph{playerwise connected games}, with any number of players and any alphabet size:
For each player $i$, we define the graph $G_i$, whose vertices are the possible questions for that player and two questions $x,x'$ are connected by an edge if there exists a vector $y$ of questions for all other players, such that both $(x,y)$ and $(x',y)$ are asked by the referee with non-zero probability.
We say that the game is \emph{playerwise connected} if for every $i$, the graph $G_i$ is connected.

Our class of playerwise connected games is strictly larger than the class of \emph{connected} games that was defined in~\cite{DHVY17} and for which  exponentially fast decay bounds are known~\cite{DHVY17}.
For playerwise connected games that are not connected, only inverse Ackermann decay bounds were previously known~\cite{Ver96}.


\paragraph{Exponential Bounds for the Anti-Correlation Game:}
In the {\em 3-player anti-correlation game}, two out of three players are given $1$ as input, and the remaining player is given $0$.  The two players who were given $1$ must produce different outputs in $\{0,1\}$. We prove that the value of the $n$-fold parallel repetition of that game decays exponentially fast to 0. That is, there exists a constant $c>0$, such that the value of the $n$-fold parallel repetition of the game is at most $2^{-c n}$.
Only inverse Ackermann decay bounds were previously known~\cite{Ver96}.

The 3-player anti-correlation game was studied and motivated in several previous works. In particular, Holmgren and Yang gave it as an example for a 3-player game whose non-signaling value (is smaller than 1 and yet) does not decrease at all under parallel repetition~\cite{HY19}.


\end{abstract}

\newpage
\tableofcontents
\newpage

\section{Introduction}
We study multiplayer games and their behavior under parallel repetition.
In a $k$-player game $\mc G$, a referee samples questions $x=(x^1, \dots, x^k)$ from some distribution $Q$.
Then, for each $j\in [k]$, the $j\textsuperscript{th}$ player is given the question $x^j$, based on which they give back an answer $a^j$.
The referee then declares if the players win or not based on the evaluation of a predicate $V(x^1, \dots, x^k, a^1, \dots, a^k)$.
The value $\val(\mc G)$ of the game $\mc G$ is defined to be the maximum winning probability for the players, where the maximum is over all possible strategies (functions mapping questions to answers) of the players.

A very basic operation on a game $\mc G$ is to consider its parallel repetition, in which the players are asked to play many independent copies of the game in parallel.
More formally, in the $n$-fold parallel repetition $\mc G^{\otimes n}$, the referee draws questions $(x_i^1, \dots, x_i^k)$ from $Q$, independently for each $i\in [n]$.
Then, for each $j \in [k]$, the $j\textsuperscript{th}$ player is given the questions $(x_1^j, \dots, x_n^j)$, based on which they answer back $(a_1^j, \dots, a_n^j)$.
The referee says that the players win if for every $i\in [n]$, the predicate $V(x_i^1, \dots, x_i^k, a_i^1, \dots, a_i^k)$ evaluates to win.

A natural question is to study how the value of the game $\mc G^{\otimes n}$ behaves as a function of~$n$, the number of parallel repetitions~\cite{FRS94}.
It is not hard to see that $\val(\mc G^{\otimes n}) \geq \val(\mc G)^n$, since the players can achieve value $\val(\mc G)^n$ in the game $ \mc G^{\otimes n}$ by simply repeating an optimal strategy for the game $\mc G$ independently in all the $n$ coordinates.
It also seems that this should be optimal, and that $\val(\mc G^{\otimes n}) \leq \val(\mc G)^n$.
However, this turns out not to be the case, and there are games such that $\val(\mc G^{\otimes n})$ is exponentially larger than $\val(\mc G)^n$ \cite{For89, Fei91, FV02, Raz11}.
Hence, it is interesting to study the behavior of $\val(\mc G^{\otimes n})$ for games $\mc G$ with $\val(\mc G)<1$.

The special case of 2-player games is very well understood, and it was proven by Raz \cite{Raz98} that if $\val(\mc G)<1$, then the value of $\mc G^{\otimes n}$ decays exponentially in $n$; that is, $\val(\mc G^{\otimes n}) \leq 2^{-\Omega(n)}$, with the constants depending on the base game $\mc G$.
There have been improvements in the constants \cite{Hol09, Rao11, BRRRS09, RR12}, and we even know tight results based on the value of the initial game \cite{DS14, BG15}.
These results on 2-player games have found many applications, in particular in the theory of interactive proofs \cite{BOGKW88}, PCPs and harness of approximation \cite{BGS98, Fei98, Has01}, geometry of foams \cite{FKO07, KORW08, AK09}, quantum information \cite{CHTW04}, and communication complexity \cite{PRW97, BBCR13, BRWY13}.
The reader is referred to this survey \cite{Raz10} for more details.

The case of general $k$-player multiplayer games is still open.
The only general result, by Verbitsky \cite{Ver96}, says that if $\val(\mc G)<1$, then $\val(\mc G^{\otimes n})\to 0$ as $n\to \infty$.
This result uses the density Hales-Jewett theorem as a black box, and gives bounds of the form $\frac{1}{\alpha(n)}$, where $\alpha$ is an inverse Ackermann function \cite{FK91, Pol12}.
Apart from being interesting in its own right, studying parallel repetition of multiplayer games has some applications.
For example, it is known that a strong parallel repetition theorem for a particular class of multiplayer games implies super-linear lower bounds for Turing machines in the non-uniform model \cite{MR21}.
Also (as mentioned by \cite{DHVY17}), the technical limitations that arise when analyzing games with more than two players seem very similar to the ones we encounter when studying direct sum and direct product questions for multiparty number-on-forehead communication complexity (which is related to lower bounds in circuit complexity).
Therefore, studying parallel repetition for multiplayer games may lead to progress in these areas.

Although we know very little about general multiplayer games, there has been some recent progress on special classes of multiplayer games:
\begin{enumerate}
	\item Dinur, Harsha, Venkat and Yuen \cite{DHVY17} extend the two player techniques of \cite{Raz98, Hol09} and show that any \emph{connected game} satisfies an exponentially small bound on the value of parallel repetition (and this includes all games for which exponentially small bounds were previously known). The class of connected games is defined as follows: Define the graph $\mc{H_\mc G}$,
whose vertices are the ordered $k$-tuples of questions to the $k$-players, and there is an edge between questions $x$ and $x'$ if they differ in the question to exactly one of the $k$ players, and are the same for the remaining $k-1$ players. The game is said to be \emph{connected} if the graph $\mc H_{\mc G}$ is connected.
	\item The GHZ game \cite{GHZ89} is defined as follows: The referee samples the questions $(x^1, x^2, x^3)$ uniformly at random from $\set{0,1}^3$ such that $x^1\oplus x^2\oplus x^3=0$. The players answer back with $a^1, a^2, a^3\in \set{0,1}$, and are said to win if $a^1\oplus a^2\oplus a^3 = x^1 \lor x^2 \lor x^3$. It has been shown that any game with the same distribution as the GHZ game satisfies an inverse polynomial bound on the value of parallel repetition \cite{HR20, GHMRZ21}.

\end{enumerate}

\subsection{Our Results}

We prove that for every 3-player game, with binary questions and answers and value $<1$, the value of the $n$-fold parallel repetition of the game decays  polynomially fast to 0. 

\begin{theorem}\label{thm:intro_main}
	Let $\mc G$ be a 3-player game such that ${\val(\mc G)<1}$ and each question and answer is in $\{0,1\}$. Then, there exists a constant $c>0$, such that $\val(\mc G^{\otimes n}) \leq n^{-c}$.
\end{theorem}

In the proof of Theorem~\ref{thm:intro_main}, we show that from the perspective  of studying the behaviour of $\val(\mc G^{\otimes n})$ as a function of $n$, every 3-player game $\mc G$ with binary questions and answers, is equivalent to, or can be reduced to, a game in one of the following five classes:

\begin{enumerate}
\item \label{C1}
{\bf 2-Player Games:} As mentioned above, exponentially small bounds on the value of the parallel repetition of games in this class have been known for a long time.
\item \label{C2}
{\bf Playerwise Connected Games:} This is a new class of games that we define and study in this work and we prove polynomially small bounds on the value of the parallel repetition of games in this class.
\item \label{C3}
{\bf The GHZ Game:} (and other games with the same query distribution): As mentioned above, polynomially small bounds on the value of the parallel repetition of games in this class were recently proved.
\item \label{C4}
{\bf The Anti-Correlation Game:} (and other games with the same query distribution and binary answers):
The 3-player anti-correlation game is defined as follows: The referee samples the questions $(x^1, x^2, x^3)$ uniformly at random from $\set{0,1}^3$ such that $x^1 + x^2 + x^3=2$ (that is, two out of three players are given $1$ as input, and the remaining player is given $0$).
The two players who were given $1$ must produce different outputs in $\{0, 1\}$. We prove exponentially small bounds on the value of the parallel repetition of that game (and all other games with the same query distribution and binary answers).
\item \label{C5}
{\bf Games over the Set of Questions $\{(0,0,0),(0,1,0),(1,0,0),(1,1,1)\}$:}
We prove polynomially small bounds on the value of the parallel repetition of games in this class.
\end{enumerate}

We note that the reduction to these five classes of games works more generally for all 3-player games with binary questions and arbitrary length of answers, except that we need to extend Class~\ref{C4}  so that it contains games with arbitrary length of answers. Note also that for all other classes, the bounds that we have  hold more generally for games with arbitrary length of answers. This  means that improving the bounds that we prove for Class~\ref{C4} so that they hold for arbitrary length of answers
(or even proving weaker, polynomially small, bounds for that case) would imply that Theorem~\ref{thm:intro_main} holds more generally, for games with arbitrary answer length.

We note also that the three new bounds that we prove in this work, the bounds for Class~\ref{C2}, Class~\ref{C4} and Class~\ref{C5}, are each proved by a completely different proof method. Next we elaborate on each of these three classes.

\subsubsection{The Anti-Correlation Game (Class~\ref{C4})}

In the hilarious essay {\em ``Test Your Telepathic Skills''}, Uri Feige tells the fictional story of the {\em ``amazing Tachman family''}, who astonished the team at FEXI {\em (the Foolproof Experiments Institute)} with their telepathic skills, by playing incredibly well the 3-fold parallel repetition of the 3-player anti-correlation game~[Uri Feige, 1995]\footnote{https://www.wisdom.weizmann.ac.il/$\sim$feige/tachman.html (Feige's description of the game is somewhat different than ours and is described below).}.
Feige showed that the value of the 3-player anti-correlation game, played in parallel 3 times, is $\tfrac{2}{3}$, exactly the same as the value of the original game.

More than two decades later, Holmgren and Yang 
proved that while the, so called, {\em non-signaling value}, of the 3-player anti-correlation game is strictly smaller than 1, it does not decrease at all under parallel repetition~\cite{HY19}.
This gave a surprising first example for a total failure of parallel repetition in reducing the value of a game, in any model of multiplayer games.

Hazla, Holenstein and Rao studied games with the same query distribution as the anti-correlation game~\cite{HHR16} and showed barriers on proving parallel repetition theorems for such games using a technique known as the {\em  forbidden subgraph bounds}~\cite{FV02}.

The anti-correlation game can also be presented as a ``pigeonhole-principle'' game, where 2 out of 3 pigeons are chosen randomly and each of them needs to choose 1 out of 2 pigeonholes, without communicating between them, so that the two chosen pigeons end in 2 different pigeonholes. This may occur in situations when 3 players share 2 identical resources (such as 2 communication channels to an external party): Two (randomly chosen) players (out of the three players) need to use one of the two resources each and there is no communication between the players. Another description of the game, the one that was presented by Feige, can be viewed as a matching game: The 3 players try to output 3 different answers $X,Y,Z$, where two of the players, chosen randomly, can only output $Y$ or $Z$ and the remaining player can only output $X$ or $Z$.

Although the 3-player anti-correlation game has been around for more than two and a half decades, no bound on the value of its parallel repetition was previously known (other than Verbitsky's general inverse Ackermann bound on the value of the parallel repetition of every game~\cite{Ver96}). In this work, we prove
that the value of the $n$-fold parallel repetition of the 3-player anti-correlation game decays exponentially fast to 0.
(We also extend this bound to all other games with the same query distribution and binary answers).

\begin{theorem}\label{thm:intro_anticorrelation}
	Let $\mc G$ be the 3-player anti-correlation game (or any other game with the same query distribution and binary answers). Then,
there exists a constant $c>0$, such that  $\val(\mc G^{\otimes n}) \leq 2^{-c n}$.
\end{theorem}

In light of the above mentioned result by Holmgren and Yang~\cite{HY19}, Theorem~\ref{thm:intro_anticorrelation} also implies an example for a 3-player game where the value of the parallel repetition of the game behaves completely differently for classical strategies versus non-signaling strategies. Namely, while parallel repetition doesn't decrease the non-signalling value of the game at all, it does decrease the classical value of the game exponentially fast to 0.

\paragraph{Techniques:}
The techniques that we use for the proof of Theorem~\ref{thm:intro_anticorrelation} are, to the best of our knowledge, completely new in the context of parallel repetition and are different than the techniques used in all  previous works. In particular, we don't use here the usual {\em embedding paradigm}, that is used in almost all previous works, where one tries to embed a copy of the original game in the set of success of the players on some set of coordinates. Instead, our proof shows a {\em local to global} property of the strategy of each player. Very roughly speaking, we prove that if the players win the parallel repetition game with  sufficiently high probability, then there exists a fixed (large) set of coordinates and a fixed global constant strategy for each of the players, that doesn't depend on the input for the player at all, and such that the global strategies win the parallel repetition game with a sufficiently high probability, on almost all the coordinates in the fixed set of coordinates. This leads to a contradiction since fixed global strategies are, in particular, independent between the different coordinates.
We note also that this is the first inverse exponential bound on the parallel repetition of any 3-player game that is not connected (in the sense of~\cite{DHVY17})\footnote{or reduces to a connected game}.

\subsubsection{Playerwise Connected Games (Class~\ref{C2})}

We define the class of \emph{playerwise connected games} as follows:
For each player $j$, we define the graph $\mc H_{\mc G}^j$, whose vertices are the possible questions for player $j$, and two questions $x^j$ and $x'^j$ are connected by an edge if there exists a vector $y$ of questions for all other players, such that both $(x^j,y)$ and $(x'^j,y)$ are asked by the referee with non-zero probability.
We say that the game is \emph{playerwise connected} if for every $j$, the graph $\mc H_{\mc G}^j$ is connected.

We prove polynomially small bounds on the value of the parallel repetition of any game in this class:


\begin{theorem}\label{thm:intro_playerwiseconnected}
	Let $\mc G$ be a \emph{playerwise connected} game such that ${\val(\mc G)<1}$ (with any number of players and any Alphabet size).
Then, there exists a constant $c>0$, such that $\val(\mc G^{\otimes n}) \leq n^{-c}$.
\end{theorem}

Theorem~\ref{thm:intro_playerwiseconnected} gives an inverse polynomial bound on the value of parallel repetition for many games for which the previously best known bound was inverse Ackerman.

Our class of playerwise connected games is related to the above mentioned class of connected games that was studied by Dinur, Harsha, Venkat and Yuen and for which exponentially small bounds were established~\cite{DHVY17}.
Observe that every connected game is also playerwise connected (the graph $\mc H_{\mc G}^j$ is simply the \emph{projection} of the graph $\mc H_{\mc G}$ in the $j$\textsuperscript{th} direction).
The vice-versa is however not true:

\begin{example}\label{ex:player_conn_but_not_conn}
The following  3-player game 
is playerwise connected, but not connected:
The referee samples $(x,y,z)$ uniformly from
$\mc S = \set{(0,0,0), (1,0,0), (0,1,0), (0,0,1), (1,1,1)}$ 
and gives $x,y,z$ to the three players respectively.
The players give answers $a,b,c\in \set{0,1}$ respectively.
The players win if the following condition holds: $a+b+c=1 \iff x+y+z\not=3$.
\end{example}

We note that the set $\mc S$ of possible questions, from Example~\ref{ex:player_conn_but_not_conn}, is the only set with 3 players and binary questions  that gives a game that is playerwise connected but not connected (up to a change of names). When the number of players is larger than 3, or the question's Alphabet size is larger than 2, there are many additional examples.

\begin{example}\label{ex:rand_3CNF}
	Fix a random 3-CNF formula $\varphi = (C_1, \dots, C_m)$, with $m$ clauses, over $d$ variables.
	This is generated by sampling $m$ times independently and uniformly from the set of all $(2d)^3 = 8d^3$ possible clauses.
	
	A 3-player game $\mc G$ is defined by this formula $\varphi$ as follows:
	The referee samples $r\in [m]$ uniformly and gives the variables corresponding to the literals in $C_r$ to the 3 players respectively (with each player getting one variable).
	The players answer back values for the variables they get, and the referee declares that the players win if these values satisfy the clause $C_r$.
	
	Then, it is not hard to show (see Appendix \ref{app:rand_3CNF}) that with high probability:
	\begin{enumerate}
		\item If $m = \omega(d)$, the value of the game is close to 7/8, and hence less than 1.
		\item If $m = \omega(d^2\log{d})$, the graph $\mc H_{\mc G}$ is connected, and we get $\val(\mc G^{\otimes n}) = 2^{-\Omega(n)}$ by \cite{DHVY17}.
		Furthermore, if $m = o(d^2)$, the graph $\mc H_{\mc G}$ is not connected, and \cite{DHVY17} is not applicable.
		\item If $m = \omega(d^{1.5}\sqrt{\log d})$, the game $\mc G$ is playerwise connected, and $\val(\mc G^{\otimes n}) = n^{-\Omega(1)}$ by Theorem \ref{thm:intro_playerwiseconnected}.
		Furthermore, if $m = o(d^{1.5})$, the game $\mc G$ is not playerwise connected.
	\end{enumerate}
	
	Note that the $\omega$ and $o$ bounds on $m$ are with respect to $d\to\infty$.
    Once the formula $\varphi$ is fixed, we think of $m$ and $d$ as constants and the $\Omega$ bounds on the value of parallel repetition are with respect to $n\to\infty$.
    
    \end{example}
		
	\begin{remark}
The above example is also interesting when compared to the works on refutation of random 3-CNFs, where different regimes of the parameter $m$ lead to different consequences. It is known that with high probability:
	\begin{enumerate}
		\item If $m=\Omega(d^{1.5})$, there is a polynomial time algorithm for refuting the random 3-CNFs \cite{FO07}.
	    \item If $m=\Omega(d^2/\log{d})$, resolution provides polynomial size witnesses for refutation. On the other hand, it fails to provide short witnesses when $m=O(d^{1.5-\epsilon})$ \cite{CS88}.
	    \item If $m=\Omega(d^{1.4})$, there exist polynomial size witnesses for refutation, based on spectral approach \cite{FKO06}.
	\end{enumerate}
	In both cases, there is a polynomial gap in $d$ between the base assumption $m=\omega(d)$ and the regime of best known results.
	\end{remark}

\paragraph{Techniques:}
Our proof of Theorem~\ref{thm:intro_playerwiseconnected} relies on information-theoretic techniques, extending the ideas of \cite{Raz98, Hol09, DHVY17}. In particular, we use here the usual {\em embedding paradigm} and condition on a {\em dependency breaking event}, as in many previous works.
However, these techniques heavily rely on the game being connected and thus the result of~\cite{DHVY17} applies only to connected games\footnote{or disjoint unions of connected games} and we are not aware of any previous work that applies these techniques to games that are not connected. We hence need to deviate from these techniques at a crucial point. Very roughly speaking, at a crucial place in the proof where the connectivity of the game is necessary, our key idea is to replace the distribution of the game with a connected distribution and we carefully analyze how this change affects the rest of the proof.

\subsubsection{Support  ${\bf \{(0,0,0),(0,1,0),(1,0,0),(1,1,1)\}}$ (Class~\ref{C5})}

We consider 3-player games where the set of possible questions for the 3 players is: $\{(0,0,0),(0,1,0),(1,0,0),(1,1,1)\}$,  and
we prove polynomially small bounds on the value of the parallel repetition of any game in this class:

\begin{theorem}\label{thm:intro_4points}
	Let $\mc G$ be a 3-player game where the possible questions for the 3 players are: $(0,0,0),(0,1,0),(1,0,0),(1,1,1)$, and such that ${\val(\mc G)<1}$ (with any length of answers).
Then, there exists a constant $c>0$, such that $\val(\mc G^{\otimes n}) \leq n^{-c}$.
\end{theorem}

Theorem~\ref{thm:intro_4points} is necessary for the proof of Theorem~\ref{thm:intro_main} and we believe that the proof technique is interesting and might be useful for other games.

\paragraph{Techniques:}
The techniques that we use for the proof of Theorem~\ref{thm:intro_4points} are, to the best of our knowledge, new.
Very roughly speaking, we consider the possible pairs of answers $(a,b)$ by Player~1 and Player~2 on questions $(1,1)$ for these two players. We distinguish between pairs $(a,b)$ for which Player~3 has an answer $c$ such that the referee accepts the answers $(a,b,c)$ on questions $(1,1,1)$ and pairs $(a,b)$ for which Player~3 has no answer $c$ such that the referee accepts the answers $(a,b,c)$ on questions $(1,1,1)$. Intuitively, if the pair $(a,b)$ is of the second type, that is, no answer $c$ causes the referee to accept on $(1,1,1)$, the pair $(a,b)$ cannot be used too often by Player~1 and Player~2, and we are able to make this intuition precise by conditioning on a carefully and inductively defined, but possibly exponentially small, product event between the inputs of Player~1 and Player~2. Very roughly speaking, when all pairs $(a,b)$ of the second type are used with negligible (polynomially small) probability when conditioning on our product event, we are able to essentially reduce the parallel repetition game to parallel repetition of a 2-player game with value $< 1$, played by Player~1 and Player~2 on a subset of the coordinates and conditioned on the product event between the inputs of Player~1 and Player~2 that we defined. We then rely on the fact that bounds for 2-player games also hold when conditioning the inputs of the two players on a product event between the two players. The final bound that we obtain is inverse polynomial, rather than inverse exponential, because we must take into account the answers $(a,b)$ of the second type that are still used with polynomially small probability. We do that using a union bound and it's crucial here that the 2-player game that we reduce to is only played on a small subset of the coordinates so that we can apply a union bound over these coordinates.

\section{Overview}

\subsection{Organization}

For the problem of parallel repetition for all three-player games on binary alphabets, the following results were known prior to our work (\Cref{sec:previous_bounds}).

\begin{enumerate}
	\item For three-player games in which there are some two players whose inputs are in a bijective correspondence, we may treat these players as identical, and thus reduce the problem to showing parallel repetition for two-player games. The Parallel Repetition Theorem of \cite{Raz98} shows that parallel repetition decreases the value of two-player games exponentially fast (\Cref{thm:2p_par_rep}).
	\item There is a class of games known as {\it connected} ({or \it expanding}) games for which~\cite{DHVY17} showed an exponential decay on the parallel repetition value (\Cref{thm:conn_par_rep}). A $k$-player game is said to be connected if the $(k-1)$-connection graph is connected. This graph is defined as follows: the vertices are the elements in the support of the query distribution and the edges are between every pair of elements that agree on the questions to all but one player (\Cref{defn:conn_graph}).
	\item For any game (with value less than one) for which the query distribution has support \\$\{(0,0,0), (1,1,0), (1,0,1), (0,1,1)\}$,~\cite{HR20, GHMRZ21} showed that parallel repetition decreases the value at least polynomially fast (\Cref{thm:ghz_par_rep}). 
\end{enumerate}

In this work, we study all three-player binary-alphabet games that do not fall into the above categories. In \Cref{sec:three_player_binary}, we classify all such games. It turns out that there are essentially three such classes of games.
\begin{enumerate}
	\item Games whose query distribution has support $\{(0,1,1),(1,0,1),(1,1,0)\}$. Of these games, the {anti-correlation game} is the most interesting one. In this game, the players who receive one need to output distinct bits (\Cref{def:anti_corr_game}). We prove an exponential decay on the parallel repetition value of this game in \Cref{sec:anti_corr_game} (\Cref{thm:anti_corr_par_rep}). In \Cref{sec:hm_wt_one}, we show that this implies a similar result for all games with binary outputs whose query distribution has the same support as the anti-correlation game (\Cref{thm:hm_wt_one_par_rep}). 
	\item Games whose query distribution has support $\{ (0,0,0), (0,1,0),(1,0,0),(1,1,1)\}$. We refer to the uniform distribution on these four points as the four-point AND distribution. In \Cref{sec:four_point_game}, we show that parallel repetition for such games decreases the value at least polynomially fast (\Cref{thm:four_point_par_rep}). We remark that our result holds even if the answers are from an arbitrary alphabet. 
	\item Games whose query distribution has support $\{ (0,0,0),(0,0,1),(0,1,0),(1,0,0),(1,1,1)\}$. Such games fall into an even more general class of games which we call {\it playerwise connected} games; these are $k$-player games in which the projection of the $(k-1)$-connection graph on every player is connected (\Cref{defn:conn_player_graph}). We show in \Cref{sec:playerconn_games}  that parallel repetition for this class of games decreases the value at least polynomially fast (\Cref{thm:playerwise_conn_par_rep}). We remark that our result holds even if the answers are from an arbitrary alphabet. 
\end{enumerate}

\subsection{The Anti-Correlation Game}
In the three-player anti-correlation game $\cG$, a random pair of players are given $1$ as input, and the remaining player is given $0$. To win, the two players who are given $1$ must produce different outputs in $\{0, 1\}$, and the output of the player who is given $0$ does not matter. We present an overview of the proof of \Cref{thm:anti_corr_par_rep} which shows that the parallel repetition of the anti-correlation game decreases the value exponentially fast. The details are presented in \Cref{sec:anti_corr_game}.

Let $Q$ denote the joint input distribution for all the players in the game $\cG$, and let $X$, $Y$, and $Z$
respectively denote the first, second, and third player's inputs in the game $\mc G^{\otimes n}$.  
Let $f, g, h : \{0,1\}^n \to \{0,1\}^n$ be any strategy that wins the $n$-wise repeated game $\cG^{\otimes n}$
with probability $\alpha > 0$.  Our goal is to prove that $\alpha \le e^{-\Omega(n)}$. 

The players' inputs are fully determined by the inputs of any pair of players by the equation $X_i + Y_i + Z_i = 2$ for all $i \in [n]$.  We will say ``$(f, g, h)$ wins on $(x, y)$'' as short-hand for ``$(f, g, h)$ wins on $(x, y, z)$ where $z_i = 2 - x_i - y_i$ for each $i$''.

\paragraph{Winning Implies Self-Agreement on Correlated $X$, $X'$:}
We first consider a distribution in which $X,X', Y$ are random variables with both $(X, Y)$ and $(X', Y)$ distributed like $Q_{X,Y}^{\otimes n}$, and with $X$ and $X'$ conditionally independent given $Y$.  More explicitly, this distribution is sampled as follows:
\begin{enumerate}
    \item Sample an $n$-bit string $Y$ according to $Q_Y^{\otimes n}$. That is, for each $i \in [n]$ independently sample $Y_i = 1$ with probability $2/3$, and $Y_i = 0$ otherwise.
    \item Independently sample $X$ and $X'$ from the conditional distribution of the first player's input in $\cG^{\otimes n}$ given that the second player's input is $Y$.  That is, for each $i \in [n]$, if $Y_i = 0$, set $X_i = X'_i = 1$.  Otherwise, independently sample $X_i, X'_i \gets \{0,1\}$ uniformly at random.
\end{enumerate}

By the assumption that $(f, g, h)$ wins with probability $\alpha$, we know that $(f, g, h)$ wins on $(X, Y)$ with probability $\alpha$, and $(f, g,h )$ wins on $(X', Y)$ with probability $\alpha$.  Because of how $(X, Y)$ and $(X', Y)$ are correlated, we show that $(f, g, h)$ must simultaneously win on both $(X, Y)$ and $(X', Y)$ with probability at least $\alpha^2$.  Thus,
\begin{equation}
\label{eq:conditional-win-prob}
\Pr_{X,X'} \Big [ \Pr_Y \big [ \text{$(f, g, h)$ wins on $(X, Y)$ and on $(X', Y)$} \big ] \ge \alpha^2 / 2 \Big ] \ge \alpha^2/2.
\end{equation}
Now suppose that $X$ and $X'$ are such that $\Pr_Y \big [ (f, g, h) \text{ wins on $(X, Y)$ and on $(X', Y)$} \big ] \ge \alpha^2 / 2$.
We have $\Pr[Y_i = 1|X, X'] = 1/3$ for $i$ such that $X_i = X'_i = 1$, and all such $Y_i$ are conditionally independent given $X$, $X'$.  Also, if $X_i = X'_i = 1$ and $f(X)_i \neq f(X')_i$, then the only way $(f, g, h)$ can win on $(X, Y)$ and on $(X', Y)$ is if $Y_i = 0$, because if $Y_i = 1$ then the win conditions require that $f(X)_i \neq g(Y)_i \neq f(X')_i$.  Combining these two facts implies that $f(X)_i \neq f(X')_i$ for at most $\log_{2/3}(\alpha^2 / 2) = O(\log (1 / \alpha))$ coordinates $i$ with $X_i = X'_i = 1$.


This shows that when $X$, $X'$, and $Y$ are sampled as above, it holds with $\poly(\alpha)$  probability that:
\begin{enumerate}
\item (Approximate Self-Agreement): For all but at most $O(\log 1 / \alpha)$ values of $i \in [n]$, if $X_i = X'_i = 1$ then $f(X)_i = f(X')_i$.  
\item (Winning): $(f, g, h)$ wins on $(X, Y)$ and on $(X', Y)$.  \end{enumerate}

A similar argument gives an analogous statement, where the Winning property is replaced by a property that we call Winning', requiring that $(f, g, h)$ has probability at least $\alpha / 2$ of winning in $\cG^{\otimes n}$ conditioned on the first player's input being $X'$ (\Cref{claim:P-good}).

We say that $(X, X')$ is \emph{good} if it satisfies Approximate Self-Agreement and Winning'. 
\paragraph{Constructing a strategy for $\cG^{\otimes n'}$:}
We next use the fact that $(X, X')$ is good with $\poly(\alpha)$ probability to show that if there exists a strategy $(f, g, h)$ that wins $\cG^{\otimes n}$ with probability $\alpha$, then there exists a strategy $(f', g', h')$ for $\cG^{\otimes n'}$ (with $n' \ge \Omega(n)$) such that:
\begin{itemize}
    \item $f'$ is a constant function, and
    \item $(f', g', h')$ wins in all but $O(\log (1 / \alpha))$ coordinates of $\cG^{\otimes n'}$ with probability $\poly(\alpha)$.
\end{itemize}
(See \Cref{prop:f-nearly-constant} for a formal statement.) The main idea is that $(X, X')$ can be equivalently sampled as follows:
\begin{enumerate}
    \item Sample each bit of $X_i$ independently such that $X_i = 1$ with probability $2/3$.
    \item Sample a set $S \subseteq [n]$ by independently including each $i \in [n]$ with probability $1/4$.
    \item  For $i \in S$, set $X'_i = X_i$.  For all other $i$, sample $X'_i$ such that $X'_i = 1$ with probability $2/3$. 
\end{enumerate}
The point of this alternative sampling process is that  conditioned on any value of $X$ and $S$, the distribution of $X'_{-S}$ is $Q_X^{\otimes n - |S|}$.  In contrast, the distribution of $X'$ given $X$ is not $Q_X^{\otimes n}$ because of the correlation between $X$ and $X'$.

We will first condition on random values of $X$, $S$, and $X'_T$, where $T = S \cup \{i : X_i = 0\}$. This ensures that
\[
\Pr_{X,S,X'_T} \Big [ \Pr_{X'_{-T}} \big [ \text{$(X, X')$ is good} \big ] \ge \poly(\alpha) \Big ] \ge \poly(\alpha).
\]  Also, the conditional distribution of $X'_{-T}$ given any values of $X$, $S$, and $X'_T$ is just the first player's input distribution in $\cG^{\otimes n - |T|}$.
This means that we can view $f$ as inducing a first-player strategy $f'$ on $\cG^{n'}$ for $n' = n - |T|$ by fixing the appropriate part of $f$'s input to $X'_{T}$.  Note that $n' = \Omega(n)$ with overwhelming probability.

Part of $(X, X')$ being good means that $(f, g, h)$ has $\poly(\alpha)$ probability of winning on $(X', Y')$ when $Y'$ is sampled from the distribution of the second player's input in $\cG^{\otimes n}$ conditioned on the first player's input being $X'$.  We split the sampling of $Y'$ into two parts: $Y'_T$ and $Y'_{-T}$.  We show that we can sample and fix $Y'_T$, and use it to define $g'$ and $h'$ analogously to $f'$ (for $h'$ implicitly defining $Z'_T$ such that $(Z'_T)_i = 2 - (X'_T)_i - (Y'_T)_i$), such that with $\poly(\alpha)$ probability over the choice of $(X'_{-T}, Y'_{-T})$:
\begin{itemize}
    \item $(f', g', h')$ wins on $(X'_{-T}, Y'_{-T})$, and 
    \item $(X, X')$ is good.  In particular, since $X_{-T}$ is the all-ones string, we have that $f'(X'_{-T})_i = f(X')_i = f(X)_i$ for all but $O(\log 1 /\alpha)$ values of $i \notin T$ for which $(X'_{-T})_i = 1$.
\end{itemize}
This implies that  up to a difference in its outputs for $O(\log 1 / \alpha)$ coordinates, $f'$ might as well be the constant function that always outputs $f(X)_{-T}$.

We could in principle continue onwards, eventually finding a smaller $n''$ (still $\Omega(n)$) such that $\cG^{\otimes n''}$ has a strategy $(f'', g'', h'')$ that wins in all but $O(\log 1 / \alpha)$ coordinates with probability $\poly(\alpha)$, but consists only of constant functions.  This is  a contradiction unless $\alpha \le e^{-\Omega(n)}$.   For simplicity, however, we instead directly show that when $f'$ is a constant function, the strategy $(f', g', h')$ must lose in a constant fraction of the coordinates with all but exponentially small probability (\cref{prop:f-cant-be-constant}).  This implies that $\alpha$ is $e^{-\Omega(n)}$.  

\subsection{Four-Point AND Distribution}

We present a technical overview of the proof of  \Cref{thm:four_point_par_rep} which shows an inverse polynomial bound for the parallel repetition value for the four-point AND distribution. The details can be found in \Cref{sec:four_point_game}. 

We first note the following observation about the set of points $\mc S:=\{(0,0,0),(0,1,0),(1,0,0),(1,1,1)\}$. Let $\mc G$ be any game on $\mc S$ with value less than one and fix any strategy $(f, g, h)$ for the players for one copy of the game $\mc G$. Observe that whenever Charlie receives 1, he knows that the inputs of Alice and Bob are 1. Consider the following two cases.
\begin{itemize}
\item Case A: Suppose the answers of Alice and Bob on input 1 are such that there is {\it no answer} that Charlie can give to satisfy the predicate, then the strategy loses on the input $(1,1,1)$. \item Case B: On the other hand, suppose the answers of Alice and Bob on input 1 are such that {\it there exists an answer for Charlie} that satisfies the predicate, then we can assume that Charlie answers this on input 1. Since the value of the game is less than one and this strategy succeeds on $(1,1,1)$, the strategy must fail on one of the remaining three points $\{(0,0,0),(0,1,0),(1,0,0)\}$. For these points, Charlie's input is fixed to zero and in particular, his answer is also fixed. Therefore, the predicate $V$ when restricted to these inputs, induces a predicate $\tilde{V}$ which depends only on the inputs and outputs of Alice and Bob. The game $\mc G$ thus defines a two-player game on the uniform distribution $\tilde{\mc S}:=\{(0,0),(0,1),(1,0)\}$ with predicate $\tilde{V}$.
\end{itemize} 
Next, we consider the $n$-fold parallel repetition of the game.
\paragraph*{Pre-processing the game.} Let $Q$ be the uniform distribution over $\mc S$ and let $P=Q^{\otimes n}$. We will always maintain a product event of the form $E=E^1\times E^2\times\{0,1\}^n$ on the players inputs where $E^1$ is a subset of Alice's inputs and $E^2$ is a subset of Bob's inputs. We begin by showing that we can assume without loss of generality that conditioned on a large product event across Alice's and Bob's inputs, all coordinates satisfy a property similar to case B with high probability, otherwise, we would get exponential decay of the parallel repetition value. For each $i\in[n]$, let $\tilde{L}_i$ denote the event that Alice and Bob receive 1 in the $i$-th coordinate and their answers are such that {\it there is no answer for Charlie} that satisfies the predicate. Let $\tilde{W}_i$ denote the complement of $\tilde{L}_i$. For $S\subseteq [n]$, let $W_S$ denote the event that the players win the game in coordinate $i$ for all $i\in S$.  Note that whenever $\tilde{L}_i$ happens, the players lose in the $i$-th coordinate. Thus,
\begin{equation}\label{intro:eq} \Pr[W_{[n]}] \le \Pr[\tilde{W}_1]\cdot \Pr[\tilde{W}_2| \tilde{W}_1]\cdots \Pr[\tilde{W}_n| \tilde{W}_{1},\ldots,\tilde{W}_{n-1}].\end{equation} 
While there exists a coordinate $i\in[n]$ such that $\Pr[\tilde{L}_i|E]$ is significant, i.e., $\Pr[\tilde{L}_i| E]\ge 1/n^{4\delta}$ for some small constant $\delta>0$, we will try to {\it condition on $\tilde{W}_i$} 
and proceed. (To ensure that we maintain a product event across Alice and Bob, we will also randomly fix their inputs and answers in the $i$-th coordinate and update $E$ based on this fixing and proceed.\footnote{This can be done since $\tilde{W}_i$ and $\tilde{L}_i$ 
depend only on the inputs and answers to Alice and Bob in the $i$-th coordinate.}) If this conditioning process happens more than $n^{5\delta}$ times, then \Cref{intro:eq} implies that the probability of winning all coordinates is at most $(1-1/n^{4\delta})^{n^{5\delta}}\le \exp(-\Omega(n^{\delta}))$. It suffices to study the other case, that is, this conditioning process happens at most  $n^{5\delta}$ times. 

\paragraph*{Reduction to a two-player game.} We are left with a product event $E$ of the form $E^1\times E^2\times \{0,1\}^n$ of measure at least $\exp(-\Omega(n^{5\delta}))$ such that for all $i\in [n]$, we have $\Pr[\tilde{L}_i|E]\le 1/n^{4\delta}$ and we will show that the probability of winning all coordinates when the inputs are drawn from $P|E$ is at most $n^{-\Omega(\delta)}$ (essentially \Cref{lemma:4_pt_main_lemma}).  Let  $\mc A_i$ be the set of answers $a_i$ such that with significant probability (namely, more than $1/n^{2\delta}$ probability over $P|E$ 
), Alice's input in $i$-th coordinate is 1 and her output in the $i$-th coordinate is $a_i$. Define $\mc B_i$ for Bob similarly. Since Alice's and Bob's inputs are independent under $P$ and $E$ is a product event across Alice and Bob, and $\Pr[\tilde{L}_i|E]\le 1/n^{4\delta}$, it follows that for every pair of answers of Alice and Bob in $\mc A_i\times  \mc B_i$, there is an answer that Charlie can give so that the predicate is satisfied when all players get 1. Define a product event $G_i$ across Alice and Bob which is true if and only if whenever Alice and Bob get input 1 in the $i$-th coordinate, they answer from $\mc A_i\times \mc B_i$ (\Cref{defn:four_point_G}). A union bound implies that 
\begin{equation}
	\label{intro:eq3}\Pr[G_i|E]\ge  1- \frac{|\mc A|}{n^{2\delta}}-\frac{|\mc B|}{n^{2\delta }}  \ge 1-O(n^{-\delta}).
\end{equation}

We now randomly fix an input $z\in\{0,1\}^n$ to Charlie. Let $K$ denote the set of coordinates that are zero in $z$, and let $m$ denote $|K|$. 
With all but exponentially small probability, $m=\Omega(n)$. We also randomly fix the inputs $x_{-K},y_{-K}$ to Alice and Bob in coordinates outside $K$.  Pick a random subset $S\subseteq K$ of size $m^{\epsilon}$ for some constant $0<\epsilon<\delta$. We have (in expectation over $z, x_{-K}, y_{-K}$): 
\begin{align} 
	\label{intro:eq2} \begin{split} 
	\Pr[W_S|E] &\le \Pr[\vee_{i\in S} \neg G_i|E] \\  &+ \Pr[\wedge_{i\in S} G_i \wedge W_S|E,z,x_{-K},y_{-K}]\end{split}\end{align}
The first term $\Pr[\vee_{i\in S} \neg G_i|E]$ is at most $O(n^{-\delta+\epsilon})$ by \Cref{intro:eq3} and a union bound. To analyze the second term, we will define a two-player game $\tilde{\mc G}$ such that the probability of winning the coordinates in $S$ in the $m$-fold parallel repetition of $\tilde{\mc G}$ is exactly the second term in the R.H.S. of \Cref{intro:eq2}. For now, we will define a game $\tilde{\mc G}_i$ for each $i\in[n]$. Although these games can be different for different $i\in[n]$, there are only finitely many possibilities for $\tilde{\mc G}_i$ and we simply restrict our attention to the game that appears in most number of coordinates. The query distribution for $\tilde{\mc G}_i$ is the uniform distribution on $\tilde{\mc S}=\{(0,0),(1,0),(1,1)\}$. The predicate $\tilde{V}_i$ is $\tilde{V}_i(x,y,a,b):= V(x,y,0,a,b,h(z)_i)\wedge (x=1 \implies a\in \mathcal{A}_i )\land (y=1 \implies b\in \mathcal{B}_i $)  (\Cref{def:restr_2p_game}). Note that the value of $\tilde{\mc G}_i$ is less than one (\Cref{claim:restr_2p_game_value}). To see this, given any strategy $(\tilde{f}, \tilde{g})$ for $\tilde{\mc G}$ with value one, there is a strategy $\tilde{h}$ for Charlie on input 1 such $(\tilde{f}, \tilde{g}, \tilde{h})$ succeeds on $(1,1,1)$, since when Alice answers from $\mc A_i$ on input 1 and Bob from $\mc B_i$ on input 1, there is an answer Charlie can give on input 1 to satisfy the predicate. Define the strategy $\tilde{h}$ to output $h(z)_i$ when Charlie receives 0. We know that $(\tilde{f},\tilde{g},\tilde{h})$ must fail on one of the remaining points in $\{(0,0,0),(0,1,0),(1,0,0)\}$. This implies that $(\tilde{f}, \tilde{g})$ falsifies the predicate $\tilde{V}_i$ at the corresponding point in $\tilde{\mc S}$. We use two-player parallel repetition techniques and show that for a random $S\subseteq K$ of size $m^{\epsilon}$, the probability of winning the $m$-fold parallel repetition of the two-player game in coordinates in $S$ is at most $\exp(-\Omega(m^{\epsilon}))$ (see \Cref{lemma:fixed_z_good_event_part1}).  We remark that we are able to apply two-player parallel repetition techniques even though the measure of $E$ could be smaller than $\exp(-\Omega(|S|))$, and this is because the set $S$ was chosen randomly. Thus, the second term in the R.H.S. of \Cref{intro:eq2} is bounded by $\exp(-\Omega(m^{\epsilon}))=\exp(-\Omega(n^{\epsilon}))$. This completes the proof overview.  


\subsection{Playerwise Connected Games}\label{sec:playerconn_outline}

We present an overview of the proof of \Cref{thm:playerwise_conn_par_rep} which shows an inverse polynomial bound on the parallel repetition value of all playerwise connected games. The details are in \Cref{sec:playerconn_games}. We focus on the case of the uniform distribution $Q$ over $\mc S:=\{(0,0,0),(1,0,0),(0,1,0),(0,0,1),(1,1,1)\}$. Let $P=Q^{\otimes n}$. For a random variable $W$, we use $P_W$ to denote the distribution of $W$ where the probability space is $P$. We use $(X,Y,Z)$ to denote inputs of the three players for the game $\mc G^{\otimes n}$.

Our proof builds on the framework of the Parallel Repetition Theorem from \cite{Raz98,Hol09,DHVY17}. We now describe this framework. Let $\mc G$ be a game whose query distribution is $Q$ and whose value is less than one. Consider its $n$-fold parallel repetition. 
Using an inductive argument, it suffices to show that for every large {\it product event} $E=E^1\times E^2\times E^3$ across the players' inputs to the $n$-fold game $\mc G^{\otimes n}$, when the inputs are drawn from the distribution $P|E$, there exists some \emph{hard coordinate} $i \in [n]$, meaning that the probability of winning the game in the $i$-th coordinate is $1-\epsilon$ for some constant $\epsilon>0$.
Since the event $E$ has large measure, it cannot reveal too much information about too many coordinates, and hence the distribution of the marginal of $P|E$ on the $i$-th coordinate is similar to the original distribution $Q$ for most $i\in[n]$. It then suffices to show a way for the players to approximately embed the inputs they receive for the original game $\mc G$, into the $i$-th coordinate of the inputs to the $n$-fold game drawn according to the distribution $P|E$. In order to do this, they need to be able to sample the remaining $n-1$ coordinates of the inputs according to the correct distribution. To do this, as in \cite{Raz98,Hol09,DHVY17}, 
we define a dependency breaking random variable $R$ as follows. The random variable $R$ (Definition \ref{defn:dep_break_rv}) for each copy $i\in [n]$ of the game, independently, chooses two players uniformly at random and samples the inputs to those players according to the distribution induced by $P|E$.  Let $(x_i,y_i,z_i)\in \mathrm{supp}(Q)$. If the players can jointly sample from $P_{R_{-i}|{E,x_i,y_i,z_i}}$, then since $R$ breaks the correlation between the player's inputs, any player can independently sample the rest of their own input given $R_{-i}$. Note that each player only knows one of $x_i,y_i$ and $z_i$ and it is not evident how the players can jointly sample from $P_{R_{-i}|{E,x_i,y_i,z_i}}$. Ideally, one would like to show that  $P_{R_{-i}|{E,x_i,y_i,z_i}}$ is close to some global distribution for all $(x_i,y_i,z_i)\in \mathrm{supp}(Q)$; in such a  case, the global distribution would be $P_{R_{-i}|E}$. We denote this by $P_{R_{-i}|E}\approx P_{R_{-i}|{E,x_i,y_i,z_i}}$. This would mean that the players only need to sample from a global distribution $P_{R_{-i}|E}$ and this can be done using shared randomness.  Prior works \cite{Raz98,Hol09,DHVY17} showed that $P_{R_{-i}|{E,x_i,y_i,z_i}}\approx P_{R_{-i}|{E,x_i,y_i}}$, but it is not clear if this suffices to prove the desired result.

In our work, we show that for the query distribution $Q$, we have $P_{R_{-i}|E}\approx P_{R_{-i}|{E,x_i,y_i,z_i}}$ for all $(x_i,y_i,z_i)\in \mathrm{supp}(Q)$ (\Cref{lemma:ri_same_distr}). (We prove a similar result for all playerwise connected games.) By choosing parameters appropriately and by Bayes' Rule, it will suffice to show that
\begin{equation}\label{intro_five_point_eq0} P_{X_i,Y_i,Z_i|{r_{-i},E}}\approx P_{X_i,Y_i,Z_i|E}\quad\text{ for most }r_{-i}\sim P_{R_{-i}|E}.\end{equation}
The key idea we use is to modify the distribution of $P$ in the $i$-th coordinate as follows. Let $\Gamma$ be a product distribution over $\{0,1\}^3$ such that the marginal on each player's input agrees with that of $Q$. The idea is to consider the distribution $P_{-i}\Gamma$, which is a  product of $n$ independent distributions, where the distribution in the $i$-th coordinate is $\Gamma$ and the distribution in every other coordinate is $Q$. We can recover the original distribution $P$ from $P_{-i}\Gamma$ by conditioning on some event that depends on the $i$-th coordinate (and possibly on some additional randomness). 

Note that  $(P_{-i}\Gamma)_{X,Y,Z|r_{-i}}$ is a product distribution across the inputs of the players, since $R$ is a correlation breaking random variable and since the distribution of $P_{-i}\Gamma$ in the $i$-th coordinate was a product distribution to begin with. Since $E=E^1\times E^2\times E^3$ is a product event, the $i$-th coordinate of $(P_{-i}\Gamma)_{X,Y,Z|r_{-i},E}$ has a product distribution across the players. Finally, since the marginal of $\Gamma$ on any player's input agrees with that of $Q$, we have 
\begin{align}\label{intro_five_point_eq1}\begin{split}  (P_{-i}\Gamma)_{X_i,Y_i,Z_i|{r_{-i},E} }&=  (P_{-i}\Gamma)_{X_i|{r_{-i},E^1}}\times  
(P_{-i}\Gamma)_{Y_i|{r_{-i},E^2}}\times (P_{-i}\Gamma)_{Z_i|{r_{-i},E^3}} \\
&= P_{X_i|{r_{-i},E^1}}\times 
P_{Y_i|{r_{-i},E^2}}\times P_{Z_i|{r_{-i},E^3}} \quad\text{ for all }r_{-i}.\\
\end{split}\end{align} 
To study the distribution $P_{X_i|r_{-i},E^1}$, we will study the distribution $P_{X_i,Y_i,Z_i|{r_{-i},E^1}}$ and show that it is close to $Q$ (\Cref{lemma:Xi_same_distr_cond_ri_only_Ej}). The support of the latter distribution is contained in $\mc S$. Observe that the distribution $P_{Y_i,Z_i|{r_{-i},E^1,X_i=x_i}}$ is exactly the uniform distribution over $\{(y_i,z_i)\in \{0,1\}^2 :(x_i,y_i,z_i)\in \mathrm{supp}(Q)\}$ for all $x_i\in \{0,1\}$.\footnote{To see this, note that the distribution $P_{XYZ|r_{-i}}$ is a product distribution across coordinates where the marginal in the $i$-th coordinate is simply the uniform distribution over $\mc S$. Once we condition on $X_i=x_i$ for any $x_i\in \{0,1\}$, the distribution $P_{Y,Z|{r_{-i},X_i=x_i}}$ is still a product distribution across coordinates and in the $i$-th coordinate is exactly the uniform distribution over $\{(y_i,z_i)\in \{0,1\}^2 :(x_i,y_i,z_i)\in \mc S\}$. If we further condition on $E^1$, it only affects the distribution of inputs in coordinates other than $i$.} This implies that the  probabilities assigned by $P_{X_i,Y_i,Z_i|{r_{-i},E^1}}$ to all points in $\mc S$ of the form $(1,*,*)$ are identical, and similarly, the probabilities assigned to all points in $\mc S$ of the form $(0,*,*)$ are identical (\Cref{lemma:same_distr_over_planes}). To conclude that the probabilities assigned to $(0,0,0)$ and $(1,0,0)$ are close for most $r_{-i}\sim P_{R_{-i}|E^1}$, we use techniques similar to~\cite{Raz98,Hol09,DHVY17} (\Cref{lemma:same_distr_over_edges}). We thus show that the probabilities assigned to all points in $\mc S$ are similar and hence the distribution of $P_{X_i,Y_i,Z_i|{r_{-i},E^1}}$  is close to $Q$ for most $r_{-i}\sim P_{R_{-i}|E^1}$. This along with a similar argument for the second and third terms in the R.H.S. of \Cref{intro_five_point_eq1} (and the fact that the marginals of $\Gamma$ on any player agree with that of $Q$) implies that
\begin{align}
	\label{intro_five_point_eq3} 
	\begin{split}	 P_{X_i|{r_{-i},E^1}}& \approx  (P_{-i}\Gamma)_{X_i}	\quad \text{for most }r_{-i}\sim P_{R_{-i}|E^1} \\	P_{Y_i|{r_{-i},E^2}}& \approx  (P_{-i}\Gamma)_{Y_i}	\quad \text{for most }r_{-i}\sim P_{R_{-i}|E^2} \\P_{Z_i|{r_{-i},E^3}}& \approx  (P_{-i}\Gamma)_{Z_i}	\quad \text{for most }r_{-i}\sim P_{R_{-i}|E^3} \\
	\end{split}
\end{align}
Let us pretend for now that \Cref{intro_five_point_eq3} actually holds for most $r_{-i}\sim P_{R_{-i}|E}$. If so, this, along with \Cref{intro_five_point_eq1} would imply that 
\[ (P_{-i}\Gamma)_{X_i,Y_i,Z_i|{r_{-i},E}} \approx (P_{-i}\Gamma)_{X_i,Y_i,Z_i}\quad \text{ for most } r_{-i}\sim P_{R_{-i}|E}.\]
We now condition both sides of the above equation on an event $T_i$ such $\Gamma|T_i=Q$. Note that $T_i$ depends only on the inputs in coordinate $i$ and some additional shared randomness. (See \Cref{sec:fivepoint_gamma} for details.) This implies that
\[ P_{X_i, Y_i ,Z_i|{r_{-i},E}} \approx P_{{X}_i,
{Y}_i, {Z}_i}\quad \text{ for most } r_{-i}\sim P_{R_{-i}|E}.\]
This along with the fact that the distribution of $P_{X_i,Y_i,Z_i|E}$ is close to $Q$ for most coordinates $i\in[n]$ completes the proof of \Cref{intro_five_point_eq0}, under the {\it incorrect assumption} that the distribution of $r_{-i}$ in \Cref{intro_five_point_eq3} was $P_{R_{-i}|E}$. To fix this, we use the property that for any random variable $G$, we have $P[G=g|E]\le P[G=g|E^1]\cdot \frac{P[E^1]}{P[E]}\le P[G=g|E^1]\cdot \frac{1}{P[E]}$. This allows us to connect probabilities over $P|E$ and probabilities over $P|E^1,P|E^2$ and $P|E^3$ (\Cref{lemma:same_distr_ith_prod}). This is the place where we incur a loss of $\frac{1}{P(E)}$ and as a result, our bound only holds for polynomially large events.

\section{Preliminaries}

Let $\N = \set{1,2,\dots}$ be the set of all natural numbers. For each $n\in \N$, we use $[n]$ to denote the set $\set{1,2,\dots,n}$.

We will mostly follow \cite{Hol09, DHVY17, HR20} for notation.

\subsection{Probability Distributions}
We will use calligraphic letters to denote sets, capital letters to denote random variables and small letters to denote values.

Let $P$ be a distribution (with the underlying \emph{finite} set clear from context). 
For a random variable $X$, we use $P_X$ to denote the distribution of $X$, that is, $P_X(x) = P(X=x)$.
For random variables $X$ and $Y$, we use $P_{XY}$ to denote the joint distribution of $X$ and $Y$.
For an event $E$ with $P(E)>0$, we use $P_{X|E}$ to denote the distribution of $X$ conditioned on the event $E$, given by \[P_{X|E}(x) = \frac{P\brac{X=x\land E}}{P(E)}. \]
Suppose $R$ is a random variable, and $r$ is such that $P_R(r)>0$.
We will frequently use the shorthand $P_{X|r}$ to denote the distribution $P_{X|R=r}$.

Let $P_X$ and $Q_X$ be distributions over a set $\mc X$.
The $L^1$-distance (or $\ell_1$-norm) between $P_X$ and $Q_X$ is defined as $\lonorm{P_X}{Q_X} = \sum_{x\in \mc X}\norm{P_X(x)-Q_X(x)}$.

We will also be using the following facts:
\begin{fact} \label{fact:chernoff} (Chernoff Bounds, see \cite{MU05} for reference)
	Let $X_1,\dots, X_n \in \set{0,1}$ be independent random variables each with mean $\mu$, and let $X = \sum_{i=1}^n X_i$. Then, for all $\delta \in (0,1)$, it holds that 
	\[\Prob{X \leq (1-\delta)\mu n} \leq e^{-\frac{\delta^2 \mu n}{2}},\]
	\[\Prob{X \geq (1+\delta)\mu n} \leq e^{-\frac{\delta^2 \mu n}{3}}.\]
	\[
	\Prob{X - \mu n \ge \delta n
	} \le e^{-2\delta^2 n}.
	\]
\end{fact}

\begin{fact}\label{fact:norm_condition}
Let $P_X$ and $Q_X$ be probability distributions over a set $\mc X$, and let $W\subseteq \mc X$ be an event such that $P_X(W), Q_X(W)>0$.
Then,
\[\lonorm{P_{X|W}}{Q_{X|W}} \leq \frac{2}{Q_X(W)}\cdot \lonorm{P_X}{Q_X}\]
\end{fact}
\begin{proof}
	See Appendix~\ref{appendix:probab_facts}.
\end{proof}

\begin{fact} \label{fact:conditional-markov}
	Let $P$ be a probability distribution, and let $X$ be a random variable over a set $\mc X$, and let $E$ be any event with $P(E)>0$.
	Let $\alpha>0$ be arbitrary, and let $\mc T = \set{x\in \mc X : P(E\st x) \geq \alpha\cdot P(E)}.$
	Then, it holds that
	\[ P\brac{X\in \mc T \st E} > 1-\alpha.\]
\end{fact}
\begin{proof}
	See Appendix~\ref{appendix:probab_facts}.
\end{proof}

\begin{fact} \label{fact:averaging}
	Let $P$ be a probability distribution, and let $X$ be a random variable over a set $\mc X$, and let $E$ be any event.
	Then, there exists $x\in \mc X$ such that $P\brac{E\st X=x} \geq P(E)$.
\end{fact}


\subsection{Multiplayer Games}
\begin{definition} (Multiplayer Game)
	A $k$-player game $\mc G$ is a tuple $\mc G = (\mc X, \mc A, Q, V)$, where the question set $\mc X = \mc X^1 \times\dots\times \mc X^k$, and the answer set $\mc A= \mc A^1 \times\dots\times \mc A^k$ are finite sets, $Q$ is a probability distribution over $\mc X$, and $V:\mc X\times \mc A\to\set{0,1}$ is a predicate.
\end{definition}

\begin{definition} (Game Value)
	Let $\mc G = (\mc X, \mc A, Q, V)$ be a $k$-player game.
	
	For a sequence $\brac{f^j:\mc X^j\to \mc A^j}_{j\in[k]}$ of functions, define the function $f = f^1\times\dots\times f^k : \mc X \to \mc A$ by $f\brac{x^1,\dots,x^k} = \brac{f^1(x^1), \dots, f^k(x^k)}$. 
	We use the term \emph{product functions} to denote functions $f$ defined in this manner.
	
	The value $\val(\mc G)$ of the game $\mc G$ is defined as 
	\[\val(\mc G) = \max_{f = f^1\times \dots\times f^k}\ \Pr_{X\sim Q}\sqbrac{V(X,f(X))=1},\]
	where the maximum is over all product functions $f = f^1\times \dots\times f^k$. The functions $(f^j)_{j\in[k]}$ are called player strategies.
\end{definition}

\begin{fact}\label{fact:rand_no_help}
	The value of the game is unchanged even if we allow the player strategies to be randomized, that is, we allow the strategies to depend on some additional shared and private randomness. 
\end{fact}

\begin{definition} (Parallel Repetition of a game)
	Let $\mc G = (\mc X, \mc A, Q, V)$ be a $k$-player game. We define its $n$-fold repetition as $\mc G^{\otimes n} = (\mc X^{\otimes n}, \mc A^{\otimes n}, P, V^{\otimes n})$.
	The sets $\mc X^{\otimes n}$ and $\mc A^{\otimes n}$ are defined to be the $n$-fold product of the sets $\mc X$ and $\mc A$ with themselves respectively.
	The distribution $P$ is the $n$-fold product of the distribution $Q$ with itself, that is, $P(x) = \prod_{i=1}^n Q(x_i)$.
	The predicate $V^{\otimes n}$ is defined as $V^{\otimes n}(x, a) = \bigwedge_{i=1}^n V(x_i, a_i)$.
	
	Note that we use the notation $\mc X^{\otimes n}$ instead of the standard notation $\mc X^n$ so as to avoid confusion with the sets $\mc X^1, \dots, \mc X^k$.
\end{definition}

Following the notation in \cite{DHVY17}, we use subscripts to denote the coordinates in the parallel repetition, and superscripts to denote the players. 
For example, for $i\in[n]$ and $j\in[k]$, we will use $x_i^j$ to refer to the question to the $j\ts{th}$ player in the $i\ts{th}$ repetition of the game.
Similarly, $x_i$ will refer to the vector of questions to the $k$ players in the $i\ts{th}$ repetition, and $x^j$ will refer to the vector of questions received by the $j\ts{th}$ player over all repetitions.
We use $x^{-j}$ to refer to the questions to all players except the $j\ts{th}$ player, and use $x_{-i}$ to refer to the questions in all coordinates except the $i\ts{th}$ coordinate.

When we are dealing with 3 player games, we will not be using superscripts to refer to different players, and rather use the notation $\mc G = \brac{\mc X\times \mc Y\times \mc Z,\ \mc A\times \mc B\times \mc C, Q, V}$.
That is, we use $\mc X, \mc Y, \mc Z$ in place of $\mc X^1, \mc X^2, \mc X^3$ and $\mc A, \mc B, \mc C$ in place of $\mc A^1, \mc A^2, \mc A^3$.

We will use the notation $A\lesssim B$ (resp. $\gtrsim$) to mean that $A\leq c\cdot B$ (resp. $A\geq c\cdot B$) for some constant $c>0$.
For our purposes, we will allow the constant to depend on the size of the (initial) game being considered, but not on the number of repetitions.


\subsection{Playerwise Connected Games}

We will particularly be interested in a special class of games, which we refer to as playerwise connected games.
Before we define this class, we recall the following definition:

\begin{definition}\label{defn:conn_graph} ($(k-1)$-connection graph \cite{DHVY17})
	Let $\mc G = (\mc X, \mc A, Q, V)$ be a $k$-player game, and	 let $\mc S \subseteq \mc X$ be the support of $Q$.
	We define its (undirected) $(k-1)$-connection graph $H_\mc G$ as follows.
	The vertex set of $H_\mc G$ is $\mc S$, and there is an edge between $x,y\in \mc S$ if and only if they differ in the question to exactly one of the players.
	That is, $\set{x,y}$ is an edge if and only if there exists $j\in[k]$ such that $x^{-j}=y^{-j}$ and $x^j\not=y^j$.
	
	We say that a game $\mc G$ is \emph{connected} if the graph $\mc H_\mc G$ is connected.
\end{definition}

We will define a game to be \emph{playerwise connected} if the projection of the above graph with respect to each of the players is connected.
This is formally defined as follows:

\begin{definition}\label{defn:conn_player_graph} (Playerwise Connected Game)
Let $\mc G = (\mc X, \mc A, Q, V)$ be a $k$-player game, and	 let $\mc S \subseteq \mc X$ be the support of $Q$.
We assume that for all $j\in [k]$, and for all $x^j\in \mc X^j$, the $j^\textsuperscript{th}$ player is given the question $x^j$ with positive probability.

For every $j\in[k]$, we define the graph $H_\mc G^j$ as the graph with vertex set $\mc X^j$, with an edge between $x^j, y^j\in \mc X^j$ if and only if there exists $x^{-j}\in \mc X^{-j}$ such that both $(x^{-j},x^j)\in \mc S$ and $(x^{-j},y^j)\in \mc S$. 
We say that the game $\mc G$ is \emph{playerwise connected} if $H_\mc G^j$ is connected for each $j\in[k]$.

Note that the assumption on $\mc G$ in the above definition is without loss of generality, because we can simply remove each $x^j$ which occurs with zero probability, without affecting the game in any meaningful way.
\end{definition}


\subsection{Previously Known Results}
\label{sec:previous_bounds}
We state the known results on parallel repetition that will be useful for us.

\begin{theorem}\label{thm:2p_par_rep}(Parallel Repetition for 2-Player Games \cite{Raz98})
	Let $\mc G = \brac{\mc X\times \mc Y, \mc A\times \mc B, Q, V}$ be a 2-player game such that $\val(\mc G)<1$.
	Then, there exists a constant $c = c(\mc G)>0$ such that for every $n\in \N$, it holds that $\val(\mc G^{\otimes n})\leq 2^{-cn}$.
	
\end{theorem}

\begin{theorem}\label{thm:conn_par_rep} (Parallel Repetition for Connected Games \cite{DHVY17})
	Let $\mc G$ be a \emph{connected} $k$-player game (see Definition~\ref{defn:conn_graph}) such that ${\val(\mc G)<1}$.
	Then, there exists a constant $c = c(\mc G)>0$ such that for every $n\in \N$, it holds that $\val(\mc G^{\otimes n}) \leq 2^{-cn}$.
\end{theorem}

\begin{theorem}\label{thm:ghz_par_rep} (Parallel Repetition for The GHZ Game \cite{HR20, GHMRZ21})
	Let $\mc G = (\mc X\times \mc Y\times \mc Z,\ \mc A \times \mc B \times \mc C,\ Q, V)$ be a 3-player game with $\mc X = \mc Y = \mc Z = \set{0,1}$, and $Q$ the uniform distribution over $\mc S = \set{(0,0,0), (1,1, 0), (1,0,1), (0,1,1)} = \set{(x,y,z)\in \set{0,1}^3 : x\oplus y \oplus z = 0 }$, and such that $\val(\mc G)<1$.
	Then, there exists a constant $c = c(\mc G)>0$, such that for every $n\in \N$, it holds that $\val(\mc G^{\otimes n})\leq n^{-c}$. 
\end{theorem}


\subsection{Some Results on Multiplayer Games}\label{sec:simplifications}

\subsubsection{Restriction to Uniform Distributions} \label{sec:wlog_unif_dist}

We state a lemma from \cite{FV02}, which shows that it suffices to prove parallel repetition in the case when the game's distribution is the uniform distribution over its support.
For the sake of completeness, we also include a short proof.

\begin{lemma}\label{lemma:wlog_unif_dist}
	Let $\mc G = (\mc X, \mc A, Q, V)$ be a $k$-player game such that $\val(\mc G)<1$, and	 let $\mc S \subseteq \mc X$ be the support of $Q$.
	Let $\tilde{\mc G} = (\mc X, \mc A, U, V)$,  where $U$ is the uniform distribution over $\mc S$.
	Let $v:\N\cup\set{0}\to[0,1]$ be the function defined by $v(n) = \val(\tilde{\mc G}^{\otimes n})$, for every $n\in \N\cup\set{0}$, with the convention $v(0)=1$.
	Then, 
	\begin{enumerate}[label=(\alph*)]
		\item $v(1) = \val(\tilde{\mc G}) < 1$.
		\item There exists a constant $\beta>0$ such that $\val(\mc G^{\otimes n}) \leq 2v(\lfloor\beta n\rfloor)$ for all $n\in \N$.
	\end{enumerate}
\end{lemma}
\begin{proof}
	See Appendix~\ref{appendix:multiplayer_game_results}.
\end{proof}

\subsubsection{A Restriction on Predicates} \label{sec:pred_for_deter_inp}

We show that to prove parallel repetition, it suffices to assume that the game has the following property:
If some input $y^j$ for the $j\ts{th}$ player completely determines the input $y$ to all the players, then on input $y$, the game's predicate does not depend on the answer $a^j$ given by the $j\ts{th}$ player.

A recursive application of the next lemma shows that we can assume the aforementioned property.

\begin{lemma}\label{lemma:pred_for_deter_inp}
	Let $\mc G = (\mc X, \mc A, Q, V)$ be a $k$-player game.
	Suppose $y\in \mc X, j\in [k]$ are such that $y$ is the unique input with $Q(y)>0$ that has $y^j$ as the input to the $j\ts{th}$ player.
	Then, there exists a predicate $V'$ such that the game  $\mc G' = (\mc X, \mc A, Q, V')$ satisfies:
	\begin{enumerate}[label=(\alph*)]
		\item For every $a, b\in \mc A$ with $a^{-j} = b^{-j}$, it holds that $V'(y, a) = V'(y, b)$. 
		\item For every $n\in \N$, it holds that $\val(\mc G^{\otimes n})\leq \val(\mc G'^{\otimes n})$.
		\item $\val(\mc G') = \val(\mc G)$.
	\end{enumerate}
\end{lemma}
\begin{proof}
	See Appendix~\ref{appendix:multiplayer_game_results}.
\end{proof}

\subsubsection{An Inductive Parallel Repetition Criterion} \label{sec:}

We state a parallel repetition criterion from \cite{Raz98, HR20}.
For the sake of completeness, we also include a proof.

\begin{definition}
	Let $\mc G = (\mc X, \mc A, Q, V)$ be a $k$-player game, and let $Q'$ be some distribution over $\mc X$. We define $\mc G\st Q'$ to be the game $\mc G\st Q' = (\mc X, \mc A, Q', V)$.
\end{definition}

\begin{definition}
	Let $\mc G = (\mc X, \mc A, Q, V)$ be a $k$-player game, and $\mc G^{\otimes n} = (\mc X^{\otimes n}, \mc A^{\otimes n}, P, V^{\otimes n})$ be its $n$-fold repetition.
	For each $i\in [n]$, we define the value of the $i\textsuperscript{th}$ coordinate of $\mc G^{\otimes n}$, denoted by $\val_i(\mc G^{\otimes n})$, to be the value of the game $(\mc X^{\otimes n}, \mc A^{\otimes n}, P, V')$, where $V'(x, a) = V(x_i, a_i)$.
\end{definition}

\begin{lemma}\label{lemma:prod_set_hard_coor}
	Let $\mc G = (\mc X, \mc A, Q, V)$ be a $k$-player game, and $\mc G^{\otimes n} = (\mc X^{\otimes n}, \mc A^{\otimes n}, P, V^{\otimes n})$ be its $n$-fold repetition.
	Suppose that there exists a constant $\epsilon >0$, and a non-increasing function $\rho:\N\to[0,1]$ such that $\rho(n)\geq 2^{-O(n)}$, and for every $n\in \N$, and every product event $E=E^1\times \dots \times E^k \subseteq (\mc X^1)^{\otimes n}\times\dots\times (\mc X^k)^{\otimes n} = \mc X^{\otimes n}$ with $P(E)\geq \rho(n)$, there exists a coordinate $i\in [n]$ such that $\val_i(\mc G^{\otimes n}\st(P|E)) \leq 1-\epsilon$.
	Then, there exists a constant $c>0$ such that $\val(\mc G^{\otimes n}) \leq \rho(n)^c$ for every $n\in \N$.
\end{lemma}
\begin{proof}
	See Appendix~\ref{appendix:multiplayer_game_results}.
\end{proof}

\section{The Anti-Correlation Game}\label{sec:anti_corr_game}

In this section, we will focus on the following game.

\begin{definition}[The Anti-Correlation Game]\label{def:anti_corr_game}
  The \emph{anti-correlation game}, which we denote as $\mc G = (\set{0,1}^3,\ \set{0,1}^3, Q, V)$,
  is a 3-player game in which the query distribution $Q$ is uniform over the set $\set{(0,1,1), (1, 0, 1), (1, 1, 0)}$ of strings of hamming-weight $2$, and the win predicate
  $V : \set{0,1}^3 \times \set{0,1}^3 \to \{0,1\}$ is defined so that
  $V\brac{(x,y,z), (a,b,c)} = 1$ if and only if
  $\langle (x,y,z), (a,b,c) \rangle = xa+yb+zc= 1$.

	In words, a random pair of players receive $1$, and these players must produce different bits.

 	If $V\brac{(x,y,z), (a,b,c)} = 1$, we say that $(a,b,c)$ \text{wins on} $(x,y,z)$.
\end{definition}

We will denote this game by  $\cG$, and denote its query distribution by $Q$ (i.e. the uniform distribution on $\big \{(0,1,1), (1, 0, 1), (1, 1, 0) \big \}$.
Observe that the value of this game is $2/3$.
We will show that parallel repetition exponentially decays the value of this game.

\begin{theorem}\label{thm:anti_corr_par_rep} Let $\mc G$ be the anti-correlation game as in \Cref{def:anti_corr_game}. Then, there exists a constant $c>0$ such that for every $n\in \N$, it holds that $\val(\mc G^{\otimes n})\le \exp(-c\cdot n)$.
\end{theorem}

We will let $Q^{\otimes n}$ denote the query distribution of $\cG^{\otimes n}$, with $X$, $Y$, and $Z$ respectively being random variables corresponding to the first, second, and third player inputs in the game $\mc G^{\otimes n}$.

\subsection{Winning Strategies Imply Partially Constant Winning Strategies}
We first claim that: 
\begin{proposition}
  \label{prop:f-nearly-constant}
  Suppose the value of $\cG^{\otimes n}$ is $\alpha$ for
    $\alpha^4 \ge 32e^{-n / 72}$. Then there exist $n' \ge n / 6$, a string $a \in \{0,1\}^{n'}$, and functions 
  $g', h' : \{0,1\}^{n'} \to \{0,1\}^{n'}$ such that with at least $\alpha^4 / 16$ probability when sampling $(x, y, z) \gets Q^{\otimes n'}$, 
  $(a_i, g'(y)_i, h'(z)_i)$ wins on $(x_i, y_i, z_i)$  for all but $5 \ln(1
  / \alpha) + 5$ values of $i \in [n']$.
\end{proposition}
\begin{proof}
Let $f, g, h : \{0,1\}^n \to \{0,1\}^n$ be a strategy that wins
$\cG^{\otimes n}$ with probability $\alpha$.  


We now define a probability space that induces a correlated
distribution on $(X, \tilde{X})$, with both $X$ and $\tilde{X}$ individually
distributed according to $Q_X^{\otimes n}$.

\paragraph{The probability space $\cP$.}
Let $\cP$ be obtained by the following process:
\begin{enumerate}
\item Sample $Y \gets Q_{Y}^{\otimes n}$.
\item Independently sample random variables $(X,Z)$ and $(\tilde{X}, \tilde{Z})$
  such that
  $\cP \big [X = x, Z = z | Y = y \big ] = \cP \big [ \tilde{X} = x, \tilde{Z} =
  z | Y = y \big ] = Q^{\otimes n}_{X,Z | Y}(x,z | y)$.
\end{enumerate}

\begin{claim}
  \label{claim:P-probabilities}
  In $\cP$, $\big \{ (X_i, \tilde{X}_i) \big \}$ are i.i.d., and
  for each $i$, we have
  \begin{align*}
      \cP[X_i = x_i, \tilde{X}_i = \tilde{x}_i] = \begin{cases}
      1/6 & \text{if $(x_i, \tilde{x}_i) \in \big \{(0, 0), (0, 1), (1, 0) \big \}$} \\
      1/2 & \text{if $(x_i, \tilde{x}_i) = (1, 1)$.}
    \end{cases}
  \end{align*}
\end{claim}
\begin{proof}
  If $x_i = 0$ or $\tilde{x}_i = 0$, the only way we can have $X_i = x_i$ and
  $\tilde{X}_i = \tilde{x}_i$ is if $Y_i = 1$.  Thus for such $(x_i,
  \tilde{x}_i)$, we have
  \begin{align*}
    \cP[X_i = x_i, \tilde{X}_i = \tilde{x}_i] &= \cP[Y_i = 1] \cdot \cP[X_i =
                                                x_i, \tilde{X}_i = \tilde{x}_i | Y_i = 1] \\
                                              &= \frac{2}{3} \cdot \frac{1}{4} = \frac{1}{6}.
  \end{align*}

  $\cP[X_i = \tilde{X}_i = 1 | Y_i = 1] = \frac{1}{4}$ as above, but
  additionally $\cP[X_i = \tilde{X}_i = 1 | Y_i = 0] = 1$.  Thus in total
  $\cP[X_i = \tilde{X}_i = 1] = \frac{2}{3} \cdot \frac{1}{4} + \frac{1}{3} = \frac{1}{2}$.
\end{proof}

\begin{definition}[Approximate Consistency]
  \label{def:approx-consistency}
  For any $x, \tilde{x} \in \{0,1\}^n$ and $f : \{0,1\}^n \to \{0,1\}^n$, we say
  that \textdef{$f$ is approximately consistent on $x$ and $\tilde{x}$} if for all
  but $5 \ln(1 / \alpha) + 5$ values of $i \in [n]$, it holds that if $x_i = \tilde{x}_i = 1$, then
  $f(x)_i = f(\tilde{x})_i$.
\end{definition}

  \begin{claim}
    \label{claim:P-good}
    It holds with probability at least $\alpha^2 / 2$ over $(x, \tilde{x}) \gets \cP_{X, \tilde{X}}$ that:
    \begin{itemize}
    \item $f$ is approximately consistent on $x$ and $\tilde{x}$, and
    \item $\cP\Big [\big (f(\tilde{X}), g(Y), h(\tilde{Z}) \big ) \text{
        wins on $(\tilde{X}, Y, \tilde{Z})$}\ \Big | \tilde{X} = \tilde{x} \Big ] \ge \alpha^2/4$.
    \end{itemize}
  \end{claim}
\begin{proof}
  Let $W$ denote the event that $\big (f(X), g(Y), h(Z))$ wins on $(X, Y, Z)$,
  and let $\tilde{W}$ denote the event that
  $\big (f(\tilde{X}), g(Y), h(\tilde{Z}) \big )$ wins on
  $(\tilde{X}, Y, \tilde{Z})$.

  We first have
  \begin{align}
    \cP[W, \tilde{W}] &= \E_{y \gets \cP_Y} \big [ \cP[W, \tilde{W} | Y = y] \big ] \notag \\ &= \E_{y \gets \cP_Y} \big [ \cP[W | Y = y]^2 \big ] && \text{($(X,Z)$ and $(\tilde{X},\tilde{Z})$ conditionally i.i.d. given $Y$)} \notag
    \\& \ge \E_{y \gets \cP_Y} \big [ \cP[W | Y = y] \big ]^2   \notag 
    \\& = \cP[W]^2  \notag
    \\&= \alpha^2 \label{eq:prob-w-tilde-w}.
  \end{align}

  Then, we have
  \begingroup \addtolength{\jot}{0.75em}
  \begin{align*}
    & \cP \left [
    \begin{array}{l}
      \cP[W , \tilde{W} | X, \tilde{X}] \ge \alpha^2/4 \text{ and } \\
      \cP[\tilde{W} | \tilde{X}] \ge \alpha^2/4
    \end{array}
    \right ] \\
    & \ge \cP \left [
    \begin{array}{l}
      \cP[W , \tilde{W} | X, \tilde{X}] \ge \alpha^2/4 \text{ and }\\
      \cP[W, \tilde{W} | \tilde{X}] \ge \alpha^2/4
    \end{array}
    \right ] \\
    & \ge \cP [ W , \tilde{W} ] \cdot    \cP \left [
      \begin{array}{l}
        \cP[W , \tilde{W} | X,
      \tilde{X}] \ge \alpha^2/4 \text{ and }\\ \cP[W , \tilde{W} | \tilde{X}] \ge
        \alpha^2/4
      \end{array}
    \middle | W , \tilde{W} \right ] \\
    & \ge \alpha^2 \cdot    \cP \left [
      \begin{array}{l}
        \cP[W , \tilde{W} | X,
        \tilde{X}] \ge \alpha^2/4 \text{ and }\\ \cP[W , \tilde{W} | \tilde{X}] \ge
        \alpha^2/4
      \end{array}
    \middle | W , \tilde{W} \right ] && \text{(\cref{eq:prob-w-tilde-w})} \\
    & = \alpha^2 \cdot    \left ( 1 - \cP \left  [
      \begin{array}{l} \cP[W , \tilde{W} | X,
        \tilde{X}] < \alpha^2/4 \ \text{ or } \\
        \cP \big [\cP[W , \tilde{W} | \tilde{X}] <
        \alpha^2/4
      \end{array}
    \middle | W , \tilde{W} \right ] \right ) \\
    & \ge \alpha^2 \cdot  \left (1 -2 \cdot \frac{1}{4} \right ) &&
                                           \text{(\cref{eq:prob-w-tilde-w}, \cref{fact:conditional-markov}, and a union bound)} \\
     & = \alpha^2 / 2.
  \end{align*}
  \endgroup
  
  Finally, we show that whenever $\cP[W,\tilde{W} | X, \tilde{X}]\ge \alpha^2 / 4$,
  $f$ is approximately consistent on $X$ and $\tilde{X}$.  Let
  $\Delta = \Delta(X, \tilde{X})$ denote the set $i \in [n]$ for which
  $X_i = \tilde{X}_i = 1$ and $f(X)_i \neq f(\tilde{X})_i$.  In order for
  $W \land \tilde{W}$ to occur, it must hold for every
  $i \in \Delta(X, \tilde{X})$ that $Y_i = 0$; otherwise we would have
  $f(X)_i \neq f(Y)_i$ and $ f(\tilde{X})_i \neq f(Y)_i$, implying that
  $f(X)_i = f(\tilde{X})_i$.

  On the other hand $\{Y_i\}_{i \in [n]}$ are conditionally independent given $X$,
  $\tilde{X}$.  The conditional distribution of $Y_i$ given
  $X_i = \tilde{X}_i = 1$ can be computed as
  \begin{align*}
    \cP[Y_i = 1 | X_i = \tilde{X}_i = 1 ] &= \frac{\cP[Y_i = 1] \cdot \cP[X_i = \tilde{X}_i = 1 |
                                            Y_i = 1]}{
    \cP[Y_i = 0] \cdot \cP[X_i = \tilde{X}_i = 1 | Y_i = 0] + \cP[Y_i = 1] \cdot
                                            \cP[X_i = \tilde{X}_i = 1 | Y_i = 1]} \\
                                          &= \frac{(2/3) \cdot (1 / 4)}{(1/3) \cdot 1 + (2/3) \cdot (1/4)} \\
    &= \frac{1}{3},
  \end{align*}
  so $\cP[W, \tilde{W} | X, \tilde{X}] \le (2/3)^{ | \Delta(X, \tilde{X})
     |}$.  Thus when $\cP[W, \tilde{W} | X, \tilde{X}] \ge \alpha^2/4$, we
  must have $|\Delta(X, \tilde{X})| \le
  \frac{2 \ln(1 / \alpha) + \ln(4)}{\ln(1.5)} \le 5 \ln(1 / \alpha) + 5$, which means that $f$
  is approximately consistent on $X$ and $\tilde{X}$.
\end{proof}

\paragraph{The probability space $\cC$.}
Let $\cC$ be obtained by the following process:
\begin{enumerate}
\item Sample $X \gets Q_X^{\otimes n}$.
\item Sample a set $S \subseteq [n]$ by including $i$ in $S$ with
  probability $1/4$,  independently for each $i \in [n]$.
\item For $i \in S$, set $\tilde{X}_i = X_i$.
\item For $i \notin S$, sample $\tilde{X}_i \gets Q_X$ (independently).
\item Sample $\tilde{Y}$, $\tilde{Z}$ such that \[
    \cC\big [\tilde{Y} = \tilde{y}, \tilde{Z}
    = \tilde{z} | S = s, X = x, \tilde{X} = \tilde{x} \big ] \eqdef Q^{\otimes
      n}_{Y,Z|X}(\tilde{y}, \tilde{z} | \tilde{x}).
  \]
\end{enumerate}

\begin{claim}
  \label{claim:C-probabilities}
    $\cC_{X,\tilde{X}} = \cP_{X, \tilde{X}}$ and $\cC_{\tilde{X}, \tilde{Y},
      \tilde{Z}} = \cP_{\tilde{X}, Y, \tilde{Z}}$.
\end{claim}
\begin{proof}
  We start by showing $\cC_{X, \tilde{X}} = \cP_{X, \tilde{X}}$.    It is clear
  from their definitions that in both $\cP$ and $\cC$,
  $\big \{(X_i, \tilde{X}_i)\big \}_{i \in [n]}$ are
  i.i.d., so it suffices to show for each $i$
  that $\cC_{X_i, \tilde{X}_i} = \cP_{X_i, \tilde{X}_i}$.
  By \cref{claim:P-probabilities}, we need to show that
  \[
    \cC[X_i = x_i, \tilde{X}_i = \tilde{x}_i]  = \begin{cases}
      1/6 & \text{if $(x_i, \tilde{x}_i) \in \big \{(0, 0), (0, 1), (1, 0) \big \}$} \\
      1/2 & \text{if $(x_i, \tilde{x}_i) = (1, 1)$.}
    \end{cases}
  \]
  For
  $x_i = \tilde{x}_i$, we have
  \begin{align*}
    \cC[X_i = x_i \land  \tilde{X}_i = \tilde{x}_i] &=
    \cC[X_i = x_i] \cdot \big (\cC[i \in S] + \cC[i \notin S] \cdot \cC[\tilde{X}_i =
                                                      \tilde{x}_i | i \notin S]
                                                      \big )
    \\
    &=     \cC[X_i = x_i] \cdot \Big (\frac{1}{4}  + \frac{3}{4} \cdot
      Q_X(\tilde{x}_i) \Big )\\
                                                    &= \begin{cases}
                                                      \frac{1}{2} & \text{if
                                                        $x_i = \tilde{x}_i = 1$}
                                                      \\
                                                      \frac{1}{6} & \text{if
                                                        $x_i = \tilde{x}_i = 0$.}
                                                    \end{cases}
  \end{align*}

  We also have
  \begin{align*}
    \cC[X_i = 0 \land \tilde{X}_i = 1] &= \cC[X_i = 0] \cdot \cC[i \notin S]
                                         \cdot \cC[\tilde{X}_i = 1] \\
                                       &= \frac{1}{3} \cdot \frac{3}{4} \cdot
                                         \frac{2}{3} \\
                                       &= \frac{1}{6}
  \end{align*}
  and similarly $\cC[X_i = 1 \land \tilde{X}_i = 0] = 1/6$.

  We now show that
  $\cC_{\tilde{X}, \tilde{Y}, \tilde{Z}} = \cP_{\tilde{X}, Y, \tilde{Z}}$.  It
  is immediate from the definition of $\cP$ that $\cP_{\tilde{X}, Y, \tilde{Z}}$
  is just $Q_{X,Y,Z}^{\otimes n}$.  Similarly, $\cC_{\tilde{X}, \tilde{Y}, \tilde{Z}} =
  \cC_{\tilde{X}} \cdot \cC_{\tilde{Y}, \tilde{Z} | \tilde{X}} = Q_{X}^{\otimes
    n} \cdot Q_{Y,Z|X}^{\otimes n} = Q_{X,Y,Z}^{\otimes n}$.
\end{proof}

Let $G$ denote the event that $f$ is approximately consistent on $X$ and
$\tilde{X}$ and $\big (f(\tilde{X}), g(\tilde{Y}), h(\tilde{Z}) \big )$ wins on $(\tilde{X}, \tilde{Y}, \tilde{Z})$.
\begin{claim}
  \label{claim:consistent-and-win}
  $\cC[G] \ge \alpha^4/8$.
\end{claim}
\begin{proof}
Let $A$ denote the event that $f$ is approximately consistent on $X$ and
$\tilde{X}$, and let $B$ denote the event that $\big (f(\tilde{X}),
g(\tilde{Y}), h(\tilde{Z}) \big )$ wins on
$(\tilde{X}, \tilde{Y}, \tilde{Z})$.  Note that $A$ depends only on $X$ and
$\tilde{X}$, while $B$ depends only on $\tilde{X}$, $\tilde{Y}$, and
$\tilde{Z}$.  We want to show that $\cC[A \land B] \ge \alpha^4/8$.

We have
\begin{align*}
  \cC[A \land B]
            &= \E [ 1_A \cdot 1_B] \\
            &= \E \big [ 1_A
              \cdot \E [1_B | X, \tilde{X} ] \big ] &&
                                                                              \text{($A$
                                                                              depends
                                                                              only
                                                                              on
                                                                              $X,
                                                                              \tilde{X}$)}
  \\
            &= \E \big [ 1_A \cdot \cC[B | X, \tilde{X}] \big ]
  \\
  & \ge  \frac{\alpha^2}{4} \cdot \cC \big [ 1_A \cdot \cC [ B | X, \tilde{X} ] \ \ge\  \alpha^2/4
    \big ] && \text{(Markov)} \\
  & = \frac{\alpha^2}{4} \cdot \cC \big [ 1_A \cdot \cC [ B | \tilde{X} ]\ \ge\ \alpha^2/4
    \big ] && \text{($\tilde{Y}, \tilde{Z}$ independent of $X$ given
              $\tilde{X}$)} \\
  & = \frac{\alpha^2}{4} \cdot \cC \big [ A\  \land\  \cC [ B | \tilde{X} ] \ge \alpha^2/4
    \big ] \\
  & \ge \alpha^4 /8 && \text{(\cref{claim:C-probabilities,claim:P-good})},
\end{align*}
where $\E$ denotes expectation in $\cC$.
\end{proof}

\begin{claim}
  \label{claim:conditional-consistent-and-win}
  There exists a set $s \subseteq [n]$ and a string $x \in \{0,1\}^n$ such that
  if we define $t \eqdef \{i \in [n] : i \in s \lor x_i = 0\}$, then
  $|t| \le 5n / 6$ and there exist strings $\tilde{x}_t$, $\tilde{y}_t$,
  $\tilde{z}_t \in \{0,1\}^{|t|}$ such that
  \[
    \cC[G | X = x, S = s, \tilde{X}_t = \tilde{x}_t,
    \tilde{Y}_t = \tilde{y}_t, \tilde{Z}_t = \tilde{z}_t] \ge \alpha^4 / 16,
  \]
  where the left-hand side is well-defined (i.e. the event conditioned on has
  non-zero probability).
\end{claim}
\begin{proof}
  Recall $S$ is a random set generated by including $i$ in $S$
  independently with probability $1/4$ for each $i \in [n]$, and $X$ is distributed like
  $\Bernoulli(2/3)^n$.  Thus by Chernoff bounds,
  \[
    \cC \big [ |S| \ge n / 3 \big ] \le e^{-2 \cdot (1/3 - 1/4)^2 \cdot n} = e^{-n/72}
  \]
  and
  \[
    \cC \Big [ \big | \{ i : X_i = 0 \} \big | \ge n / 2 \Big ] \le e^{-2 \cdot
      (1/2 - 1/3)^2 \cdot n} = e^{-n/18}.
  \]
  Thus with all but $e^{-n/72} + e^{-n/18}$ probability, $|T| \le n/3 + n/2 =
  5n/ 6$.

  Together with \cref{claim:consistent-and-win}, this implies that
  \begin{align*}
    \cC[G \land |T| \le 5n/6] \ge \alpha^4/8 - e^{-n/18} - e^{-n/72} & \ge
    \alpha^4 / 8 - 2 \cdot e^{-n/72} \\
    & \ge \alpha^4 / 16. && \text{(because $\alpha^4 \ge 32 e^{-n/72}$)}
  \end{align*}
  
  By \cref{fact:averaging} there exist fixed values
  $x \in \{0,1\}^n$ and $s \subseteq [n]$ such that
  \[
    \cC \big [ G \land |T| \le 5n/6 \big | X = x, S = s \big ] \ge \alpha^4 / 16.
  \]

  Fix such an $x$ and $s$.  By construction
  $t \eqdef \{i \in [n] : i \in s \lor x_i = 0\}$ must satisfy $|t| \le 5n/6$.
  Applying \cref{fact:averaging} again, there must also exist values
  $\tilde{x}_t, \tilde{y}_t, \tilde{z}_t \in \{0,1\}^{|t|}$ such that
  \[
    \cC \big [G \big | X = x, S = s, \tilde{X}_t = \tilde{x}_t, \tilde{Y}_t =
    \tilde{y}_t, \tilde{Z}_t = \tilde{z}_t \big ] \ge \alpha^4/16.  \qedhere
  \]
\end{proof}
We are now in a position to finish the proof of \cref{prop:f-nearly-constant}.
Let $x$, $s$, $t$, $\tilde{x}_t$, $\tilde{y}_t$, and $\tilde{z}_t$ be as guaranteed
by \cref{claim:conditional-consistent-and-win}, and let
$n' \eqdef n - |t|$.  Since $|t| \le 5n /6$, $n' \ge n / 6$.  Define $f' : \{0,1\}^{n'} \to \{0,1\}^{n'}$ to be the
function that on input $x' \in \{0,1\}^{n'}$ outputs
$f(\tilde{x}_t, x')_{[n] \setminus t}$,\footnote{Here we define
  $(\tilde{x}_t, x')$ as the $n$-bit string whose $i^{th}$ bit is
  $(\tilde{x}_t)_i$ if $i \in t$ and is $x'_i$ otherwise, where we view the bits
  of $\tilde{x}_t$ and $x'$ as respectively indexed by $t$ and by
  $[n] \setminus t$.} and define $g', h' : \{0,1\}^{n'} \to \{0,1\}^{n'}$
analogously.  Let $a \in \{0,1\}^{n'}$ denote $f(x)_{[n] \setminus t}$.

Suppose $G$ holds, $X = x$, $S = s$, $\tilde{X}_t = \tilde{x}_t$,
$\tilde{Y}_t = \tilde{y}_t$, and $\tilde{Z}_t = \tilde{z}_t$.  By the definition
of $G$, $\big (f(\tilde{X})_i, g(\tilde{Y})_i, h(\tilde{Z})_i \big )$ wins on
$(\tilde{X}_i, \tilde{Y}_i, \tilde{Z}_i)$ for all $i \in [n]$.  In particular
\begin{equation}
\text{$\big (f'(\tilde{X}_{-t})_i, g'(\tilde{Y}_{-t})_i, h'(\tilde{Z}_{-t})_i \big )$
wins on
$\big ( (\tilde{X}_{-t})_i, (\tilde{Y}_{-t})_i, (\tilde{Z}_{-t})_i \big )$ for
all $i \in [n']$.}
\end{equation}
But $f'(\tilde{X}_{-t})_i = a_i$ for all but $5 \ln(1 / \alpha) + 5$ of the values of
$i \in [n']$ for which $(\tilde{X}_{-t})_i = 1$.  This follows from the approximate
consistency of $f$ on $X$ and $\tilde{X}$, and the fact that $X_{[n] \setminus
  t} = x_{[n] \setminus
  t} = 1^{n'}$.  For all other $i$ (that is, $i$ for which $(\tilde{X}_{-t})_i = 0$), the value
of $f'(\tilde{X}_{-t})_i$ is
irrelevant to whether or not
$\big (f'(\tilde{X}_{-t})_i, g'(\tilde{Y}_{-t})_i, h'(\tilde{Z}_{-t})_i \big )$ wins
on $\big ((\tilde{X}_{-t})_i, (\tilde{Y}_{-t})_i, (\tilde{Z}_{-t})_i \big )$.
Thus for
all but $5 \ln(1 / \alpha) + 5$ values of $i \in [n']$,
\begin{equation}
  \text{ $\big (a_i, g'(\tilde{Y}_{-t})_i, h'(\tilde{Z}_{-t})_i \big )$ wins on
$\big ((\tilde{X}_{-t})_i, (\tilde{Y}_{-t})_i, (\tilde{Z}_{-t})_i \big )$ }
\end{equation}

To complete the proof of \cref{prop:f-nearly-constant}, we simply note that
\begin{align*}
  \cC \big [ G \big |  X = x, S = s, \tilde{X}_t = \tilde{x}_t, \tilde{Y}_t =
    \tilde{y}_t, \tilde{Z}_t = \tilde{z}_t \big ] \ge \alpha^4/16 && \text{(\cref{claim:conditional-consistent-and-win})}
\end{align*}
and  that
\begin{align*}
\cC_{\tilde{X}_{-t}, \tilde{Y}_{-t}, \tilde{Z}_{-t}|X = x, S = s, \tilde{X}_t =
  \tilde{x}_t, \tilde{Y}_t = \tilde{y}_t, \tilde{Z}_t = \tilde{z_t}} =
  Q^{\otimes n'} && \text{(Definition of $\cC$).} & \qedhere
\end{align*}
\end{proof}

\subsection{Strategies with Constant Functions Can't Win}

\begin{proposition}
  \label{prop:f-cant-be-constant}
   For any $n \in \N$, any $a \in \{0,1\}^{n}$, and any functions $g, h :
   \{0,1\}^{n} \to \{0,1\}^{n}$,
  \[
    \Pr_{(x, y, z) \gets Q^{\otimes n}} \big [ (a_i, g(y)_i, h(z)_i) \text{ wins
      on } (x_i, y_i, z_i) \text{ for more than $0.99n$ values of $i \in [n]$}
    \big ] \le 3 \cdot e^{-10^{-7} \cdot n}.
  \]
\end{proposition}
\begin{proof}
  Let $W_i$ denote the event that $(a_i, g(y)_i, h(z)_i)$ wins on
  $(x_i, y_i, z_i)$, and let $W$ denote the event that $W_i$ holds for at least
  $0.99n$ values of $i \in [n]$.  Let $E$ denote the event that there are at
  least $0.1n$ coordinates $i \in [n]$ where $y_i = z_i = 1$ (and $x_i = 0$).

  By a Chernoff bound,
  $\Pr_{(x, y, z) \gets Q^{\otimes n}}[E] \ge 1 - e^{-2 (1/3 - 1/10)^2 n} \ge 1
  - e^{-0.1n}$.  By a union bound, \begin{equation} \Pr_{(x, y, z) \gets
      Q^{\otimes n}} \big [ W \land E \big ] \ge \Pr[W] - e^{-0.1n}.
  \end{equation}

  When $W$ and $E$ both hold, there must exist at least $.05n$ coordinates $i$
  for which $W_i$ holds and $y_i = z_i = 1$, meaning that $g(y)_i \neq h(z)_i$,
  which in turn means that either:
  \begin{itemize}
  \item For at least $0.025n$ coordinates $i$, $y_i = 1$ and $g(y)_i = a_i$; or
  \item For at least $0.025n$ coordinates $i$, $z_i = 1$ and $h(z)_i = a_i$.
  \end{itemize}

  Thus it must either hold that
  \begin{equation}
    \label{eq:g-agrees}
    \Pr_{(x, y, z) \gets Q^{\otimes n}} \big [ W \text{ and for at least $0.025n$ coordinates $i$, $y_i = 1$ and
    $g(y)_i = a_i$}
    \big ] \ge \frac{1}{2} \cdot \big (P[W] - e^{-0.1n} \big )
  \end{equation}
  or
  \begin{equation}
    \label{eq:h-agrees}
    \Pr_{(x, y, z) \gets Q^{\otimes n}} \big [ W \text{ and for at least $0.025n$ coordinates $i$, $z_i = 1$ and
    $h(z)_i = a_i$}
    \big ] \ge \frac{1}{2} \cdot \big (P[W] - e^{-0.1n} \big )
  \end{equation}

  Suppose it is \cref{eq:h-agrees} that holds.  In
  fact this supposition is without loss of generality because $X$, $Y$, and
  $Z$ are interchangeable in the definition $Q^{\otimes n}$.

  Consider any $z$ such that  $I_z \eqdef \big \{ i \in [n] : z_i = 1 \text{ and } h(z)_i =
  a_i \big \}$ satisfies $|I_z| \ge 0.025n$.  The
  distribution of $X_{I_z}$ conditioned on $Z = z$ is uniform on $\{0,1\}^{|I_z|}$.
  However, $W$ can happen only if $X_i = 0$ for all but $0.01n$ values of $i \in I_z$.  Thus by a Chernoff bound, 
  \begin{align*}
  \cP \big [ W \big | Z = z \big
  ] \le e^{-2 \cdot (0.0025 ) ^2 \cdot 0.025 \cdot n} \le e^{ - 10^{-7} \cdot n }.
  \end{align*}
  Since this holds for every
  such $z$, we must have
  \[
    \Pr_{(x, y, z) \gets Q^{\otimes n}} \big [ W \text{ and for at least $0.025n$ coordinates $i$, $z_i = 1$ and
    $h(z)_i = a_i$}
    \big ] \le e^{-10^{-7}n}.
  \]
  This contradicts
  \cref{eq:h-agrees} unless $P[W] \le 2 \cdot e^{-10^{-7} \cdot n} + e^{-0.1n} \le 3
  \cdot e^{-10^{-7} \cdot n}$.
\end{proof}
\subsection{Parallel Repetition for the Anti-Correlation Game}

We now complete the proof of \Cref{thm:anti_corr_par_rep}. 
\begin{proof}[Proof of \Cref{thm:anti_corr_par_rep}]
  Let $n$ be sufficiently large, and let $f,g,h : \{0,1\}^n \to \{0,1\}^n$ be
  any strategy for $\cG^{\otimes n}$.  If $(f, g, h)$ wins $\cG^{\otimes n}$
  with probability $\alpha$ satisfying $\alpha > e^{-10^{-9} \cdot n}$ (which in
  particular means that $\alpha^4 \ge 32e^{-n / 72}$), then by
  \cref{prop:f-nearly-constant}, there exists $n' \ge n / 6$ and a strategy
  $f', g', h' : \{0,1\}^{n'} \to \{0,1\}^{n'}$ for $\cG^{\otimes n'}$, with $f'$
  a constant function, such that $(f', g', h')$ wins on all but at most
  $5 \ln(1 / \alpha) + 5 < 0.01n'$ coordinates of
  $\cG^{\otimes n'}$ with probability at least
  $\alpha^4 / 16 > 3 \cdot e^{-10^{-7} \cdot n'}$, contradicting
  \cref{prop:f-cant-be-constant}.
\end{proof}
\section{Four Point AND Distribution} \label{sec:four_point_game}

This section is devoted to the proof of the following theorem:

\begin{theorem}\label{thm:four_point_par_rep}
	Let $\mc G = (\mc X\times \mc Y\times \mc Z,\ \mc A \times \mc B \times \mc C,\ Q, V)$ be a 3-player game with $\mc X = \mc Y = \mc Z = \set{0,1}$, and $Q$ the uniform distribution over $\mc S = \set{(0,0,0), (1,0, 0), (0,1,0), (1,1,1)} = \set{(x,y,z)\in \set{0,1}^3 : z=x\land y}$, and such that $\val(\mc G)<1$.
	Then, there exists a constant $c = c(\mc G)>0$, such that for every $n\in \N$, it holds that $\val(\mc G^{\otimes n})\leq n^{-c}$. 
\end{theorem}

For the rest of this section, we fix 3-player game $\mc G$ satisfying the theorem hypothesis, and a large enough $n\in \N$.
Note that we prove the theorem only for large enough $n$, as the theorem trivially holds for small $n$, for a sufficiently small constant $c>0$.
By Lemma~\ref{lemma:pred_for_deter_inp}, we also assume that the predicate $V$, on player inputs $(1,1,1)$, does not depend on the answer given by the $3\ts{rd}$ player.
Let $\mc D \subseteq \mc A\times \mc B$ be the set of pairs of answers of player 1 and player 2, that lose the game $\mc G$ when the players get inputs $(1,1,1)$.

Consider the game $\mc G^{\otimes n} = (\brac{\mc X\times \mc Y\times \mc Z}^{\otimes n}, (\mc A \times \mc B \times \mc C)^{\otimes n}, P, V^{\otimes n})$, and suppose that $(X,Y,Z)$ is the random variable denoting the inputs to the three players in $\mc G^{\otimes n}$.
Let $f:\set{0,1}^n\to \mc A^{\otimes n},\ g:\set{0,1}^n\to \mc B^{\otimes n},\ h:\set{0,1}^n\to \mc C^{\otimes n} $ denote a set of optimal strategies for the three players respectively, and let $(A,B,C) = (f(X), g(Y), h(Z))$ be the random variables denoting the answers given by the players.

For each $S\subseteq [n]$, let $W_S$ be the event of winning all coordinates $i\in S$.
Let $W = W_{[n]}$ be the event of winning all the coordinates.

Before we prove the theorem, we state the main technical lemma required for the proof, which is then proved later.

\begin{lemma}\label{lemma:4_pt_main_lemma}
	Let $\delta = \frac{1}{100}$.
	Suppose $E = E^1\times E^2\times \set{0,1}^n$ is a product event on the inputs of the players, with $P(E) \geq 2^{-n^{10\delta}}$.
	Also suppose that for each $i\in [n]$, there exist sets $\mc A_i \subseteq \mc A, \mc B_i \subseteq \mc B$ such that the following hold:
	\begin{enumerate}[label=(\alph*)]
		\item $\mc D \subseteq \brac{\mc A\times \mc B}\setminus\brac{\mc A_i\times \mc B_i}$.
		\item $P\brac{X_i=1 \land A_i\not\in \mc A_i\st E}\leq n^{-\delta}$ and $P\brac{Y_i=1 \land B_i\not\in \mc B_i\st E}\leq n^{-\delta}$.	\end{enumerate} 
	Then, it holds that $P(W\st E)\leq n^{-\delta/3}$.
\end{lemma}

Next, assuming this lemma, we complete the proof of Theorem~\ref{thm:four_point_par_rep}.

\begin{proof}[Proof of Theorem~\ref{thm:four_point_par_rep}]
	Let $\delta=\frac{1}{100}$ be as in Lemma~\ref{lemma:4_pt_main_lemma}, and let $m = \lceil n^{5\delta}\rceil$.
	We will show that for every $k\in [m]\cup\set{0}$, and every product event $E = E^1\times E^2\times \set{0,1}^n$ with $P(E)\geq 2^{-n^{4\delta}(m-k)}$, it holds that $P\brac{W|E} \leq v_k$, where $v_k = \max\set{n^{-\delta/3},\ \brac{1-n^{-4\delta}}^k+n^{-4\delta}}$.
	Assuming this, the case $m=k$, with $E = \set{0,1}^n\times \set{0,1}^n\times \set{0,1}^n$ gives us 
	\[P(W) \leq  \max\set{n^{-\delta/3},\  \brac{1-n^{-4\delta}}^{n^{5\delta}}+n^{-4\delta}} \leq \max\set{n^{-\delta/3},\ n^{-3\delta}} =n^{-\delta/3} ,\]
	which completes the proof.
	
	We prove the above claim by induction on $k$. The base case $k=0$ holds trivially as $v_0\geq {\brac{1-n^{-4\delta}}^k = 1}$.
	Next, we suppose that the claim holds for $k-1$, for some $1\leq k\leq m$, and we prove it for $k$.
	
	Let $E = E^1\times E^2\times \set{0,1}^n$ be any event with $P(E)\geq 2^{-n^{4\delta}(m-k)}\geq 2^{-n^{10\delta}}$.
	Recall that $\mc D \subseteq \mc A\times \mc B$ is the set of pairs of answers of player 1 and player 2, that lose the game $\mc G$ on inputs $(1,1,1)$.
	For each $i \in [n]$, let $\tilde{L_i}$ denote the event $\brac{(X_i, Y_i, Z_i) = (1,1,1)}\land \brac{\brac{A_i, B_i}\in \mc D}$, and let $\tilde{W_i}$ be the complement event of $\tilde{L_i}$.
	
	We consider the following cases:
	
	\begin{enumerate}
		\item Suppose that for every $i\in [n]$, it holds that $P\brac{\tilde{L_i} \st E} \leq n^{-4\delta}$.
		
		Consider any fixed $i\in [n]$.
		Then, for every $(a_i,b_i)\in \mc D$, we have that 
		\begin{align*}
			n^{-4\delta} &\geq P\brac{(A_i, B_i) = (a_i, b_i) \land (X_i,Y_i,Z_i) = (1,1,1) \st E}
			\\&= P\brac{X_i=1 \land A_i=a_i\st E}\cdot P\brac{Y_i=1 \land B_i=b_i\st E},
		\end{align*}
		where the  equality follows from the fact that $P_{XY} = P_X\times P_Y$ is a product distribution, the product structure of $E$, and that $Z_i$ is a deterministic function of $X_i,Y_i$.
		Hence, for every $(a_i,b_i)\in \mc D$, either $P\brac{X_i=1 \land A_i=a_i\st E}\leq n^{-2\delta}$ or $ P\brac{Y_i=1 \land B_i=b_i\st E} \leq n^{-2\delta}$.
		
		We define
		\[\mc A_i = \set{a_i\in \mc A: P\brac{X_i=1 \land A_i=a_i\st E}> n^{-2\delta}},\]
		\[\mc B_i = \set{b_i\in \mc B:  P\brac{Y_i=1 \land B_i=b_i\st E}> n^{-2\delta}}.\]
		Then, for each $(a_i, b_i)\in \mc D$, it holds that either $a_i\not\in \mc A_i$ or $b_i\not\in \mc B_i$.
		Also, by a union bound, we have  that
		\[ P\brac{X_i=1 \land A_i\not\in \mc A_i\st E} \leq \norm{\mc A}\cdot n^{-2\delta} \leq n^{-\delta},\]
		\[ P\brac{Y_i=1 \land B_i\not\in \mc B_i\st E} \leq \norm{\mc B}\cdot n^{-2\delta} \leq n^{-\delta}.\]

		Hence, we have that $\mc A_1, \mc B_1, \dots, \mc A_n, \mc B_n$ satisfy the hypothesis in Lemma~\ref{lemma:4_pt_main_lemma}, and so ${P\brac{W\st E}\leq n^{-\delta/3} \leq v_k}$.

		\item Let $i \in [n]$ be such that $P(\tilde{L_i} |E) \geq n^{-4\delta}$. That is, $P(\tilde{W_i} |E) \leq 1-n^{-4\delta}$.
			
			 Let $T = \brac{(X_i, Y_i, Z_i), (A_i, B_i)}$.
			 Then, the event $\tilde{W_i}$ depends deterministically on $T$.
			 Let $\mc T$ be the set of all $t$ that satisfy the event $\tilde{W_i}$, and let $\mc T'\subseteq \mc T$ consist of all $t\in \mc T$ such that $P\brac{t\st E}\geq 2^{-n^{4\delta}}$.
			 Then, we have that 
		\begin{align*}
			P\brac{W\st E} &= P\brac{W\land \tilde{W_i}\st E}
			\\&= \sum_{t\in \mc T} P\brac{t\st E} P\brac{W \st E, t}
			\\&= \sum_{t\in \mc T'} P\brac{t\st E} P\brac{W \st E, t} + \sum_{t\in \mc T\setminus \mc T'} P\brac{t\st E} P\brac{W \st E, t} 
			\\&\leq \sum_{t\in \mc T'} P\brac{t\st E} P\brac{W \st E, t} + \norm{\mc T\setminus \mc T'}\cdot 2^{-n^{4\delta}}
			\\&\leq \sum_{t\in \mc T'} P\brac{t\st E} P\brac{W \st E, t} + 4\norm{\mc A}\norm{\mc B}\cdot 2^{-n^{4\delta}}
		\end{align*}
		Now, for any $t\in \mc T'$, the event $E\land (T=t)$ is a product event on the inputs of player 1 and player 2, and it holds that $P(E, t) \geq 2^{-n^{4\delta}(m-k)}\cdot 2^{-n^{4\delta}}  = 2^{-n^{4\delta}(m-(k-1))}$.
		Hence, using the induction hypothesis, we get
		\begin{align*}
			P\brac{W\st E} &\leq \sum_{t\in \mc T'} P\brac{t\st E} \cdot v_{k-1} + 4\norm{\mc A}\norm{\mc B}\cdot 2^{-n^{4\delta}}
			\\&\leq P\brac{\tilde{W_i}\st E} \cdot v_{k-1} + 4\norm{\mc A}\norm{\mc B}\cdot 2^{-n^{4\delta}}
			\\&\leq \brac{1-n^{-4\delta}} \cdot v_{k-1} + 4\norm{\mc A}\norm{\mc B}\cdot 2^{-n^{4\delta}}
			\\&\leq v_k. \qedhere
		\end{align*}
	\end{enumerate}
\end{proof} 
%
%

The remainder of this section is devoted to the proof of Lemma~\ref{lemma:4_pt_main_lemma}, which we assumed to complete the proof of Theorem~\ref{thm:four_point_par_rep}.

We recall that $\mc D \subseteq \mc A\times \mc B$ is the set of pairs of answers of player 1 and player 2, that lose the game $\mc G$ on inputs $(1,1,1)$.

Let $\delta = \frac{1}{100}$. We will work with a fixed product event  $E = E^1\times E^2\times \set{0,1}^n$ with $P(E) \geq 2^{-n^{10\delta}}$, and sets $\mc A_1, \dots, \mc A_n \subseteq \mc A$ and $\mc B_1, \dots, \mc B_n \subseteq \mc B$ such that for all $i\in [n]$:
\begin{enumerate}[label=(\alph*)]
	\item\label{item:4_pt_main_lemma_lose_prop} $\mc D \subseteq \brac{\mc A\times \mc B}\setminus\brac{\mc A_i\times \mc B_i}$.
	\item\label{item:4_pt_main_lemma_small_prob} $P\brac{X_i=1 \land A_i\not\in \mc A_i\st E}\leq n^{-\delta}$ and $P\brac{Y_i=1 \land B_i\not\in \mc B_i\st E}\leq n^{-\delta}$.
\end{enumerate}

\begin{definition}\label{defn:four_point_G}
	For any $i\in [n]$, we define $G(i)$ to be the event 
	\[ \brac{X_i\not=1 \lor A_i\in \mc A_i} \land \brac{Y_i\not=1 \lor B_i\in \mc B_i} .\]
	For any subset $S\subseteq [n]$, we define $G(S)$ to be the event $\bigwedge\limits_{i\in S} G(i)$. 
\end{definition}

\begin{lemma}\label{lemma:bad_event_prob}
	For every $S\subseteq [n]$, it holds that $P(\lnot G(S)\st E) \leq 2\norm{S}\cdot n^{-\delta}$.
\end{lemma}
\begin{proof}
	This follows from a union bound on property~\ref{item:4_pt_main_lemma_small_prob}.
\end{proof}

\begin{definition}\label{defn:z_zero_set}
	For any $z\in \set{0,1}^n$, we define $K'_z = \set{i\in [n] : z_i = 0}$, and $m'_z = \norm{K'_z}$.
	Also, we define $K_z\subseteq K'_z$ to be a subset of the largest possible size such that $(\mc A_i,\mc B_i, h(z)_i)$ are the same for each $i\in K_z$, and let $m_z = \norm{K_z}$.
\end{definition}


\subsection{Fixed Input For Player 3}

Throughout this subsection, we will consider a fixed input $z\in \set{0,1}^n$ for player 3.
Let $K'=K'_z,\ m'=m'_z,\ K=K_z,\ m=m_z$ be as in Definition~\ref{defn:z_zero_set}.

We will also consider fixed inputs $x_{-K}, y_{-K}\in \set{0,1}^{[n]\setminus K}$ to player 1 and player 2 in coordinates $[n]\setminus K$, and assume that the following hold:
\begin{enumerate}[label=(\alph*)]
	\item $m'\geq \frac{n}{2}$, and hence $m\geq cn$, for $c= \frac{1}{2\cdot 2^{\norm{\mc A}+\norm{\mc B}}\cdot \norm{\mc C}}$.
	\item $P\brac{E|z, x_{-K}, y_{-K}} \geq 2^{-n^{11\delta}} \geq 2^{-m^{12\delta}}$.
\end{enumerate}

Under these assumptions, we will prove a result, namely Lemma~\ref{lemma:fixed_z_good_event}, which we will need for the proof of Lemma~\ref{lemma:4_pt_main_lemma}.

\begin{definition}\label{def:restr_2p_game}
	For every $i\in K$, we define a 2-player game $\tilde{\mc G}_i$ as follows:
	\begin{enumerate}
	\item The inputs $(\tilde{x}_i, \tilde{y}_i)$ are distributed uniformly over the set $\tilde{\mc S} = \set{(0,0), (0,1), (1,0)}$.
	\item The players give answers $\tilde{a}_i \in \mc A$ and $\tilde{b}_i \in \mc B $ respectively.
	\item The predicate $\tilde{V}_{i}$ is given by \[\tilde{V}_{i}\brac{(\tilde{x}_i, \tilde{y}_i), (\tilde{a}_i, \tilde{b}_i)} = V\brac{(\tilde{x}_i, \tilde{y}_i, 0), (\tilde{a}_i, \tilde{b}_i, h(z)_i)} \land \brac{\tilde{x}_i\not=1\lor \tilde{a}_i\in \mc A_i}\land \brac{\tilde{y}_i\not=1\lor \tilde{b}_i\in \mc B_i}.\] 
\end{enumerate}
Observe that this game is actually the same for all $i\in K$, and hence we will simply denote it by $\tilde{\mc G}$.
\end{definition}

\begin{claim}\label{claim:restr_2p_game_value}
	$\val(\tilde{\mc G})<1$.
\end{claim}
\begin{proof}
	Consider any $i\in K$.
	Suppose for the sake of contradiction that $\val(\tilde{\mc G}_{i})=1$.
	Let $\tilde{f}_i:\set{0,1}\to \mc A$, and $\tilde{g}_i:\set{0,1}\to \mc B$ be winning strategies of player 1 and player 2.
	Let $\tilde{h}_i:\set{0,1}\to \mc C$ satisfy $\tilde{h}_i(0) = h(z)_i$.
	
	By the definition of the predicate $\tilde{V}_i$, it must hold that $\tilde{f}_i(1)\in\mc A_i $ and $\tilde{g}_i(1)\in \mc B_i$.
	Then, by property~\ref{item:4_pt_main_lemma_lose_prop}, the answers $(\tilde{f}_i(1),\tilde{g}_i(1))\not\in \mc D$, and hence must win the game $\mc G$ on player inputs $(1,1,1)$.	
	This fact, along with how the predicate $\tilde{V}_{i}$ is defined, shows that the strategies $\tilde{f}_i, \tilde{g}_i, \tilde{h}_i$ win the game $\mc G$ with probability 1, which is a contradiction.
\end{proof}

\begin{lemma}\label{lemma:fixed_z_good_event_part1}
	For any constant $\epsilon\in(0,1)$, it holds that 
	\[\E_S\sqbrac{P\brac{W_S\land G(S)\st E, z, x_{-K},y_{-K}}} \leq 2^{-m^{\epsilon/2}},\]
	where the expectation is over uniformly random $S\subseteq K$ of size $\lfloor m^{\epsilon}\rfloor$.
\end{lemma}
\begin{proof}		
	Observe that when $z, x_{-K}, y_{-K}$ are fixed, the distribution $P_{X_KY_K|z, x_{-K}, y_{-K}}$ of the remaining inputs is exactly the same as the input distribution of the game $\tilde{\mc G}^{\otimes m}$.
	Furthermore, the product structure of $E = E^1\times E^2\times \set{0,1}^n$ implies that the distribution $P_{X_KY_K|E, z, x_{-K}, y_{-K}}$ is the same as the distribution of game $\tilde{\mc G}^{\otimes m}$, when inputs drawn conditioned on $\tilde{E}=\tilde{E}^1\times \tilde{E}^2$, where  $\tilde{E}^1 = \set{x_K : (x_K, x_{-K})\in E^1}\subseteq \set{0,1}^m$, and  $\tilde{E}^2 = \set{y_K : (y_K, y_{-K})\in E^2}\subseteq \set{0,1}^m$.
		
	Under this identification, we also get strategies $\tilde{f}:\set{0,1}^m\to\mc A^{\otimes m}$ and $\tilde{g}:\set{0,1}^m\to\mc B^{\otimes m}$ for the game $\mc G^{\otimes m}$ given by $\tilde{f}(x_K) = f(x_K, x_{-K})_K$ and $\tilde{g}(y_K) = g(y_K, y_{-K})_K$ respectively. 
	Given how the predicate $\tilde{V}$ for the game $\tilde{\mc G}$ is defined, we have that for any subset $S\subseteq K$, the probability $P\brac{W_S\land G(S)\st E, z, x_{-K},y_{-K}}$ equals the probability that $\tilde{f}, \tilde{g}$ win the coordinates corresponding to $S$, of the game $\tilde{\mc G}^{\otimes m}$, when the input distribution conditioned on $\tilde{E}$.
	
	Hence, by Claim~\ref{claim:restr_2p_game_value} and Proposition~\ref{prop:2_player_hard_set}, we get that for all $\epsilon \in (0,1)$, it holds that
	\[\E_S\sqbrac{P\brac{W_S\land G(S)\st E, z, x_{-K},y_{-K}}} \leq 2^{-m^{\epsilon/2}},\]
	where the expectation is over uniformly random $S\subseteq K$ of size $\lfloor m^{\epsilon}\rfloor$.
	We used that under the input distribution of $\tilde{\mc G}^{\otimes m}$, the event $\tilde{E}$ has measure $P\brac{E|z, x_{-K}, y_{-K}} \geq 2^{-m^{12\delta}}$, and $12\delta\in (0,1)$.
\end{proof}

\begin{lemma}\label{lemma:fixed_z_good_event}
	For any constant $\epsilon \in (0,1)$, it holds that 
	\[\E_S\sqbrac{P\brac{W\land G(S)\st E, z, x_{-K},y_{-K}}} \leq 2^{-n^{\epsilon/8}},\]
	where the expectation is over uniformly random $S\subseteq [n]$ of size $\lfloor n^{\epsilon}\rfloor$.
\end{lemma}
\begin{proof}
	Let $S\subseteq [n]$ of size $\lfloor n^{\epsilon}\rfloor$ be chosen uniformly at random. 
	
	Then, the quantity $\E_S\sqbrac{P\brac{W\land G(S)\st E, z, x_{-K},y_{-K}}}$ is at most 
	\[P\brac{\norm{S\cap K} \leq m^{\epsilon/2}} + \E_{S}\sqbrac{P\brac{W\land G(S)\st E, z, x_{-K},y_{-K}}\big\vert \ \norm{S\cap K} \geq m^{\epsilon/2}}.\]

	We bound the two terms on the right hand side as follows:
	\begin{enumerate}[label=(\alph*)]
		\item Recall that $m \geq cn$, for a constant $c>0$.
		Then,
		\begin{align*}
			P\brac{\norm{S\cap K} \leq m^{\epsilon/2} } &\leq P\brac{\norm{S\cap K} \leq n^{\epsilon/2}}
			\\&\leq \frac{ \binom{m}{\lfloor n^{\epsilon/2} \rfloor}\binom{n-m+\lfloor n^{\epsilon/2} \rfloor}{\lfloor n^{\epsilon}\rfloor} }{\binom{n}{\lfloor n^{\epsilon}\rfloor}}
			\\&\leq n^{\lfloor n^{\epsilon/2}\rfloor}\cdot \brac{\frac{n-m+\lfloor n^{\epsilon/2} \rfloor}{n}}^{\lfloor n^{\epsilon} \rfloor}
			\\&\leq e^{n^{\epsilon/2} \ln{n}}\cdot \brac{1-\frac{c}{2}}^{n^{\epsilon}/2}
			\\&\leq 2^{-n^{\epsilon/2}}.
		\end{align*}
		
		\item Observe that the distribution of a uniformly random subset $S\subseteq [n]$ of size $\lfloor n^\epsilon\rfloor$, conditioned on $\norm{S\cap K}\geq m^{\epsilon/2}$, is the same as the following.
		We choose $T\subseteq K$ of size $\lfloor m^{\epsilon/2}\rfloor$ uniformly at random, and then $T'\subseteq [n]\setminus T$ of size $\lfloor n^{\epsilon}\rfloor-\lfloor m^{\epsilon/2}\rfloor$ uniformly at random, and then let $S = T\cup T'$.
		Then, by Lemma~\ref{lemma:fixed_z_good_event_part1}, we can bound the second term as
		\[ \E_{T, T'}\sqbrac{P\brac{W\land G(T\cup T')\st E, z, x_{-K},y_{-K}}} \leq \E_{T}\sqbrac{P\brac{W_T\land G(T)\st E, z, x_{-K},y_{-K}}} \leq 2^{-m^{\epsilon/4}}.\]
	\end{enumerate}
	
	Combining the above, we get that the desired quantity is at most $2^{-n^{\epsilon/2}}+2^{-m^{\epsilon /4}} \leq 2^{-n^{\epsilon/8}}$.
\end{proof}


\subsection{Completing the Proof}

Let $T$ be the random variable given by ${T = (Z, X_{-K_Z}, Y_{-K_Z})}$, let ${\mc T_1 = \set{(z, x_{-K_z}, y_{-K_z}): m'_z \leq \frac{n}{2}}}$, and let ${\mc T_2 = \set{(z, x_{-K_z}, y_{-K_z}):P\brac{E\st z, x_{-K_z}, y_{-K_z}}\leq 2^{-n^{11\delta}}}}$.

\begin{lemma}\label{lemma:bad_T1}
	\[P\brac{T\in \mc T_1 \st E} \leq 2^{-n/20}.\]
\end{lemma}
\begin{proof}
	By Fact~\ref{fact:chernoff}, it holds that 
	\[P\brac{m'_Z \leq \frac{n}{2}}\leq e^{-\frac{\brac{\frac{1}{3}}^2 \brac{\frac{3}{4}} n}{2}} =  e^{-n/24} \leq 2^{-n/19}.\] 
	This gives us
	\[P\brac{T\in \mc T_1 \st E} \leq \frac{P\brac{m'_Z \leq \frac{n}{2}}}{P(E)} \leq \frac{2^{-n/19}}{2^{-n^{10\delta}}} \leq 2^{-n/20}. \qedhere\]
\end{proof}

\begin{lemma}\label{lemma:bad_T2}
	\[P\brac{T\in \mc T_2 \st E} \leq 2^{-n^{10\delta}}.\]
\end{lemma}
\begin{proof}
	The left hand side equals
	\[\sum_{t\in \mc T_2}P\brac{t\st E} = \sum_{t\in \mc T_2} \frac{P\brac{E\st t}P(t)}{P(E)} \leq \sum_{t\in \mc T_2} \frac{2^{-n^{11\delta}}P(t)}{2^{-n^{10\delta}}} \leq\frac{2^{-n^{11\delta}}}{2^{-n^{10\delta}}} \leq 2^{-n^{10\delta}}. \]
\end{proof}

\begin{proof}[Proof of Lemma~\ref{lemma:4_pt_main_lemma}]
	Observe that
	\begin{align*}
		P(W\st E) &= \sum_{t}P\brac{t|E}\cdot P(W \st E, t)
		\\&\leq \sum_{t\not\in \mc T_1\cup \mc T_2}P\brac{t|E}\cdot P(W \st E, t) + P\brac{T\in \mc T_1\st E} + P\brac{T\in \mc T_2\st E}
		\\&\leq \sum_{t\not\in \mc T_1\cup \mc T_2}P\brac{t|E}\cdot P(W \st E, t) + 2^{-n/20} + 2^{-n^{10\delta}}.
	\end{align*}
	We used Lemma~\ref{lemma:bad_T1} and Lemma~\ref{lemma:bad_T2} for the last inequality.	
	
	Let $\epsilon = \delta/2$, and let $S\subseteq [n]$ be a uniformly random subset of size $\lfloor n^\epsilon \rfloor$.	
	By Lemma~\ref{lemma:fixed_z_good_event}, we know that for each $t\not\in \mc T_1\cup \mc T_2$, it holds that
	\[P\brac{W\st E, t} =  \E_S\sqbrac{P\brac{W\land G(S)\st E, t} + P\brac{W\land \lnot G(S)\st E, t}}   \leq 2^{-n^{\epsilon/8}}+\E_S\sqbrac{P\brac{\lnot G(S)\st E, t}}.\]
	Hence, by Lemma~\ref{lemma:bad_event_prob}, we get
	\[\sum_{t\not\in \mc T_1\cup \mc T_2}P\brac{t|E}\cdot P(W \st E, t) \leq 2^{-n^{\epsilon/8}}+ \E_S\sqbrac{P\brac{\lnot G(S)\st E}}\leq 2^{-n^{\epsilon/8}} + 2 n^{\epsilon-\delta} = 2^{-n^{\delta/16}}+ 2n^{-\delta/2}. \]
	
	Putting everything together, we get
	\[P(W\st E) \leq 2^{-n^{\delta/16}}+ 2n^{-\delta/2} + 2^{-n/20}+ 2^{-n^{10\delta}} \leq n^{-\delta/3}.\qedhere\]
\end{proof}


\section{Playerwise Connected Games}\label{sec:playerconn_games}

In this section, we will prove the following theorem:

\begin{theorem}\label{thm:playerwise_conn_par_rep} (Parallel Repetition for Playerwise Connected Games)
	Let $\mc G$ be a playerwise connected $k$-player game such that ${\val(\mc G)<1}$.
	Then, there exists a constant $c = c(\mc G)>0$ such that for every $n\in \N$, it holds that $\val(\mc G^{\otimes n}) \leq n^{-c}$.
\end{theorem}

We will follow the proof outline described in Section~\ref{sec:playerconn_outline}.
We fix a $k$-player game $\mc G = (\mc X, \mc A, Q, V)$ such that ${\val(\mc G) < 1}$, and $Q$ is the uniform distribution over its support $\mc S\subseteq \mc X$.
Consider the game $\mc G^{\otimes n} = (\mc X^{\otimes n}, \mc A^{\otimes n}, P, V^{\otimes n})$ for any large enough $n\in \N$.
Let $X$ be the random variable denoting the inputs to the players in the game $\mc G^{\otimes n}$.
Let $E=E^1\times \dots \times E^k \subseteq \mc X^{\otimes n}$ be an arbitrary product event such that $P(E)\geq n^{-\delta}$, where $\delta = 0.4$.

Following \cite{Raz98, Hol09, DHVY17}, we define our dependency breaking random variable variable $R$ as follows.
\begin{definition}\label{defn:dep_break_rv}
	For each $i\in [n]$, define $R_i = (D_i, M_i)$, where $D_i\in [k]$ is chosen uniformly at random, and $M_i = X_i^{-D_i}$.
	Define $R = (R_1,\dots,R_n)$.
	
	We will use $P_R$ to denote the distribution of $R$, ignoring the fact that the definition of $R$ uses randomness additional to that coming from $P$.
\end{definition}

Next, we state two important lemmas, which are proved in Section~\ref{sec:rand_coord_good}.
The expectations in these lemmas are with respect to the uniform distribution over $i\in [n]$.

\begin{lemma}\label{lemma:Xi_same_distr}
	\[ \E_{i\in [n]} \lonorm{P_{X_i|E}}{P_{X_i}} \lesssim \sqrt{\frac{1}{n}\log_2{\frac{1}{P(E)}}}. \]
\end{lemma} 

\begin{lemma}\label{lemma:ri_same_distr}
	\[ \E_{i\in [n]} \E_{x_i\sim P_{X_i|E}} \lonorm{P_{R_{-i}|x_i, E}}{P_{R_{-i}|E}} \lesssim \frac{1}{P(E)} \sqrt{\frac{1}{n}\log_2{\frac{1}{P(E)}}}. \]
\end{lemma} 
We remark that proving the above lemma without the $1/P(E)$ factor would show that the value of the repeated game decays exponentially.

The following claim from \cite{DHVY17} follows from the definition of the random variable $R$.

\begin{claim}\label{claim:dep_break_prop} (Dependency Breaking Property)
For each $i\in [n]$, and for each $x_i, r_{-i}$, it holds that 
\[ P_{X|x_i, r_{-i}} = P_{X^1|r_{-i}, x_i^1}\times\dots\times P_{X^k|r_{-i}, x_i^k}. \]
Since $E = E^1\times\dots\times E^k$ is a product event, it further holds that 
\[ P_{X|x_i, r_{-i}, E} = P_{X^1|r_{-i}, x_i^1, E^1}\times\dots\times P_{X^k|r_{-i}, x_i^k, E^k}.\]
\end{claim}

\begin{lemma}\label{lemma:prod_set_hard_coor_exists}
	There exists a constant $\epsilon>0$ such that
	\[\E_{i\in [n]} \sqbrac{\val_i(\mc G^{\otimes n} \st (P|E))} \leq 1-\epsilon.\]
\end{lemma}

\begin{proof}
	Consider the following (randomized) strategy for the game $\mc G$. The $k$ players get an input $y$ sampled from $Q$, with the $j\textsuperscript{th}$ player getting $y^j$, for each $j\in [k]$.
	\begin{enumerate}
		\item Using shared randomness, the players sample $i\in [n]$ uniformly, and sample $r_{-i}\sim P_{R_{-i}|E}$. Let $(f^j)_{j\in [k]}$ be a strategy that achieves value $\val_i(\mc G^{\otimes n} \st (P|E))$ (the players can decide on such strategies for each $i\in [n]$ beforehand).
		\item For each $j\in [k]$, the $j\textsuperscript{th}$ player samples $x^j\sim P_{X^j|r_{-i}, X_i^j=y^j, E^j}$, and outputs the $i\textsuperscript{th}$ coordinate of $f^j(x^j)$ as an answer.
	\end{enumerate}
	Let $\ind_{i}(x)$ be the indicator random variable for the event that the optimal strategy for the $i\textsuperscript{th}$ coordinate wins on questions $x_i$.
	Then, the value obtained by the above strategy is
	{\small
	\begin{align*}
		&\sum_{y}Q(y)\cdot \E_{i\in [n]}\sum_{r_{-i}}P\brac{r_{-i}|E}\cdot\sum_{x}\brac{\prod_{j\in [k]}P\brac{x^j| r_{-i}, X_i^j=y^j, E^j}\cdot \ind_i(x) }
		\\= & \E_{i\in [n]} \sum_{y, r_{-i}, x}Q(y)\cdot P\brac{r_{-i}|E}\cdot P(x|r_{-i}, X_i=y, E)\cdot \ind_i(x) & (\text{Claim~\ref{claim:dep_break_prop}})
		\\\geq& \E_{i\in [n]}\sum_{y, r_{-i}, x}P(X_i=y|E)\cdot P\brac{r_{-i}|E}\cdot P(x|r_{-i}, X_i=y, E)\cdot \ind_i(x) - O\brac{\sqrt{\frac{\delta \log_2 n}{n}}} & (\text{Lemma~\ref{lemma:Xi_same_distr}})
		\\\geq& \E_{i\in [n]}\sum_{y, r_{-i}, x}P(X_i=y|E)\cdot P\brac{r_{-i}|E, X_i=y}\cdot P(x|r_{-i}, X_i=y, E)\cdot \ind_i(x) - O\brac{n^{\delta}\cdot \sqrt{\frac{\delta \log_2 n}{n}}} & (\text{Lemma~\ref{lemma:ri_same_distr}})
		\\=& \E_{i\in [n]}\sum_{x}P(x|E)\cdot \ind_i(x) - O\brac{n^{\delta}\cdot \sqrt{\frac{\delta\log_2 n}{n}}}
		\\=& \E_{i\in [n]} \sqbrac{\val_i(\mc G^{\otimes n} \st (P|E))} - O\brac{n^{\delta}\cdot \sqrt{\frac{\delta\log_2 n}{n}}}.
	\end{align*}}
	
	By Fact~\ref{fact:rand_no_help}, this value must be at most $\val(\mc G)<1$, and the choice $\delta = 0.4$ and $\epsilon = \frac{1-\val(\mc G)}{2} > 0$ gives the desired result for large enough $n$. 
\end{proof}

The above lemmas together prove our main theorem.

\begin{proof}[Proof of Theorem~\ref{thm:playerwise_conn_par_rep}]
	Let $\mc G$ be any playerwise connected $k$-player game with $\val(\mc G)<1$.
	Combining Lemma~\ref{lemma:wlog_unif_dist}, Lemma~\ref{lemma:prod_set_hard_coor}, and Lemma~\ref{lemma:prod_set_hard_coor_exists} (along with the probabilistic method), we get that there exists a constant $c>0$ such that for large enough $n$, it holds that $\val(\mc G^{\otimes n}) \leq n^{-c}$. 
	The result then holds for all $n\in \N$ by making the constant $c>0$ small enough.
\end{proof}

\subsection{Technical Lemmas}\label{sec:rand_coord_good}

In this section, we will prove Lemma $\ref{lemma:Xi_same_distr}$ and Lemma $\ref{lemma:ri_same_distr}.$

First, we state the following lemma from \cite{Raz98, Hol09}.

\begin{lemma}\label{lemma:pdt_distr_cond}
	Let $P_{V} = P_{V_1}\times\dots\times P_{V_n}$ be a product distribution over a set $\mc V^{\otimes n}$, and $W$ an event.
	Then,
	\[ \frac{1}{n}\sum_{i=1}^n \lonorm{P_{V_i|W}}{P_{V_i}} \lesssim \sqrt{\frac{1}{n}\log_2{\frac{1}{P_V(W)}}}.\]
\end{lemma}
\begin{proof}
	See Appendix~\ref{appendix:2player_par_rep_results} for a proof sketch.
\end{proof}

\begin{proof}[Proof of Lemma~\ref{lemma:Xi_same_distr}]
	Applying Lemma~\ref{lemma:pdt_distr_cond} to the distribution $P_X = P_{X_1}\times\dots\times P_{X_n}$ and the event $E$ gives the desired result.
\end{proof}

Now, we state the main lemma of this section.
\begin{lemma}\label{lemma:Xi_same_distr_cond_ri}
	\[ \E_{i\in [n]} \E_{r_{-i}\sim P_{R_{-i}|E}} \lonorm{P_{X_i|r_{-i}, E}}{P_{X_i}} \lesssim \frac{1}{P(E)}\sqrt{\frac{1}{n}\log_2{\frac{1}{P(E)}}}. \]
\end{lemma}

Next, we prove Lemma~\ref{lemma:ri_same_distr}, assuming the above lemma.
\begin{proof}[Proof of Lemma~\ref{lemma:ri_same_distr}]
	The equations in Lemma~\ref{lemma:Xi_same_distr} and Lemma~\ref{lemma:Xi_same_distr_cond_ri}, along with the triangle inequality give that
	\[ \E_{i\in [n]} \E_{r_{-i}\sim P_{R_{-i}|E}} \lonorm{P_{X_i|r_{-i}, E}}{P_{X_i|E}} \lesssim \frac{1}{P(E)}\sqrt{\frac{1}{n}\log_2{\frac{1}{P(E)}}}. \]
	By Bayes' rule, it holds that for every $i\in [n]$,
	\[ \E_{r_{-i}\sim P_{R_{-i}|E}} \lonorm{P_{X_i|r_{-i}, E}}{P_{X_i|E}} = \lonorm{P_{X_i,R_{-i}|E}}{P_{X_i|E}\cdot P_{R_{-i}|E}} =  \E_{x_i\sim P_{X_i|E}} \lonorm{P_{R_{-i}|x_i, E}}{P_{R_{-i}|E}}.\]
	Substituting this in the above inequality, we get the desired result.
\end{proof}

The remainder of this section is devoted to the proof of Lemma~\ref{lemma:Xi_same_distr_cond_ri}. 

\subsection{Proof of The Main Lemma}

\subsubsection{Conditioning For a Single Player}
Fix some $j\in [k]$.
We show that a stronger version of Lemma~\ref{lemma:Xi_same_distr_cond_ri} holds in the case we are conditioning only on $E^j$ (by which we mean the event $X^j\in E^j$), instead of conditioning on $E$.

\begin{lemma}\label{lemma:Xi_same_distr_cond_ri_only_Ej}
	\[ \E_{i\in [n]} \E_{r_{-i}\sim P_{R_{-i}|E^j}} \lonorm{P_{X_i|r_{-i}, E^j}}{P_{X_i}} \lesssim \sqrt{\frac{1}{n}\log_2{\frac{1}{P\brac{E^j}}}}. \]	
\end{lemma} 

To prove the above lemma, we will need the following lemmas.

\begin{lemma}\label{lemma:same_distr_over_edges}
Let $H_\mc G$ be the $(k-1)$-connection graph (see Definition~\ref{defn:conn_graph}). 
Suppose there is an edge between $y, z \in \mc S$ in the $j\textsuperscript{th}$ direction, that is, $y^{-j}=z^{-j}$ and $y^j\not= z^j$.
Then,
	\[ \E_{i\in [n]} \E_{r_{-i}\sim P_{R_{-i}|E^j}} \norm{P\brac{X_i=y|r_{-i}, E^j} - P\brac{X_i=z|r_{-i}, E^j}} \lesssim \sqrt{\frac{1}{n}\log_2{\frac{1}{P\brac{E^j}}}}. \]
\end{lemma}
\begin{proof}
	For any $r$, applying Lemma~\ref{lemma:pdt_distr_cond} to the distribution $P_{X|r} = P_{X_1|r}\times\dots\times P_{X_n|r}$, and the event $E^j$ gives that
	\[ \E_{i\in [n]} \lonorm{P_{X_i|r, E^j}}{P_{X_i|r}} \lesssim  \sqrt{\frac{1}{n}\log_2{\frac{1}{P\brac{E^j|r}}}}. \]
	Taking expectation over $r\sim P_{R|E^j}$, we get that
	\begin{align*}
		\E_{r\sim P_{R|E^j}}\E_{i\in [n]} \lonorm{P_{X_i|r, E^j}}{P_{X_i|r}} \lesssim  \E_{r\sim P_{R|E^j}} \sqrt{\frac{1}{n}\log_2{\frac{1}{P\brac{E^j|r}}}}.
	\end{align*}
	We simplify the two sides of the above equation separately.
	\begin{description}
		\item \textit{Right-hand side:} The function $\sqrt{\log_2(\cdot)}$ is concave over the domain $[1,\infty)$, and hence by using Jensen's inequality, we get 
			\[	\E_{r\sim P_{R|E^j}} \sqrt{\frac{1}{n}\log_2\brac{\frac{1}{P\brac{E^j|r}}}} 
				\leq \sqrt{\frac{1}{n}\log_2\brac{\E_{r\sim P_{R|E^j}}\frac{1}{P\brac{E^j|r}}}}\leq \sqrt{\frac{1}{n}\log_2\frac{1}{P\brac{E^j}}}. \]
		\item \textit{Left-hand side:}
			\begin{align*}
				\E_{r\sim P_{R|E^j}}\E_{i\in [n]}\lonorm{P_{X_i|r, E^j}}{P_{X_i|r}} &= \E_{i\in [n]}\E_{r\sim P_{R|E^j}}\lonorm{P_{X_i|r, E^j}}{P_{X_i|r}}
				\\&= \E_{i\in [n]}\E_{r_{-i}\sim P_{R_{-i}|E^j}}\E_{r_i\sim P_{R_i|r_{-i}, E^j}}\lonorm{P_{X_i|r, E^j}}{P_{X_i|r}}.
			\end{align*}
			We know that with probability $1/k$ (independent of everything else), $R_i=(j, X_i^{-j})$.
			Hence, the above quantity is at least $1/k$ times the following: 			\begin{align*}
				&\E_{i\in [n]}\E_{r_{-i}\sim P_{R_{-i}|E^j}} \E_{x_i^{-j}\sim P_{X_i^{-j}|r_{-i}, E^j}} \lonorm{P_{X_i|x_i^{-j}, r_{-i}, E^j}}{P_{X_i|x_i^{-j}, r_{-i}}}
				\\=& \E_{i\in [n]}\E_{r_{-i}\sim P_{R_{-i}|E^j}} \E_{x_i^{-j}\sim P_{X_i^{-j}|r_{-i}, E^j}} \lonorm{P_{X_i^j|x_i^{-j}, r_{-i}, E^j}}{P_{X_i^j|x_i^{-j}}}
				\\=& \E_{i\in [n]}\E_{r_{-i}\sim P_{R_{-i}|E^j}} \sum_{x_i} P\brac{x_i^{-j}|r_{-i}, E^j}\cdot \norm{P\brac{x_i^j|x_i^{-j}, r_{-i}, E^j}-P\brac{x_i^j|x_i^{-j}}}
				\\=& \E_{i\in [n]}\E_{r_{-i}\sim P_{R_{-i}|E^j}} \sum_{x_i} \norm{P\brac{x_i|r_{-i}, E^j} - P\brac{x_i^{-j}|r_{-i}, E^j} P\brac{x_i^j|x_i^{-j}}}
				\\\geq & \E_{i\in [n]}\E_{r_{-i}\sim P_{R_{-i}|E^j}} \norm{P\brac{X_i=y|r_{-i}, E^j} - P\brac{X_i=z|r_{-i}, E^j}}.
			\end{align*}
			The last inequality follows from the triangle inequality applied to the terms corresponding to $x_i=y$ and $x_i=z$.
			We use the observation that since $y^{-j}=z^{-j}$, and $P_{X_i}$ is the uniform distribution over $\mc S$, the second quantity inside the bracket is equal in the cases $x_i=y$ and $x_i=z$.
	\end{description}
	Combining the left-hand side and the right-hand side, we get the desired result.
\end{proof}

\begin{lemma}\label{lemma:same_distr_over_planes}
	Let $y, z \in \mc S$ be such that $y^j=z^j$. Then, for each $i\in [n]$, and each $r_{-i}$, it holds that
	\[ P\brac{X_i=y|r_{-i}, E^j} = P\brac{X_i=z|r_{-i}, E^j}.\]	
\end{lemma}
\begin{proof}
	Fix any $i\in [n]$ and $r_{-i}$. If $P\brac{X_i^j=y^j|r_{-i}, E^j} = 0$, then both quantities are zero. If not, we observe that since $P_{X_i}$ is the uniform distribution over $\mc S$, and that we are conditioning only on $E^j$, it holds that
	\begin{align*}
		P\brac{X_i=y|r_{-i}, E^j} &=  P\brac{X_i^j=y^j|r_{-i}, E^j}\cdot  P\brac{X_i^{-j}=y^{-j}|X_i^j=y^j, r_{-i}, E^j}
		\\&= P\brac{X_i^j=y^j|r_{-i}, E^j}\cdot  P\brac{X_i^{-j}=y^{-j}|X_i^j=y^j}
		\\&= P\brac{X_i^j=y^j|r_{-i}, E^j}\cdot \frac{1}{\norm{\set{w\in \mc S: w^j=y^j}}}.
	\end{align*}
	A similar expression is valid for $z$, and equals the above one since $y^j=z^j$.
\end{proof}

\begin{proof}[Proof of Lemma~\ref{lemma:Xi_same_distr_cond_ri_only_Ej}]
Since our game $\mc G$ is playerwise connected (see Definition~\ref{defn:conn_player_graph}), the graph $H_{\mc G}^j$ is connected. This fact, along with Lemma~\ref{lemma:same_distr_over_edges}, and Lemma~\ref{lemma:same_distr_over_planes} gives that for any $y, z \in \mc S$,
\[ \E_{i\in [n]} \E_{r_{-i}\sim P_{R_{-i}|E^j}} \norm{P\brac{X_i=y|r_{-i}, E^j} - P\brac{X_i=z|r_{-i}, E^j}} \lesssim \sqrt{\frac{1}{n}\log_2{\frac{1}{P\brac{E^j}}}}. \]

Now, by the above inequalities, along with the fact that $P_{X_i}$ is the uniform distribution over $\mc S$,
\[ \E_{i\in [n]} \E_{r_{-i}\sim P_{R_{-i}|E^j}} \lonorm{P_{X_i|r_{-i}, E^j}}{P_{X_i}} \lesssim \sqrt{\frac{1}{n}\log_2{\frac{1}{P\brac{E^j}}}}. \qedhere\]	
\end{proof}

\subsubsection{The General Case}
\label{sec:fivepoint_gamma}
We know that $Q$ is the uniform distribution over its support $\mc S\subseteq \mc X = \mc X^1\times\dots\times \mc X^k$.
Let $\Gamma$ be the distribution over $\mc X$, which equals the product of marginals of $Q$.
That is, $\Gamma(y) = Q(y^1)Q(y^2)\dots Q(y^k)$, for each $y\in \mc X$.
For each $i\in [n]$, let $P_{-i}\Gamma$ be the distribution over $\mc X^{\otimes n}$ given by $(P_{-i}\Gamma)(x) = P(x_{-i})\Gamma(x_i).$

\begin{observation}\label{fact:prod_distr_full_support}
	The distribution $\Gamma$ has support $\mc X$.
\end{observation}
\begin{proof}
	For each $j\in [k]$, the graph $H_\mc G^j$ is connected, and hence each question to each player occurs with positive probability under $Q$ (also see the remark after Definition~\ref{defn:conn_player_graph}).
\end{proof}

We show that Lemma~\ref{lemma:Xi_same_distr_cond_ri} holds when the distribution $P$ is replaced by $P_{-i}\Gamma$.
Note that the conditional distributions in the next lemma are well-defined since for any event $W$, if $P(W)>0$, then $P_{-i}\Gamma (W) > 0$.
This is because by Observation~\ref{fact:prod_distr_full_support}, we can write $\Gamma$ as a non-trivial convex combination of $Q$ and some other distribution over $\mc X$.

\begin{lemma}\label{lemma:same_distr_ith_prod}
	\[ \E_{i\in [n]} \E_{r_{-i}\sim P_{R_{-i}|E}} \lonorm{(P_{-i}\Gamma)_{X_i|r_{-i}, E}}{(P_{-i}\Gamma)_{X_i}} \lesssim \frac{1}{P(E)}\sqrt{\frac{1}{n}\log_2{\frac{1}{P(E)}}}. \]
\end{lemma}
\begin{proof}
	Fix some $i\in [n]$.
	By the definition of the random variable $R$, and the fact that $\Gamma$ is a product distribution, it holds that 
	\[(P_{-i}\Gamma)_{X|r_{-i}} = (P_{-i}\Gamma)_{X^1|r_{-i}}\times\dots\times (P_{-i}\Gamma)_{X^k|r_{-i}}.\]
	Since $E=E^1\times\dots\times E^k$ is a product event, we get
	\[(P_{-i}\Gamma)_{X|r_{-i}, E} = (P_{-i}\Gamma)_{X^1|r_{-i}, E^1}\times\dots\times (P_{-i}\Gamma)_{X^k|r_{-i}, E^k}.\]
	In particular, this implies that
	\[(P_{-i}\Gamma)_{X_i|r_{-i}, E} = (P_{-i}\Gamma)_{X_i^1|r_{-i}, E^1}\times\dots\times (P_{-i}\Gamma)_{X_i^k|r_{-i}, E^k}.\]
	Also, since $\Gamma$ is a product distribution, we know that $(P_{-i}\Gamma)_{X_i} = (P_{-i}\Gamma)_{X_i^1}\times\dots\times (P_{-i}\Gamma)_{X_i^k}$.
	
	This gives us
	\begin{align*}
		& \E_{i\in [n]} \E_{r_{-i}\sim P_{R_{-i}|E}} \lonorm{(P_{-i}\Gamma)_{X_i|r_{-i}, E}}{(P_{-i}\Gamma)_{X_i}} 
		\\ \leq&  \sum_{j\in [k]}\E_{i\in [n]} \E_{r_{-i}\sim P_{R_{-i}|E}} \lonorm{(P_{-i}\Gamma)_{X_i^j|r_{-i}, E^j}}{(P_{-i}\Gamma)_{X_i^j}}
		\\ =&  \sum_{j\in [k]}\E_{i\in [n]} \E_{r_{-i}\sim P_{R_{-i}|E}} \lonorm{P_{X_i^j|r_{-i}, E^j}}{P_{X_i^j}}.
	\end{align*}
	The equality $(P_{-i}\Gamma)_{X_i^j|r_{-i}, E^j} = P_{X_i^j|r_{-i}, E^j}$ uses the following facts:
	\begin{enumerate}
		\item The distribution $\Gamma$ has the same marginals as $Q$.
		\item The distribution over the variables $X_i^{-j}$ is irrelevant for the distribution $(P_{-i}\Gamma)_{X_i^j|r_{-i}, E^j}$. This is because we are conditioning only on $E^j$ (rather than $E$), and $R_{-i}$ has no dependence on the variables $X_i^{-j}$.
	\end{enumerate}
	
	Observe that for each $j\in [k]$, and each $r_{-i}$
	\[P\brac{r_{-i}|E} = \frac{P\brac{r_{-i}, E}}{P\brac{E}} \leq  \frac{P\brac{r_{-i}, E^j}}{P(E)} = P\brac{r_{-i}|E^j}\cdot \frac{P\brac{E^j}}{P(E)}. \]
	Using this, we get that the above expression is at most
	\begin{align*}
		&\sum_{j\in [k]}\frac{P\brac{E^j}}{P(E)}\E_{i\in [n]} \E_{r_{-i}\sim P_{R_{-i}|E^j}} \lonorm{P_{X_i^j|r_{-i}, E^j}}{P_{X_i^j}}
		\\\leq\ &\sum_{j\in [k]}\frac{P\brac{E^j}}{P(E)}\E_{i\in [n]} \E_{r_{-i}\sim P_{R_{-i}|E^j}} \lonorm{P_{X_i|r_{-i}, E^j}}{P_{X_i}} 
		\\\lesssim\ & \sum_{j\in [k]} \frac{P\brac{E^j}}{P(E)}\cdot \sqrt{\frac{1}{n}\log_2{\frac{1}{P\brac{E^j}}}} & (\text{Lemma }\ref{lemma:Xi_same_distr_cond_ri_only_Ej})
		\\\lesssim\ & \frac{1}{P(E)}\sqrt{\frac{1}{n}\log_2{\frac{1}{P(E)}}} & \brac{\text{for each }j\in[k],\ P\brac{E}\leq P\brac{E^j}\leq 1}.
	\end{align*}
\end{proof}

\begin{proof}[Proof of Lemma~\ref{lemma:Xi_same_distr_cond_ri}]
	By Observation~\ref{fact:prod_distr_full_support}, we know that there exists a constant $\gamma >0$ such that for each $y\in \mc X$, $\Gamma(y)\geq \gamma$.

	For each $i\in [n]$, we define a random variable $T_i$ over the set $\set{0,1}$, depending only on $X_i$ (and additional randomness), as follows.
	For each $x_i \in \mc X$, we define $\Prob{T_i = 1 | X_i = x_i} = \frac{\gamma Q(x_i)}{\Gamma(x_i)}\in [0,1]$.
	
	By Lemma~\ref{lemma:same_distr_ith_prod}, we know that
	\[ \E_{i\in [n]} \E_{r_{-i}\sim P_{R_{-i}|E}} \lonorm{(P_{-i}\Gamma)_{X_i|r_{-i}, E}}{(P_{-i}\Gamma)_{X_i}} \lesssim \frac{1}{P(E)}\sqrt{\frac{1}{n}\log_2{\frac{1}{P(E)}}}. \]
	
	Conditioning on the event $T_i=1$, by Fact~\ref{fact:norm_condition} we get that
	\[ \E_{i\in [n]} \E_{r_{-i}\sim P_{R_{-i}|E}} \lonorm{P_{X_i|r_{-i}, E}}{P_{X_i}} \lesssim \frac{2}{\gamma}\cdot \frac{1}{P(E)}\sqrt{\frac{1}{n}\log_2{\frac{1}{P(E)}}},\]
	which completes the proof. Note that technically we apply Fact~\ref{fact:norm_condition} to the product of the distribution $P$ and the distribution from which the extra randomness for $(T_i)_{i\in [n]}$ is coming.
	
	We used the following facts:
	\begin{enumerate}
		\item \[(P_{-i}\Gamma)(T_i=1) = \sum_{x_i\in \mc X}\Gamma(x_i)\cdot \frac{\gamma Q(x_i)}{\Gamma(x_i)}  = \gamma > 0.\]
		
		\item The distribution $(P_{-i}\Gamma)|_{T_i=1}$ equals the distribution $P$, since for each $x\in \mc X^{\otimes n}$,
			\begin{align*}
				(P_{-i}\Gamma)\brac{X=x|T_i=1} &= \frac{(P_{-i}\Gamma)\brac{T_i=1|X=x}\cdot (P_{-i}\Gamma)\brac{X=x}}{(P_{-i}\Gamma)\brac{T_i=1}} 
				\\&= \frac{1}{\gamma}\cdot \brac{\frac{\gamma Q(x_i)}{\Gamma(x_i)}\cdot P(x_{-i})\Gamma(x_i)}
				\\&=P(x_{-i})Q(x_i) = P(x). 
			\end{align*}
			
		\item For each $r_{-i}$ such that $P(r_{-i}|E)>0$, 
			\begin{align*}
				(P_{-i}\Gamma)(T_i=1|r_{-i}, E) &\geq  (P_{-i}\Gamma)(T_i=1, r_{-i}, E)
				\\&\geq (P_{-i}\Gamma)(r_{-i}, E | T_i=1)\cdot (P_{-i}\Gamma)(T_i=1)
				\\&= P(r_{-i}, E) \cdot \gamma
				\\& > 0.
				\qedhere
			\end{align*} 
	\end{enumerate}
\end{proof}

\section{Hamming Weight One Distribution with Binary Outputs} \label{sec:hm_wt_one}

In this section, we analyze parallel repetition for three-player games with inputs drawn uniformly from the set $\mc S = \set{(1,0,0), (0,1,0), (0,0,1)}$ of hamming-weight one inputs, and having binary outputs.
Formally, we prove the following theorem:

\begin{theorem}\label{thm:hm_wt_one_par_rep}
	Let $\mc G = (\mc X\times \mc Y\times \mc Z, \mc A\times \mc B\times \mc C, Q, V)$ be a 3-player game with $\mc X = \mc Y = \mc Z = \mc A = \mc B= \mc C= \set{0,1}$, and $Q$ the uniform distribution over $\mc S = \set{(1,0,0), (0,1, 0), (0,0,1)}$, and such that $\val(\mc G)<1$.
	Then, there exists a constant $c = c(\mc G)>0$, such that for every $n\in \N$, it holds that $\val(\mc G^{\otimes n})\leq 2^{-cn}$. 
\end{theorem}

\begin{proof}

By Lemma~\ref{lemma:pred_for_deter_inp}, it suffices to consider the case when the predicate $V$ only depends on the answers of the two players that receive the input 0.
The predicate $V$ is then given by three predicates $V_1(b,c), V_2(a,c), V_3(a,b)$, where $a,b,c\in \set{0,1}$ denote the outputs of the players when they receive input 0.
We will think of these predicates as given by 3 tables, as shown in the figure.

\begin{table}[h]
\begin{subtable}{0.33\linewidth}\centering
{\begin{tabular}{|c|c|c|}
\hline 		 & $c=0$ & $c=1$ \\
\hline $b=0$ & $*$   & $*$	 \\
\hline $b=1$ & $*$   & $*$	 \\
\hline
\end{tabular}
}
\caption*{$V_1$}
\end{subtable}
\begin{subtable}{0.33\linewidth}\centering
{\begin{tabular}{|c|c|c|}
\hline 		 & $c=0$ & $c=1$ \\
\hline $a=0$ & $*$   & $*$	 \\
\hline $a=1$ & $*$   & $*$	 \\
\hline
\end{tabular}
}
\caption*{$V_2$}
\end{subtable}
\begin{subtable}{0.33\linewidth}\centering
{\begin{tabular}{|c|c|c|}
\hline 		 & $b=0$ & $b=1$ \\
\hline $a=0$ & $*$   & $*$	 \\
\hline $a=1$ & $*$   & $*$	 \\
\hline
\end{tabular}
}
\caption*{$V_3$}
\end{subtable}
\end{table}

We will do a case analysis over all predicates in the following manner:
\begin{description}
	\item \textbf{Case 1:} At least one of the rows or columns in some table has all 0s. This is analyzed in Section~\ref{sec:three_point_case1}.
	\item \textbf{Case 2:} Each row and column in each table has at least one 1. This is analyzed in Section~\ref{sec:three_point_case2}.
\end{description}
In either case, we prove the theorem only for large enough $n$, as the theorem trivially holds for small $n$, for a sufficiently small constant $c>0$.
\end{proof}

\begin{remark}
	 Another way to view this predicate is to think of it as a tripartite graph $G = (W_1\cup W_2\cup W_3, E)$,  with each $W_i$ having 2 vertices, corresponding to player answers 0 and 1 (on input 0).
	 The accepting answer pairs in $V_1$ correspond to the edges $E\cap (W_2\times W_3)$, and similarly $V_2$ corresponds to $E\cap (W_1\times W_3)$, and $V_3$ corresponds to $E\cap (W_1\times W_2)$.
\end{remark}

\subsection{Case 1 Analysis}\label{sec:three_point_case1}

We consider the case when at least one of the rows or columns in some table has all 0s.
By using symmetry among the players and possibly using the transformation $a\mapsto 1-a$, we can assume this is the row labelled by $a=1$ in $V_3$.
More precisely, we assume $V_3(a,b) \leq \ind\sqbrac{a=0}$, where we use $\ind$ to denote the indicator operator.

Next, we prove an exponential parallel repetition bound, assuming that the predicate $V$ is such that $\val(\mc G)<1$.

Consider the $n$-fold repeated game ${\mc G^{\otimes n} = (\brac{\mc X\times \mc Y\times \mc Z}^{\otimes n}, \brac{\mc A\times \mc B\times \mc C}^{\otimes n}, P, V^{\otimes n})}$, and suppose that $(X,Y,Z)$ is the random variable denoting the inputs to the three players in $\mc G^{\otimes n}$.
Let the functions $f, g, h:\set{0,1}^n\to\set{0,1}^n$ denote a set of optimal strategies for the three players respectively.
Let $W = \set{(x,y,z)\in \mc S^{\otimes n}\st V(x,y,z,f(x),g(y),h(z)) = 1}$ denote the event of winning the game.
Let $\alpha = P(W)$, and for the sake of contradiction, assume that $\alpha \geq 2^{-\epsilon n}$ for some small enough constant $\epsilon>0$ (to be specified later).

Let $E = \set{(x,y,z)\in W\st P(W|x) \geq \alpha/2}$ denote the sub-event of $W$ that remains after removing \emph{negligible} inputs (with respect to $W$) of the first player.
Then, it holds that
\begin{enumerate}[label=(\alph*)]
	\item $P(E) \geq \alpha /2$.
	\item For each $x$ such that $P(E|x)>0$, it holds that $P(E|x) = P(W|x)\geq \alpha/2$.
\end{enumerate}

Now, consider any fixed input $x$ for player 1, and suppose $\alpha_x = P(E|x)>0$ (and hence $\alpha_x\geq \alpha/2$).
Let $S_x = \set{i\in [n]: x_i=0, f(x)_i = 1}$, and $s_x = \norm{S_x}$.
Then, for each $y,z$ such that $(x,y,z)\in E$, and for each $i\in S_x$, it must hold that $y_i=1,z_i=0$  (by the assumed condition on $V_3$).
In particular, this implies that $\alpha_x \leq \brac{1/2}^{s_x}$, since after conditioning on $x_i=0$, the inputs $(y_i,z_i)$ equal $(0,1)$ or $(1,0)$ each with probability 1/2.
Hence, we get that $s_x\leq \log_2\brac{1/\alpha_x} \leq \log_2\brac{2/\alpha}$.

The above argument shows that the strategy $(0^n,g,h)$, where $0^n$ denotes the constant all zeros function, wins at least $n-\log_2\brac{2/\alpha} \geq 3n/4$ (for small enough $\epsilon$) coordinates when the input lies in $E$, which happens with probability at least $\alpha/2$.

Next, we consider the 2-player game $\tilde{\mc G}$, which naturally arises when we think of the first player's answer to be fixed to 0 in $\mc G$.
It is defined formally as follows:
\begin{enumerate}
	\item The inputs $(\tilde{p}, \tilde{q})$ are distributed uniformly over the set $\tilde{\mc S} = \set{(0,0), (0,1), (1,0)}$.
	\item The players give answers $\tilde{b}, \tilde{c} \in \set{0,1}$.
	\item The predicate $\tilde{V}$ is given by $\tilde{V}\brac{(\tilde{p}, \tilde{q}), (\tilde{b}, \tilde{c})} = V\brac{(1-\tilde{p}-\tilde{q}, \tilde{p}, \tilde{q}), (0, \tilde{b}, \tilde{c})}$. 
\end{enumerate}
The following are easy to verify:
\begin{enumerate}[label=(\alph*)]
	\item $\val(\tilde{\mc G}) \leq \val(\mc G) \leq 2/3$.
	\item With probability at least $\alpha/2$, the strategy $(g,h)$ wins at least $3n/4$ coordinates in $\tilde{\mc G}^{\otimes n}$.
\end{enumerate}

By standard concentration bounds on 2-player parallel repetition (see Proposition~\ref{prop:2p_concentration_bd}), we know that for any strategy, the probability of winning at least $3n/4$ coordinates in $\tilde{\mc G}^{\otimes n}$ is at most $2^{-\delta n}$, where $\delta>0$ is an absolute constant.
Hence, it must hold that $\alpha/2 \leq 2^{-\delta n}$, which is a contradiction for small enough $\epsilon>0$.
\qed

We remark that the strong result of Proposition~\ref{prop:2p_concentration_bd} is not really needed here, and we only use it for brevity.
A case analysis on the possible predicates makes this more evident: each of the resulting games turns out to be very easy to analyze once the answer for player 1 is fixed to all zeros.

\subsection{Case 2 Analysis}\label{sec:three_point_case2}

Suppose that each row and column in each table has at least one 1.
Then, in each table, at least one of the two diagonals (indexed by the $\set{(0,0), (1,1)}$ or $\set{(1,0), (0,1)}$ entries) has all 1s.
By possibly using the transformations $a \mapsto 1-a$ and $b \mapsto 1-b$, we can assume that the diagonal indexed by the $\set{(0,0), (1,1)}$ entries has all 1s, in both $V_1$ and $V_2$.
More precisely, we assume that $V_1(b,c)\geq \ind\sqbrac{b=c}$ and $V_2(a,c)\geq \ind\sqbrac{a=c}$, where we use $\ind$ to denote the indicator operator.
Further, we assume that the predicate $V$ is such that $\val(\mc G)<1$.

Now, if $V_3(0,0)=1$ then $(a,b,c)=(0,0,0)$ is a strategy that wins on all points in $\mc S$, and $\val(\mc G) = 1$.
Similarly, if $V_3(1,1)=1$ then $(a,b,c)=(1,1,1)$ wins on all points in $\mc S$.
Hence, it must hold that $V_3(0,0)=V_3(1,1)=0$.
Since we assumed each row in $V_3$ has at least one 1, it holds that $V_3(a,b) = \ind\sqbrac{a\not=b}$.

At this point, it is not hard to show that all the remaining $*$ entries must be zero:
\begin{itemize}
	\item If $V_1(0,1) = 1$, then $(a,b,c) = (1,0,1)$ wins on all points in $\mc S$.
	\item If $V_1(1,0) = 1$, then $(a,b,c) = (0,1,0)$ wins on all points in $\mc S$.
	\item If $V_2(0,1) = 1$, then $(a,b,c) = (0,1,1)$ wins on all points in $\mc S$.
	\item If $V_2(1,0) = 1$, then $(a,b,c) = (1,0,0)$ wins on all points in $\mc S$.
\end{itemize}

Hence, the three predicates are given by $V_1(b,c) = \ind\sqbrac{b=c}$, $V_2(a,c) = \ind\sqbrac{a=c}$ and $V_3(a,b) = \ind\sqbrac{a\not=b}$.

Applying the transformation $c\mapsto 1-c$, the predicate $V$ is given by $V_1(b,c) = \ind\sqbrac{b\not=c}$, $V_2(a,c) = \ind\sqbrac{a\not=c}$ and $V_3(a,b) = \ind\sqbrac{a\not=b}$.
This is exactly the anti-correlation game, an exponential parallel repetition decay bound for which is proven in Section~\ref{sec:anti_corr_game} (see Theorem~\ref{thm:anti_corr_par_rep}).
Note that in Section~\ref{sec:anti_corr_game}, we think of the input as uniform over $\set{(0,1,1), (1,0,1), (1,1,0)}$, which is easily seen to be equivalent.
\qed
\section{Three Player Games over Binary Alphabet} \label{sec:three_player_binary}

\subsection{The Main Theorem}

In this section, we prove the following main result:

\begin{theorem}\label{thm:main_thm}
	Let $\mc G = (\mc X\times \mc Y\times \mc Z, \mc A\times \mc B\times \mc C, Q, V)$ be a 3-player game with $\mc X = \mc Y = \mc Z = \mc A = \mc B= \mc C= \set{0,1}$, and such that $\val(\mc G)<1$.
	Then, there exists a constant $c = c(\mc G)>0$, such that for every $n\in \N$, it holds that $\val(\mc G^{\otimes n})\leq n^{-c}$. 
\end{theorem}

\begin{proof}
By Lemma~\ref{lemma:wlog_unif_dist}, it suffices to only consider the case when the distribution $Q$ is the uniform distribution over its support $\mc S\subseteq \set{0,1}^3$. Also notice that we only need to analyze $\mc S$ up to symmetry among the 3 players, and up to symmetry of the inputs $0$ and $1$ (that is, up to symmetries of the cube $\set{0,1}^3$).

When $\mc G$ is connected, \cref{thm:conn_par_rep} provides an inverse exponential bound $\val(\mc G^{\otimes n})= 2^{-\Omega(n)}$. Therefore we only need to consider the case when $\mc G$ is not connected, or equivalently, the graph $\mc H_{\mc G}$ (see \cref{defn:conn_graph}) is not connected. Notice that the graph $\mc H_{\mc G}$  is the subgraph of the cubical graph $\{0,1\}^3$ induced by $\mc S$. Since $\mc H_{\mc G}$ is not connected, there must be a smallest connected component $\mc S'\subseteq\mc S$ in $\mc H_{\mc G}$ of size $1$ or $2$.

If $|\mc S'|=2$, by symmetry assume that $S'=\{(1,1,0),(1,1,1)\}$. Then $S\setminus S'$ is contained in $\{(0,0,0),(0,0,1)\}$, which implies that the input $(x,y,z)$ always satisfies $x=y$. This means that the game $\mc G$ is essentially a two-player game, where an inverse exponential decay bound is known by \cref{thm:2p_par_rep}.

If $|\mc S'|=1$, by symmetry assume that $S'=\{(1,1,1)\}$. Then $S\setminus S'$ is contained in $\{(0,0,0),(1,0,0),(0,1,0),(0,0,1)\}$, and we perform a case analysis in below:

\begin{enumerate}
    \item $|S|\leq 2$. Then $\mc G$ always degenerates to a two-player game, similar to the case of $|\mc S'|=2$ above.
    \item $|S|=3$. To avoid degeneracy it must hold that $(0,0,0)\notin S$, so by symmetry we only consider $S=\{(1,0,0),(0,1,0),(1,1,1)\}$, or equivalently, $S=\{(1,0,0),(0,1,0),(0,0,1)\}$. The specific game of interest, the anti-correlation game, was studied in \cref{sec:anti_corr_game}. The general game with binary outputs was analyzed in \cref{sec:hm_wt_one}, where we proved an inverse exponential decay bound (see Theorem~\ref{thm:hm_wt_one_par_rep}).
    \item $|S|=4$ and $(0,0,0)\in S$. By symmetry we consider $S=\{(0,0,0),(1,0,0),(0,1,0),(1,1,1)\}$. We proved an inverse polynomial decay bound for this four-point AND distribution in \cref{sec:four_point_game} (see Theorem~\ref{thm:four_point_par_rep}).
    \item $|S|=4$ and $(0,0,0)\notin S$, that is, $S=\{(1,0,0),(0,1,0),(0,0,1),(1,1,1)\}$. This is equivalent to the support of the GHZ game, and and an inverse polynomial decay bound is known (see \cref{thm:ghz_par_rep})
    \item $|S|=5$, that is, $S=\{(0,0,0),(1,0,0),(0,1,0),(0,0,1),(1,1,1)\}$. In particular, the game $\mc G$ is playerwise connected (see \cref{defn:conn_player_graph}), and we proved inverse polynomial decay bounds for all playerwise connected games in Section~\ref{sec:playerconn_games} (see Theorem~\ref{thm:playerwise_conn_par_rep}). \qedhere
\end{enumerate}
\end{proof}


\subsection{A General Game on Hamming Weight One Input}
    We observe that Theorem~\ref{thm:main_thm} works for arbitrary answer lengths in all cases except when the support $\mc S$ has $\norm{\mc S}=3$ with all disjoint points, for example, $\mc S = \set{(1,0,0), (0,1,0), (0,0,1)}$.
	
	Next, we describe a very simple family of 3-player games $\set{\mc G_k}_{k\in \N}$, such that proving a bound on the value of parallel repetition for games in this family will extend Theorem~\ref{thm:main_thm} to all games with $\mc X = \mc Y = \mc Z = \set{0,1}$, and arbitrary answer sets $\mc A, \mc B, \mc C$.

\begin{definition}\label{defn:3_pt_general game}
	For every $k\in \N$, we define a 3-player game $\mc G_k = (\mc X\times \mc Y\times \mc Z, \mc A_k\times \mc B_k\times \mc C_k, Q, V_k)$, with $\mc X = \mc Y = \mc Z = \set{0,1}$, and $Q$ the uniform distribution over $\mc S = \set{(1,0,0), (0,1,0), (0,0,1)}$, as follows:
	\begin{enumerate}[label=(\alph*)]
		\item $\mc A_k = \mc B_k = \set{0,1}^k$, and $\mc C_k = [k]$.
		\item For all $(x,y,z)\in \mc S$, and $(a,b,c)\in \mc A\times \mc B\times \mc C$,
			\[V_k\brac{(x,y,z),(a,b,c)} = \begin{cases} (a_i\land b_i)_{i\in [k]} = 0^k ,& \text{if }(x,y,z) = (0,0,1) \\ a_c = 1, &\text{if }(x,y,z) = (0,1,0) \\ b_c = 1, & \text{if }(x,y,z) = (1,0,0) \end{cases}.\]
	\end{enumerate}
	It is an easy check that $\val(\mc G_k) = 2/3$.
	For every $n\in \N$, we define $\rho_k(n) = \val(\mc G_k^{\otimes n})$. 
\end{definition}

\begin{proposition}\label{prop:redn_to_3pt_general}
	Let $\mc G = (\mc X\times \mc Y\times \mc Z, \mc A\times \mc B\times \mc C, Q, V)$ be a 3-player game with $\mc X = \mc Y = \mc Z = \set{0,1}$, and $Q$ the uniform distribution over $\mc S = \set{(1,0,0), (0,1, 0), (0,0,1)}$, and such that $\val(\mc G)<1$.
	Then, for every $n\in \N$, it holds that $\val(\mc G^{\otimes n}) \leq \rho_k(n)$, where $k = \max\set{\norm{\mc A}, \norm{\mc B}, \norm{\mc C} }$.
\end{proposition}
\begin{proof}
	Let $\mc G$ be a game as specified.
	Without loss of generality, we assume that $\mc A = \mc B = \mc C = [k]$.
	By Lemma~\ref{lemma:pred_for_deter_inp}, it also suffices to assume that the predicate $V$ only depends on the answers of the two players that get input 0.
	We observe that since $\val(\mc G)<1$, for every $a,b\in [k]$, if it holds that $V\brac{(0,0,1),(a,b,*)} = 1$, then for every $c\in [k]$, it holds that $V\brac{(0,1,0),(a,*,c)} \land V\brac{(1,0,0),(*,b,c)} = 0$
	
	We consider the game $\mc G^{\otimes n} = (\brac{\mc X\times \mc Y\times \mc Z}^{\otimes n}, \brac{\mc A\times \mc B\times \mc C}^{\otimes n}, P, V^{\otimes n})$.
	Let $f,g,h:\set{0,1}^n\to [k]^n$ be optimal strategies for the game $\mc G^{\otimes n}$.
	
	We define strategies $f_k:\set{0,1}^n\to \mc A_k^{\otimes n}, \ g_k:\set{0,1}^n\to \mc B_k^{\otimes n}, \ h_k:\set{0,1}^n\to \mc C_k^{\otimes n}$, for the game $\mc G_k^{\otimes n}$, as follows: 
	For every $x,y,z \in \set{0,1}^n$, and every $i\in [n]$, we define
	\[f_k(x)_i = \brac{V\brac{(0,1,0),(f(x)_i, *,  c)}}_{c\in [k]},\]
	\[g_k(y)_i = \brac{V\brac{(1,0,0),(*, g(y)_i,  c)}}_{c\in [k]},\]
	\[h_k(z)_i = h(z)_i.\]
	It is clear (from the above observation about $\mc G$) that if the strategies $f, g, h$ win the game $\mc G^{\otimes n}$ on an input $(x,y,z)$, then the strategies $f_k, g_k, h_k$ also win the game $\mc G_k^{\otimes n}$ on input $(x,y,z)$.
	This shows that $\rho_k(n) = \val(\mc G_k^{\otimes n}) \geq \val(\mc G^{\otimes n})$.	 
\end{proof}

\appendix
\section{Preliminary Lemmas}

\subsection{Probability Facts}\label{appendix:probab_facts}

\begin{proof}[Proof of Fact~\ref{fact:norm_condition}]
	\begin{align*}
		\lonorm{P_{X|W}}{Q_{X|W}} &= \sum_{x\in W} \norm{\frac{P_X(x)}{P_X(W)}-\frac{Q_X(x)}{Q_X(W)}}
		\\&\leq  \sum_{x\in W} \norm{\frac{P_X(x)}{P_X(W)}-\frac{P_X(x)}{Q_X(W)}} + \sum_{x\in W} \norm{\frac{P_X(x)}{Q_X(W)}-\frac{Q_X(x)}{Q_X(W)}}
		\\&\leq P_X(W)\cdot\frac{\norm{P_X(W)-Q_X(W)}}{P_X(W)\cdot Q_X(W)} + \frac{1}{Q_X(W)}\cdot \lonorm{P_X}{Q_X}
		\\&\leq \frac{2}{Q_X(W)}\cdot\lonorm{P_X}{Q_X}. \qedhere
	\end{align*}
\end{proof}

\begin{proof}[Proof of Fact~\ref{fact:conditional-markov}]
	\[ P\brac{X\not\in \mc T \st E} = \sum_{x\not\in \mc T} P\brac{X=x\st E} =  \sum_{x\not\in \mc T} \frac{P(E\st x)}{P(E)} \cdot P(X=x) < \sum_{x\not\in \mc T} \alpha \cdot P(X=x) \leq \alpha. \qedhere\]
\end{proof}

\subsection{Multiplayer Game Results}\label{appendix:multiplayer_game_results}

\begin{proof}[Proof of Lemma~\ref{lemma:wlog_unif_dist}]
	It is easy to see that $v(1) = \val(\tilde{\mc G}) < 1$, since any strategy that wins with probability 1 with respect to $U$ also wins with probability 1 with respect to $Q$, since they have the same support $\mc S$.
	
	Observe that we can write $Q = \gamma U + (1-\gamma)Q'$, for some constant $\gamma \in (0,1]$, and for some distribution $Q'$ over $\mc S$. 
	Equivalently, drawing a sample from $Q$ can be thought of as the following two step process: 
	First, we toss a coin that lands heads with probability $\gamma$ and tails with probability $1-\gamma$. 
	If the coin lands heads, we draw a sample from $U$, and else we draw a sample from $Q'$. 
	In a similar manner, drawing $n$ times independently from $Q$ can be thought of as first tossing $n$ such coins, and then choosing $n$ samples based on the values of the coin tosses.
	Let $Z = (Z_1,\dots, Z_n) \in \set{0,1}^n$ be a random variable denoting the values of these coins.
	Then, for each $i\in [n]$, independently, $Z_i$ is 1 with probability $\gamma$ and 0 with probability $1-\gamma$.
	
	We wish to bound the value of the game $\mc G^{\otimes n}$. For this, we can assume that each of the players is also given $Z$ as input, since this can only increase the game's value.
	Now, consider any fixed value $z\in \set{0,1}^n$, and let $m = \abs{z}$ be the number of coordinates of $z$ that are ones.
	Observe that conditioned on the event $Z=z$, the value of the game is at most $v(m)$.
	This holds because the players can simply \emph{embed} a copy of the game $\tilde{\mc G}^{\otimes m}$ in the $m$ coordinates corresponding to the ones in $z$, and use shared randomness to sample the remaining $n-m$ coordinates from $Q'$ independently.
	
	Hence, we get that for a small enough constant $\beta>0$,
	\begin{align*}
		\val(\mc G^{\otimes n}) & \leq \sum_{m=0}^n\Prob{\abs{Z} = m}\cdot v(m)
		\\& \leq \Prob{\abs{Z}\leq \frac{\gamma n}{2}}\cdot 1 + 1\cdot v\brac{\left \lfloor\frac{\gamma n}{2}\right\rfloor} &  (v \text{ is non-increasing})
		\\& \leq e^{-\gamma n/8} + v\brac{\left \lfloor\frac{\gamma n}{2}\right\rfloor} & \brac{\text{Fact } \ref{fact:chernoff}}
		\\& \leq 2v(\lfloor\beta n\rfloor)  & (v(n)\geq v(1)^n,\ v(1)=\val(\tilde{\mc G})<1).
	\end{align*}
\end{proof}

\begin{proof}[Proof of Lemma~\ref{lemma:pred_for_deter_inp}]
	For every $x\in \mc X, a\in \mc A$, we define
	\[ V'(x, a) = \begin{cases} \max_{b\in \mc A: b^{-j} = a^{-j}}\set{V(y, b)}, & x=y \\V(x, a), & o/w \\  \end{cases}.\]
	The first property follows directly from the definition, and the second property simply follows from the fact $V(x, a) \leq V'(x, a)$ for every $x\in \mc X, a\in \mc A$.
	
	For the third property, it suffices to observe that in the game $\mc G$, when the $j\ts{th}$ player gets input $y^j$, they know the entire input $y$, and hence also the answers $a^{-j}$ of the remaining players.
	Hence, they are able to answer an optimal $a^j$ such that $V(x,a) = V'(x,a)$. 
\end{proof}

\begin{proof}[Proof of Lemma~\ref{lemma:prod_set_hard_coor}]
	Fix any $n\in \N$, and let $\rho = \rho(n)$. 
	Consider the repeated game $\mc G^{\otimes n} = (\mc X^{\otimes n}, \mc A^{\otimes n}, P, V^{\otimes n})$, and let $f = f^1\times\dots\times f^k: \mc X^{\otimes n}\to\mc A^{\otimes n}$ be a product strategy for the players that achieves the game value $\val(\mc G^{\otimes n})$.
	Let $X$ be the random variable denoting the inputs to the players in the game $\mc G^{\otimes n}$, and let $A = f(X)$ be the random variable denoting the players' answers.
	
	We define a sequence of random variables $J_1,\dots, J_n \in [n]$, and $Z_1,\dots, Z_n\in \mc X\times \mc A$ as follows.
	For each $i\in [n]$, we define $J_{i} \in [n]\setminus\set{J_1,\dots, J_{i-1}}$ to be a coordinate with the lowest winning probability, conditioned on the value of $(Z_1,\dots, Z_{i-1})$. Then, we define $Z_i = (X_{J_i}, A_{J_i})$.
	For each $i\in [n]$, let $W_i$ be the event that the players win in coordinate $J_i$.

	Let $m = \min\set{\left\lfloor\frac{\log_2\brac{1/\rho}}{\log_2\brac{8\norm{\mc X}\norm{\mc A}}}\right\rfloor, n}$. 
	Assume $m>0$, as else the desired result holds trivially for a small enough constant $c>0$.
	We claim that for each integer $\ell\in [m]$, it holds that ${P\brac{W_1\land\dots\land W_\ell} \leq \brac{1-\epsilon/2}^{\ell}}$.
	In that case, there exists a constant $c>0$ such that
	\[\val(\mc G^{\otimes n})  = P\brac{W_1\land\dots\land W_n} \leq P\brac{W_1\land\dots\land W_m}\leq  \brac{1-\epsilon/2}^m \leq \rho^c. \]
	
	We prove the claim by induction on $\ell$.
	For each $\ell\in[m]$, let $W_{[\ell]} = \cap_{i\in [\ell]}W_i$, and let $Z_{[\ell]} = (Z_1,\dots, Z_\ell)$. 
	The base case $\ell=1$ follows by applying the lemma hypothesis to the event $E = \mc X^{\otimes n}$. 
	For the inductive step, let $\ell\in[m-1]$ be such that $P\brac{W_{[\ell]}} \leq \brac{1-\epsilon/2}^{\ell}$.
	Further, we assume that $P\brac{W_{[\ell]}}  \geq \brac{1/2}^{\ell+1}$, or else the inductive step holds trivially in this case.
	Observe that the event $W_{[\ell]}$ depends deterministically on the value of $Z_{[\ell]}$.
	Let $\mc T$ be the set of all such tuples $z_{[\ell]}$ of questions and answers that win on all coordinates in $[\ell]$, and $\mc T'\subseteq \mc T$ consist of those $z_{[\ell]}\in \mc T$ such that $P\brac{Z_{[\ell]} = z_{[\ell]}} \geq \rho$. Then,
	{\small
	\begin{align*}
		P\brac{W_{\ell+1} | W_{[\ell]}} &= \sum_{z_{[\ell]}\in \mc T} P\brac{W_{\ell+1}|Z_{[\ell]}=z_{[\ell]}}\cdot\frac{P\brac{Z_{[\ell]}=z_{[\ell]}}}{P\brac{W_{[\ell]}}}
		\\&\leq \sum_{z_{[\ell]}\in \mc T'} (1-\epsilon)\cdot\frac{P\brac{Z_{[\ell]}=z_{[\ell]}}}{P\brac{W_{[\ell]}}} + \sum_{z_{[\ell]}\in \mc T\setminus \mc T'} 1 \cdot\frac{P\brac{Z_{[\ell]}=z_{[\ell]}}}{P\brac{W_{[\ell]}}} & (\text{by lemma hypothesis})
		\\&= (1-\epsilon) + \epsilon\cdot\sum_{z_{[\ell]}\in \mc T\setminus \mc T'} \frac{P\brac{Z_{[\ell]}=z_{[\ell]}}}{P\brac{W_{[\ell]}}} & \brac{\text{as }\sum_{z_{[\ell]}\in \mc T} \frac{P\brac{Z_{[\ell]}=z_{[\ell]}}}{P\brac{W_{[\ell]}}}=1}
		\\&\leq (1-\epsilon) + \epsilon\cdot \frac{\rho \norm{\mc T\setminus \mc T'}}{P\brac{W_{[\ell]}}}
		\\&\leq (1-\epsilon) + \epsilon\cdot \frac{\rho \norm{\mc X}^\ell\norm{\mc A}^\ell}{\brac{1/2}^{\ell+1}} &\brac{\text{as }\norm{\mc T'\setminus \mc T} \leq \norm{\mc T} \leq \norm{\mc X}^\ell\norm{\mc A}^\ell}
		\\&\leq 1-\epsilon/2 & \brac{\text{by choice of }m}.
	\end{align*}}
	
Hence, $P\brac{W_{[\ell+1]}} = P\brac{W_{[\ell]}}\cdot P\brac{W_{\ell+1}|W_{[\ell]}} \leq \brac{1-\epsilon/2}^{\ell}\cdot \brac{1-\epsilon/2} = \brac{1-\epsilon/2}^{\ell+1}.$
\end{proof}

\section{2-player Parallel Repetition}\label{appendix:2player_par_rep_results}

In this section, we state some useful results from 2-player parallel repetition.

\begin{lemma}(Lemma~\ref{lemma:pdt_distr_cond} restated)
	Let $P_{V} = P_{V_1}\times\dots\times P_{V_n}$ be a product distribution over a set $\mc V^{\otimes n}$, and $W$ an event.
	Then,
	\[ \frac{1}{n}\sum_{i=1}^n \lonorm{P_{V_i|W}}{P_{V_i}} \lesssim \sqrt{\frac{1}{n}\log_2{\frac{1}{P_V(W)}}}.\]
\end{lemma}
\begin{proof}[Proof Sketch]
For two distributions $P$ and $Q$ over a set $\Omega$, the relative entropy (also known as the KL divergence) between $P$ and $Q$ is defined as $\relent{P}{Q} = \sum_{\omega\in\Omega} P(\omega)\log_2\frac{P(\omega)}{Q(\omega)}$, with the convention that $0\cdot \log_2{0}=0\cdot \log_2\frac{0}{0} = 0$.

Observe that 
\[\relent{P_{V|W}}{P_V}  = \sum_{v\in \mc V^{\otimes n}}P_V(v|W)\log_2{\frac{P_V(v|W)}{P_V(v)}} \leq \log_2{\frac{1}{P_V(W)}}.\]
Also,
\[\frac{1}{n}\cdot \relent{P_{V|W}}{P_V}\geq \frac{1}{n}\sum_{i=1}^n \relent{P_{V_i|W}}{P_{V_i}} \gtrsim \frac{1}{n}\sum_{i=1}^n \lonorm{P_{V_i|W}}{P_{V_i}}^2 \geq \brac{\frac{1}{n}\sum_{i=1}^n \lonorm{P_{V_i|W}}{P_{V_i}}}^2.\]
For the first inequality, we used that relative entropy satisfies a super-additive property when the second distribution is a product distribution.
The second inequality is an application of Pinsker's inequality, which says that the $L^1$-distance between two distributions is at most $\sqrt{2\ln{2}}$ times the square root of relative entropy.
\end{proof}

The next result is essentially the core of all known (information-theoretic) proofs of 2-player parallel repetition.
For the sake of completeness, we give a rough proof sketch.
The reader is referred to the proofs in $\cite{Raz98, Hol09}$ for details.

\begin{proposition}\cite{Raz98}\label{prop:2_player_hard_set}
	Let $\mc G = \brac{\mc X\times \mc Y,\mc A\times \mc B,Q, V}$ be a 2-player game with $\val(\mc G)< 1$.
	Let $n\in \N$ be large enough, and consider the game $\mc G^{\otimes n} = \brac{\brac{\mc X\times \mc Y}^{\otimes n},\ \brac{\mc A\times \mc B}^{\otimes n},\ P, V}$.
	Suppose $E = E^1\times E^2\subseteq \mc X^{\otimes n}\times \mc Y^{\otimes n}$ is a product event with $P(E)\geq 2^{-n^{\delta}}$, for some constant $\delta\in(0,1)$.
	
	Let $f:\mc X^{\otimes n}\to\mc A^{\otimes n}$, and $g:\mc Y^{\otimes n}\to\mc B^{\otimes n}$ be any strategies for the 2-players.
	
	For any $S \subseteq [n]$, let $W_S$ be the event that the players win all the coordinates indexed by $S$.
	Then, for every constant $\epsilon\in(0,1)$, it holds that $\E_S\sqbrac{P\brac{W_S\st E}}\leq 2^{-n^{\epsilon/2}}$, where the expectation is over uniformly random $S\subseteq [n]$ of size $\lfloor n^{\epsilon}\rfloor$.
\end{proposition}

\begin{proof}[Proof Sketch]
	Let $c\in (0,1)$ be such that $\val(\mc G) = 1-c$, and let $\eta = 1-\frac{c}{2}$.
	The main result that goes into the proof is Theorem 1.2 in \cite{Raz98}, which says the following: When conditioning on a large product event (with measure $2^{-o(n)}$), a random coordinate has winning probability at most $\eta$ in expectation.
	We note that $\cite{Raz98}$ only talks about the existence of a hard coordinate, but the proof also shows that a random coordinate is hard in expectation.	
	
	We show that $\E_S\sqbrac{P\brac{W_S\st E}}\leq k\cdot \eta^{k}$, where the expectation is over uniformly random $S\subseteq [n]$ of size $k \leq \lfloor n^{\epsilon}\rfloor$, by induction on $k$.
	
	The case $k=1$ follows directly from the above statement from $\cite{Raz98}$.
	
	For the case $k+1$, 	we first fix an arbitrary set $S = \set{i_1,\dots, i_k}$ and observe the following:
	\begin{itemize}
		\item If $P(W_S|E)\leq \eta^{k+1} $, then $\E_{i\in [n]\setminus S} \sqbrac{P(W_{S\cup\set{i}}|E)}\leq \eta^{k+1}$ also.
		\item If $P(W_S|E)\geq \eta^{k+1} \geq 2^{-(k+1)}$, then by the above statement from \cite{Raz98} (note that $k=o(n)$), we get $\E_{i\in [n]\setminus S}\sqbrac{P\brac{W_{\set{i}}\st E, W_S}}\leq \eta$.
		We remark that while conditioning on $E, W_S$ is not necessarily a product event, we can rather condition on $E$ and \emph{typical} questions and answers of the players in coordinate $S$, as these deterministically determine $W_S$ (this uses a similar argument as the proof of Lemma~\ref{lemma:prod_set_hard_coor_exists}).
	\end{itemize}
	This gives us $\E_{S\cup\set{i}}\sqbrac{P\brac{W_{S\cup\set{i}}\st E}}\leq \eta^{k+1} + \eta \cdot \E_S\sqbrac{P\brac{W_S\st E}} \leq (k+1)\cdot \eta^{k+1}$ by induction.
\end{proof}

We will also state a concentration bound on the parallel repetition of 2-player games.

\begin{proposition}\label{prop:2p_concentration_bd} \cite{Rao11}
	Let $\mc G = \brac{\mc X\times \mc Y,\mc A\times \mc B,Q, V}$ be a 2-player game with $\val(\mc G)\leq 1-\epsilon$, for some $\epsilon>0$.
	Then, for any constant $0<\delta<\epsilon$, and any player strategies, the probability that the players win at least $(1-\epsilon+\delta)$ fraction of the coordinates in the game $\mc G^{\otimes n}$ is at most $2^{-\Omega_{\epsilon}\brac{\frac{\delta^3 n}{\log_2\norm{\mc A}\norm{\mc B}}}}$, where the constant in the exponent may depend on $\epsilon$.
    
\end{proposition}


\section{Random 3-CNF Example}\label{app:rand_3CNF}
Here, we prove the claims in Example~\ref{ex:rand_3CNF}.
\begin{enumerate}
	\item Suppose $m = \omega(d)$. Fix any constant $\epsilon >0$. Consider any fixed strategy for the players (which consists of three functions $[d]\to \set{0,1}$). The number of possible clauses on which this strategy loses is exactly $d^3$ (out of the possible $8d^3$ clauses). Hence, by Fact~\ref{fact:chernoff}, the probability (over the random 3-CNF $\varphi$) that this strategy wins the game $\mc G$ with value at least $7/8 + \eps$ is at most $2^{-\Omega(\eps^2 m)}$. By a union bound over the possible player strategies, we get \[\Pr_{\varphi}\sqbrac{\val\brac{\mc G} \geq \frac{7}{8}+\eps} \leq (2^d)^3 \cdot 2^{-\Omega(\eps^2 m)}= o(1).\] In particular, with high probability, it holds that $\val(\mc G)<1$.
	
	By a similar argument, with probability $1-o(1)$, any fixed strategy has value at least $7/8-\eps$, and in particular $\val\brac{\mc G}\geq 7/8-\eps$.
	\item Consider the graph $\mc H$ with vertex set $[d]^3$, with edges between vertices $v$ and $v'$ if they differ in exactly one coordinate. The graph $\mc H_\mc G$ (based on a random $\varphi$) is then the induced subgraph of this graph obtained by choosing $m$ vertices uniformly and independently.
		\begin{enumerate}
			\item Suppose $m = \omega(d^2\log d)$. For each $x\in[d]$, construct a bipartite graph $\mc K_x$ on the vertex sets $[d]$ and $[d]$, such that for each chosen vertex $(x,y,z)$ in $\mc H_{\mc G}$, there is an edge $(y,z)$ in $\mc K_x$. Let $k_x$ be the number of chosen vertices $(x,y,z)$, which is at least $\frac{m}{2d}=\omega(d\log d)$ with probability $1-o(d^{-1})$ by Fact~\ref{fact:chernoff}. On the other hand, when $k_x$ is fixed, the $k_x$ edges are uniformly randomly chosen in $\mc K_x$, therefore by Erd\'{o}s-Renyi the graph $\mc K_x$ is connected with probability $1-o(d^{-1})$ (see e.g. \cite{Pal64} ).

			That means with probability $1-o(1)$, the induced subgraphs of $\mc H_{\mc G}$ on vertices $\{(x,y,z)\}$ are connected for all fixed $x\in[d]$ (the edges of $\mc K_x$ are connected since the vertices are connected). From the proof in item~\ref{item:3cnf_playerconn_d1.5} below we also know that $\mc H_{\mc G}^1$ is connected with probability $1-o(1)$, thus overall $\mc H_{\mc G}$ is connected with probability $1-o(1)$.
			\item Suppose $m = o(d^2)$. The probability that the first chosen vertex is isolated in $\mc H_{\mc G}$ equals $\brac{1-\frac{3(d-1)}{d^3}}^{m-1} \geq 1-\frac{3(m-1)(d-1)}{d^3} = 1-o(1)$. In this case, the graph $\mc H_\mc G$ must not be connected.
		\end{enumerate}
	\item \begin{enumerate}
		\item \label{item:3cnf_playerconn_d1.5} Suppose $m = \omega(d^{1.5}\sqrt{\log d})$. We make the following deductions:
		\begin{itemize}
		    \item First, we can assume that for each player $j$, the graph $\mc H_{\mc G}^j$ contains all vertices in $[d]$, since the probability of it not happening is at most $3d\cdot (1-1/d)^m=o(1)$.
		    \item Therefore, it suffices to prove the claim when a subset of exactly $m/2$ vertices in $\mc H$ are chosen, for the more vertices of $\mc H_{\mc G}$ there are, the more likely the graphs $\mc H_{\mc G}^j$ are connected when their vertex sets are fixed to $[d]$.
		    \item Furthermore, it suffices to prove the claim when each vertex in $\mc H$ is chosen independently with probability $p=\frac{m}{4d^3}=\omega(d^{-1.5}\sqrt{\log d})$, as when less than $m/2$ vertices are chosen (which happens with probability $1-o(1)$ by Fact~\ref{fact:chernoff}), we can always randomly choose a superset of size $m/2$ instead. 
		    \item Now with each vertex independently chosen, for each player $j$ the graph $\mc H_{\mc G}^j$ is now a random intersection graph $G(d,d^2,p)$ as defined in \cite{Sin96}, and there it was proved to be connected with probability $1-o(1)$. A union bound on the players shows that $\mc G$ is playerwise connected with probability $1-o(1)$.
		\end{itemize}
		
		\item Suppose $m = o(d^{1.5})$, and let $(x_0,y_0,z_0)$ denote the first chosen vertex in $\mc H_\mc G$. We show that with high probability, the vertex $x_0$ is isolated in $\mc H_\mc G^1$, and hence the game $\mc G$ is not playerwise connected.
		
			Observe that $x_0$ is not isolated in $\mc H_{\mc G}^1$ if and only if one of the following holds:
			\begin{itemize}
				\item There exists $x\in [d]\setminus\set{x_0}$ such that $(x,y_0, z_0)$ is chosen. This occurs with probability $1-\brac{1-\frac{d-1}{d^3}}^{m-1}\leq \frac{(d-1)(m-1)}{d^3} = o(1)$.
				\item There exists $x\in [d]\setminus\set{x_0}$, and $(y,z)\in [d]^2\setminus\set{(y_0,z_0)}$ such that both $(x,y,z)$ and $(x_0, y, z)$ are chosen. This occurs with probability at most 
				\begin{align*}
					&\brac{d-1}\cdot \brac{d^2-1}\cdot \brac{1-\brac{1-\frac{1}{d^3}}^{m-1} - \brac{1-\frac{1}{d^3}}^{m-1} + \brac{1-\frac{2}{d^3}}^{m-1} }
					\\&\leq  d^3\cdot \brac{1-\brac{1-\frac{m-1}{d^3}} - \brac{1-\frac{m-1}{d^3}} + \brac{1-\frac{2(m-1)}{d^3}+\frac{2(m-1)^2}{d^6}} }
					\\&= d^3\cdot \frac{2(m-1)^2}{d^6} \leq \frac{2m^2}{d^3} = o(1). 
				\end{align*}		
				\qed		
			\end{itemize}
	\end{enumerate}
\end{enumerate}

\bibliographystyle{alpha}
\bibliography{ref.bib}

\end{document}